%% file: thesis.tex
\documentclass{llncs}[11pt]

\usepackage{amsmath, amssymb}
\usepackage{xspace}
\usepackage{dsfont}
\usepackage[colorlinks=true,urlcolor=webblue,linkcolor=webgreen,filecolor=webblue,citecolor=webgreen,pdfpagemode=UseOutlines,pdfstartview=FitH,pdfpagelayout=OneColumn,bookmarks=true]{hyperref}
\usepackage{enumerate}
\usepackage{graphicx}
\usepackage{shadethm}
\usepackage{braket}
\usepackage{booktabs}
\usepackage{bbm}   

\usepackage{caption}
\usepackage{subcaption}

\usepackage[ruled,vlined]{algorithm2e}
\SetKwInput{KwInput}{Input}
\SetKwInput{KwOutput}{Output}

\usepackage{tikz}
\usepackage[numbers]{natbib}

\usepackage{mdframed}

\definecolor{webgreen}{rgb}{0,.5,0}
\definecolor{webblue}{rgb}{0,0,.5}

\usepackage{ntheorem}
\theoremstyle{plain}
\theorembodyfont{}
\theoremsymbol{}
\theoremprework{}
\theorempostwork{}
\theoremseparator{.}

\newtheorem{construction}{Construction}

\newcommand{\C}{\mathbb{C}}
\newcommand{\Z}{\mathbb{Z}}

\newcommand{\de}[1]{\left( #1 \right)}

\renewcommand{\Re}[1]{{\mathrm{Re}}\de{#1}}
\renewcommand{\Im}[1]{{\mathrm{Im}}\de{#1}}

\newcommand{\abs}[1]{\left| #1 \right|}

\renewcommand{\rho}{\varrho}

\newcommand{\POVM}{\ensuremath{\mathsf{POVM}}\xspace}

\newcommand{\one}{\mathds 1}

{\begin{framed}\begin{small}}
{\end{small}\end{framed}}

\newcommand{\SKES}{\ensuremath{\textsf{SKES}}\xspace}

\newcommand{\expref}[2]{\texorpdfstring{\hyperref[#2]{#1~\ref{#2}}}{#1~\ref{#2}}}

\newcommand{\A}{\ensuremath{\mathcal{A}}\xspace}

\renewcommand{\C}{\ensuremath{\mathcal{C}}\xspace}
\newcommand{\D}{\ensuremath{\mathcal{D}}\xspace}

\renewcommand{\H}{\ensuremath{\mathcal H}\xspace}
\newcommand{\K}{\ensuremath{\mathcal{K}}\xspace}
\newcommand{\Primes}{\ensuremath{\mathsf{Primes}}\xspace}
\newcommand{\M}{\ensuremath{\mathcal{M}}\xspace}

\newcommand{\negl}{\ensuremath{\operatorname{negl}}\xspace}
\newcommand{\supp}{\textbf{supp\,}}
\newcommand{\from}{\ensuremath{\leftarrow}}
\newcommand{\bit}{\{0,1\}}

\newcommand{\rand}{\raisebox{-1pt}{\ensuremath{\,\xleftarrow{\raisebox{-1pt}{$\scriptscriptstyle\$$}}\,}}}

\DeclareMathOperator{\Span}{span}

\newcommand{\KeyGen}{\ensuremath{\mathsf{KeyGen}}\xspace}
\newcommand{\Enc}{\ensuremath{\mathsf{Enc}}\xspace}
\newcommand{\Dec}{\ensuremath{\mathsf{Dec}}\xspace}

\newcommand{\poly}{\operatorname{poly}}
\newcommand{\algo}{\mathcal}
\newcommand{\inrand}{\raisebox{-1pt}{\ensuremath{\,\xleftarrow{\raisebox{-1pt}{$\scriptscriptstyle\$$}}\,}}}

\newcommand{\PRF}{\ensuremath{\mathsf{PRF}}\xspace}

\newcommand{\QPRF}{\ensuremath{\mathsf{QPRF}}\xspace}

\newcommand{\QPT}{\ensuremath{\mathsf{QPT}}\xspace}

\newcommand{\RelabelingGame}{\ensuremath{\mathsf{RelabelingGame}}\xspace}
\newcommand{\ClassicalRelabeling}{\ensuremath{\mathsf{ClassicalRelabeling}}\xspace}

\newcommand{\LWE}{\ensuremath{\mathsf{LWE}}\xspace}

\newcommand{\DLWE}{\ensuremath{\mathsf{DecLWE}}\xspace}
\newcommand{\Sim}{\ensuremath{\mathsf{S}}\xspace}
\newcommand{\LPN}{\ensuremath{\mathsf{LPN}}\xspace}
\newcommand{\PAC}{\ensuremath{\mathsf{PAC}}\xspace}
\newcommand{\Poly}{\ensuremath{\mathsf{P}}\xspace}
\newcommand{\NP}{\ensuremath{\mathsf{NP}}\xspace}
\newcommand{\CPTP}{\ensuremath{\mathsf{CPTP}}\xspace}
\newcommand{\PPT}{\ensuremath{\mathsf{PPT}}\xspace}
\newcommand{\QFT}{\ensuremath{\mathsf{QFT}}\xspace}
\newcommand{\RSA}{\ensuremath{\mathsf{RSA}}\xspace}
\newcommand{\ECC}{\ensuremath{\mathsf{ECC}}\xspace}
\newcommand{\DHKE}{\ensuremath{\mathsf{D\mbox{-}H}}\xspace}
\newcommand{\CPA}{\ensuremath{\mathsf{CPA}}\xspace}
\newcommand{\CCA}{\ensuremath{\mathsf{CCA1}}\xspace}
\newcommand{\CCAA}{\ensuremath{\mathsf{CCA2}}\xspace}

\newcommand{\IND}{\ensuremath{\mathsf{IND}}\xspace}
\newcommand{\INDCPA}{\ensuremath{\mathsf{IND\mbox{-}CPA}}\xspace}
\newcommand{\INDCCA}{\ensuremath{\mathsf{IND\mbox{-}CCA1}}\xspace}
\newcommand{\INDCCAA}{\ensuremath{\mathsf{IND\mbox{-}CCA2}}\xspace}

\newcommand{\INDQCPA}{\ensuremath{\mathsf{IND\mbox{-}QCPA}}\xspace}
\newcommand{\INDQCCA}{\ensuremath{\mathsf{IND\mbox{-}QCCA1}}\xspace}
\newcommand{\INDQCCAA}{\ensuremath{\mathsf{IND\mbox{-}QCCA2}}\xspace}

\newcommand{\QIND}{\ensuremath{\mathsf{QIND}}\xspace}

\newcommand{\QCPA}{\ensuremath{\mathsf{QCPA}}\xspace}
\newcommand{\QCCA}{\ensuremath{\mathsf{QCCA1}}\xspace}
\newcommand{\QCCAA}{\ensuremath{\mathsf{QCCA2}}\xspace}

\newcommand{\QINDQCCA}{\ensuremath{\mathsf{QIND\mbox{-}QCCA1}}\xspace}

\newcommand{\SEM}{\ensuremath{\mathsf{SEM}}\xspace}
\newcommand{\Samp}{\ensuremath{\mathsf{Samp}}\xspace}

\newcommand{\SEMCCA}{\ensuremath{\mathsf{SEM\mbox{-}CCA1}}\xspace}
\newcommand{\SEMCCAA}{\ensuremath{\mathsf{SEM\mbox{-}CCA2}}\xspace}

\newcommand{\SEMQCCA}{\ensuremath{\mathsf{SEM\mbox{-}QCCA1}}\xspace}

\newcommand{\INDGame}{\ensuremath{\mathsf{IndGame}}\xspace}
\newcommand{\IndGame}{\ensuremath{\mathsf{IndGame}}\xspace}
\newcommand{\SEMGame}{\ensuremath{\mathsf{SemGame}}\xspace}

\makeatletter
\renewcommand*\l@author[2]{}
\renewcommand*\l@title[2]{}
\makeatletter

\title{Quantum Learning Algorithms\\ and Post-Quantum Cryptography$^*$}
\author{Alexander M. Poremba}
\institute{QMATH, Department of Mathematical Sciences, University of Copenhagen, \\  1165 Copenhagen, Denmark.
\and Department of Physics and Astronomy, University of Heidelberg, \\ 69047 Heidelberg, Germany.}

\makeindex

\begin{document}

\maketitle

\begin{abstract}
Quantum algorithms have demonstrated promising speed-ups over classical algorithms
in the context of computational learning theory - despite the presence of noise. 
In this work,
we give an overview of recent quantum speed-ups, revisit the
Bernstein-Vazirani algorithm in a new learning problem extension over an arbitrary cyclic group and
discuss applications in cryptography, such as the Learning with Errors problem.

We turn to post-quantum cryptography and investigate attacks in which an adversary is given
quantum access to a classical encryption scheme. 
In particular, we consider new notions of security under non-adaptive quantum chosen-ciphertext attacks
and propose symmetric-key encryption schemes based on quantum-secure pseudorandom functions that fulfil our definitions.
In order to prove security, we introduce novel relabeling techniques and show
that, in an oracle model with an arbitrary advice state, no quantum algorithm making superposition queries can reliably
distinguish between the class of functions that are randomly relabeled at a small subset of the domain.

Finally, we discuss current progress in quantum computing technology, particularly with a focus on implementations of quantum algorithms on the ion-trap architecture,
and shed light on the relevance and effectiveness of common noise models adopted in computational learning theory.
\end{abstract}
\vfill
\underline{\hspace{20mm}}\\
\footnotesize{$^*$This work was carried out as part of my Master's thesis at the University of Heidelberg.\newline
Contact: \href{mailto:alexander.poremba@gmail.com}{alexander.poremba@gmail.com}}


\newpage

\ \\
\ \\
\ \\
\ \\
\ \\
\ \\
\ \\
\ \\
\ \\
\ \\
\ \\
\ \\
\ \\
\ \\
\ \\
\ \\
\ \\
\ \\
\ \\
\ \\

\begin{flushright}
\textit{
\textsc{Principal advisor:}\\
Gorjan Alagic,\\
Joint Center for Quantum Information\\
and Computer Science, \\
University of Maryland, College Park, MD\\
\ \\
\ \\
\textsc{Co-advisor:}\\
Thomas Gasenzer,\\
Kirchhoff-Institute for Physics,\\
University of Heidelberg, Germany.
}
\end{flushright}
\newpage

\tableofcontents

\input{abbreviations}

\input{intro} 
\input{cryptographic_primitives}

\input{quantum_computation} 
\input{quantum_algorithms} 
\input{quantum_fourier_transform} 
\input{bernstein_problem}
\input{blinding}

\input{blinding_Gorjan}
\input{ph_realization}
\input{conclusion}

\input{open_problems}

\section*{Acknowledgments}\label{sec:Acknowledgments}

I want to thank my supervisor Gorjan Alagic for his support and mentorship,
and for inspiring me to pursue a thesis in cryptography. In addition, I want to thank Thomas Gasenzer for making 
this thesis collaboration possible. I also want to thank everyone who offered additional support, contributed with comments or engaged
in useful discussions, especially everyone at QMATH: 
Matthias Christandl, Florian Speelman, Chris Perry, Roberto Ferrara and Christian Majenz,
as well as Benjamin Freist at Heidelberg University. Finally, I want to thank Stacey Jefferey, Maris Ozols, Harry Buhrman, Ronald de Wolf and Christian Schaffner at CWI, Amsterdam.

\newpage
\appendix
\input{appendix}

\end{document}

%% file: abbreviations.tex
\newpage
\section{List of Abbreviations}

\textit{
\begin{footnotesize}
\begin{description}
\item $\bot$ reject symbol
\item \CPA chosen-plaintext attack
\item \CPTP completely positive and trace preserving
\item \CCA non-adaptive chosen-ciphertext attack
\item \CCAA adaptive chosen-ciphertext attack
\item \DLWE Decision Learning with Errors
\item \IND indistinguishable encryptions
\item \INDCPA indistinguishable encryptions under chosen-plaintext attack
\item \INDCCA indistinguishable encryptions under non-adaptive chosen-ciphertext attack
\item \INDCCAA indistinguishable encryptions under adaptive chosen-ciphertext attack
\item \INDQCPA indistinguishable encryptions under quantum chosen-plaintext attack
\item \INDQCCA indistinguishable encryptions under non-adaptive quantum chosen-ciphertext
\item \INDQCCAA indistinguishable encryptions under adaptive quantum chosen-ciphertext attack
\item \LWE Learning with Errors
\item \LPN Learning Parity with Noise
\item \NP nondeterministic polynomial time
\item \Poly polynomial time
\item \PAC probably approximately correct
\item \PPT probabilistic polynomial time
\item \POVM positive operator valued measure
\item \PRF pseudorandom function
\item \QCPA quantum chosen-plaintext attack
\item \QCCA non-adaptive quantum chosen-ciphertext attack
\item \QCCAA adaptive quantum chosen-ciphertext attack
\item \QFT quantum Fourier transform
\item \QPRF quantum-secure pseudorandom function
\item \QPT quantum polynomial time
\item \SKES symmetric-key encryption scheme
\item \SEM semantic security
\item \SEMCCA semantic security under non-adaptive chosen-ciphertext attack
\item \SEMCCAA semantic security under adaptive chosen-ciphertext attack
\item \SEMQCCA semantic security under non-adaptive quantum chosen-ciphertext attack
\end{description}
\end{footnotesize}
}
\newpage

%% file: intro.tex
\section{Introduction}

Most of our present day communication takes place on the internet and produces enormous amounts of personal data.
Whereas traditional notions of security are concerned with electronic mail or bank transfers,
today's security needs have since expanded to many unexpected areas such as smartcards, medical devices or modern cars.
Cryptography, understood as the science of secure communication, is becoming increasingly relevant for our safety in the modern world. 
For many years, popular cryptographic protocols such as \RSA, the Diffie-Hellman key-exchange $(\DHKE)$ or ellyptic curve cryptography $(\ECC)$, have served 
greatly as building blocks towards establishing secure communication, despite lower costs and ever increasing computational power on the markets. 
In 1994, Peter Shor proposed an efficient quantum algorithm for the factoring of integers and the computation 
of discrete logarithms~\cite{Sho94}, a profound discovery that drew the attention towards the field of quantum computation 
and its potential impact on cryptography. Many of the protocols still in use today,
such as \RSA, \DHKE or \ECC, are completely broken by attackers in possession of quantum computers running Shor's algorithm.
This discovery is oftentimes regarded as the beginning of a new race towards \textit{post-quantum cryptography}, a security standard for
secure classical communication, even in the presence of quantum computers~\cite{BL17}. 
At the same time, modern quantum technology also enables entirely new forms of communication, such as quantum key distribution \cite{BB84}.
Due to both practical and economical reasons, it is nevertheless reasonable to suspect that some form of classical communication 
will continue to exist for years to come, particularly for implementations on light-weight devices.
Even as reliably fault-tolerant quantum computers have yet to be built, the cryptographic community has already
started shifting towards a new direction in which the feasibility of classical cryptography in a quantum world presents us with a paramount challenge.

A fundamental approach in cryptography is the use of hard computational problems towards
the implementation of secure communication. Consider, for example, the \RSA protocol whose security is based on the fact that
factoring large integers appears to be computationally intractable on a classical computer. Ever since the discovery of Shor's algorithm, the search towards
computational hardness in a quantum world has dominated
the cryptographic community. Since 2005, the \textit{Learning with Errors} (\LWE) problem \cite{Reg05} has gained the status of a promising cryptographic 
basis of hardness, in particular in a post-quantum setting.
The central promise of the \LWE problem lies in a reduction in which it is shown to be as hard as worst-case lattice problems \cite{Reg09},
a class of computational problems believed to be hard for more than two decades.
Consequently, it is tempting to build cryptographic constructions on the basis of the \LWE problem and 
achieve security under the assumption that worst-case lattice problems are likely to remain hard for quantum computers. 
Apart from being a candidate for security against
quantum computers, many private companies have also shown interest in variants of \LWE due to its promise for
light-weight implementation, as compared to many other promising schemes in post-quantum cryptography.
As of today, the security of lattice-based cryptography against quantum computers remains
one of the key areas of modern research in cryptography.\footnote{For an excellent review on modern cryptography in the age 
of quantum computers, we refer to a popular science article in
a 2015 issue of Quanta Magazine: \url{www.quantamagazine.org/quantum-secure-cryptography-crosses-red-line-20150908/}}\\
\noindent In a nutshell, the \LWE problem in \cite{Reg09} is as follows:

 \begin{mdframed}[backgroundcolor=green!5] 
 \textbf{Learning with Errors Problem:}\\
\textit{
Given an integer $n$ and modulus $q$,
learn a secret string $\vec s \in \mathbb{Z}_q^n$ given a set of random noisy linear equations over $\mathbb{Z}_q$ on $\vec s$.\\
For example, for $n=4$ and modulus $q=23$, each equation (with probability less than $1/2$) is of small additive error $\pm 1$, and:
\begin{align*}
11s_1 + 2s_2 + 13s_3 + 19s_4 &\approx 8  \mod 23 \\
14s_1 + 6s_2 + 19s_3 + s_4 &\approx 5  \mod 23 \\
3s_1 + 15s_2 + 4s_3 + 2s_4 &\approx 0  \mod 23 \\
4s_1 + 6s_2 + 20s_3 + 15s_4 &\approx 11  \mod 23 \\
7s_1 + 18s_2 + 8s_3 + 9s_4 &\approx 21  \mod 23 \\
8s_1 + 5s_2 + 17s_3 + 12s_4 &\approx 10  \mod 23 \\
& \vdots \\
16s_1 + s_2 + 11s_3 + 22s_4 &\approx 14  \mod 23
\end{align*}}
\end{mdframed}
If $q$ is prime, the integers modulo $q$ form a finite field under addition and multiplication, hence, given enough samples on $\vec s$, 
there exists a unique solution to the problem.
In our case, the hidden string to be determined is $s = ( 12, 0, 7, 2 )$. If not for the error, the secret string can be recovered
in polynomial time $O(n^3)$ using Gaussian elimination after observing $n$ linear independent equations, where $n$ denotes the length of the string.
Let us also note that the probability of acquiring $n$ linear independent equations on $\vec s$ after only observing $n$ sample queries is 
easily shown to be greater than a constant independent on $n$.

\noindent The difficulty in decoding noisy linear equations lies in the fact that the errors propagate during the computation, hence
amplify the uncertainty and ultimately lead to no information on the actual secret string. As the best known algorithm for the
\LWE problem runs in time $O(2^n)$ \cite{BKW03}, the problem is believed to be asymptotically intractible for classical computers.
Moreover, due to the reduction in \cite{Reg05}, any breakthrough in \LWE would also most likely imply an algorithm for lattice-based problems.

In an earlier problem, Bernstein and Vazirani \cite{BV93} considered the task of determining a hidden string from inner product of bit strings in
a setting where an algorithm is granted input access to evaluations of the function (here $\oplus$ denotes addition modulo $2$):\\

 \begin{mdframed}[backgroundcolor=green!5] 
\textbf{Bernstein-Vazirani Problem:}\\
\textit{
Learn a string $s \in \{0,1\}^n$ by making queries to a Boolean function,
$f_s: \{0,1\}^n \rightarrow \{0,1\}$, where
$$f_s(\vec x) = s_1 \cdot x_1 \oplus ... \oplus s_n \cdot x_n \, =\, \braket{\vec s,\vec x} \pmod 2.$$}
\end{mdframed}
Note that this problem features a curious resemblance to a variant of the \LWE problem in which the modulus is given by $q=2$,
the algorithm is free to choose all inputs (instead of receiving samples uniformly at random) and where the
noise is absent from all evaluations of the function.
In the classical query setting, we observe that a single query to the function can only reveal as much as a single bit of information 
about the secret string $\vec s$. 
In fact, this can easily be done by considering queries on strings $\vec e_i = (0,...\,, 1, ...\,,0)$,
where the $i$-th index is $1$ and $\vec e_i$ is $0$ everywhere else. Hence, any algorithm performing the above queries achieves an overall query complexity 
of at least $\Omega(n)$ when determining the secret, as each query reveals an outcome
\begin{equation}
f_s(\vec e_i) = \, \braket{\vec s,\vec e_i} \,  \pmod  2 \, = \, s_i,
\end{equation}
such that $\vec s$ is fully determined after a total of $n$ queries to the function. Therefore, it is tempting to approach
the \LWE problem by first closely examining this simplified model.

In this thesis, we consider the Bernstein-Vazirani problem in a setting in which an algorithm is given quantum access to the function,
hence is able to exploit quantum parallelism and to evaluate the inner product simultaneously on a superposition of inputs.
More formally, the algorithm can evaluate $f_s$ through a quantum operation, a black box whose inner workings
towards the computation of the function are unknown to the algorithm. 
We introduce the notion of an \textit{oracle}, a quantum operation $\mathcal{O}_{f_s}$ that allows for
the reversible evaluation of a function $f$ upon a set of inputs as follows:
\begin{equation}\label{oracle_query}
\mathcal{O}_{f_s}: \, \sum_{x,y \in \{0,1\}^n} \alpha_{x,y}\ket{x} \ket{y} \longrightarrow \sum_{x,y \in \{0,1\}^n} \alpha_{x,y} \ket{x}\ket{y \oplus f_s(x)}.
\end{equation}
Remarkably, as Bernstein and Vazirani \cite{BV93} showed, only a single oracle query to the the function as in Eq.\eqref{oracle_query} is sufficient
to determine the secret string.
We generalize this model to a group $\mathbb{Z}/q\mathbb{Z}$ of arbitrary positive integers $q$ under cyclic addition in a new learning problem extension of the Bernstein-Vazirani algorithm and discuss
its speed-up over classical algorithms. Cross et al. \cite{CSS14} have recently demonstrated a robustness of quantum learning
for certain classes of noise in which samples are also likely to be corrupted. While this setting is known to cause most learning problems
intractable for classical algorithms, the analogue using quantum samples remains easy. 
Recently, Grilo et al. \cite{GK17} independently considered a similar algorithm for \LWE, a special variant of
our proposed extended Bernstein-Vazirani algorithm in which $q$ is prime.
While this algorithm does not solve the \LWE problem in its original formulation using classical samples, it does
however suggest further caution when allowing access to quantum samples in any cryptographic application.
Nevertheless, not even a quantum computer receiving classical \LWE samples, i.e. classical strings of noisy linear equations, seems to be able to challenge
the hardness of \LWE \cite{Reg09}. For this reason, \LWE is still believed to be an excellent basis of hardness in post-quantum cryptography.

While quantum superposition access is regularly shown to be a powerful model, it also possesses limitations. Our goal in this work is also to
find such limitations in order to provide quantum-secure encryption schemes, even in a setting in which an attacker has quantum access 
to the encryption procedure.
An essential building block for the construction of secure cryptographic schemes is found in so-called pseudorandom functions, 
a family of keyed functions that seem indistuinguishable from perfectly random functions to any adversary with limited computational recources.
In fact, recent breakthroughs in quantum cryptography allow for quantum-secure pseudorandom functions that
are secure, even if an adversary is given the ability to evaluate the function using quantum superpositions.
Remarkably, as shown by Zhandry in 2012, such constructions can be built using the classical sample hardness of \LWE
in the quantum world \cite{Zha12}:\\

\textit{If \LWE with classical samples is hard for quantum computers, then there exist quantum-secure pseudorandom functions.}\\
\hspace{2mm}\\
As parallelism remains one of the key features of quantum algorithms, modern research is concerned with
exploitation of the nature of complex-valued amplitudes of quantum states
in order to cause them to interfere around the desired outputs through the use of quantum operations.
Only then, a final measurement of the state collapses the superposition into the desired outcome 
with high probability. The following fact guarantees that quantum parallelism can be achieved for
all efficiently computable functions \cite{NC10}:\\
\hspace{2mm}\\
\textit{Any classical efficiently computable function has an efficient circuit description, hence can also be implemented efficiently 
using a quantum computer.
Moreover, the quantum circuit for the function consists entirely of unitary gates and can thus be evaluated on a superposition of inputs
due to the linearity of quantum mechanics.}\\
\ \\
A fundamental question arises immediately. Just how powerful is knowledge represented in a quantum superposition evaluating a function on all of its inputs?
This thesis is concerned with both the limitation and exploitation of quantum parallelism in the context of modern cryptography.

An important attack model in cryptanalysis is that of \textit{chosen-ciphertext attacks},
a setting in which an adversary exercises control over the encryption scheme, for example by manipulating an honest party
into generating both encryptions
and decryptions of plaintexts or ciphertexts. The security under chosen-ciphertext attacks
is commonly formalized in an indistinguishability game that takes place in two phases.
In the pre-challenge phase, the adversary is allowed to perform encryption and decryption queries. Then, upon a pair of two messages,
the adversary receives a challenge ciphertext, an encryption of one of the two messages at random, and proceeds with another query phase.
Typically, we grant encryption access during both phases, while for decryption access, we differentiate between two important variants:
\begin{itemize}
\item \emph{(non-adaptive access)} the adversary exercises partial control and can only generate decryptions
prior to seeing a challenge ciphertext.
\item \emph{(adaptive access)} the adversary exercises full control and can perform informed decryption queries both
before and after the challenge phase begins, with the exception of the challenge ciphertext itself.
\end{itemize}
Oftentimes in cryptography jargon, the term \textit{lunchtime attack} is adopted in order to highlight a possible realistic
setting for a non-adaptive attack model, whereas an adaptive attack corresponds to full control over an honest party.

At STOC 2000, Katz and Yung \cite{KY00} offered a complete characterization of classical security notions for private-key encryption. In the 
case of classical communication in a quantum world, 
many of these security notions are still widely unexplored, and only few separation results have been successfully proven in recent years.
At CRYPTO 2013, Boneh and Zhandry first introduced the notion of
adaptive quantum chosen-ciphertext security and proposed classical encryption schemes
for which such security can be achieved \cite{BZ13}. An interesting open problem concerns the class of 
non-adaptive quantum chosen-ciphertext attacks, a security notion in which we 
allow adversaries to issue quantum superposition queries to encryption and non-adaptively to decryption.
In particular, it is unknown whether many of the standard encryption schemes satisfy such a weaker notion of security.

\newpage
\section{Technical Summary of Results}

Let us now give an overview of the main contents provided in this thesis.

In \textsc{Chapter 3}, we review selected topics in modern cryptography required for the proposed constructions in this thesis.
In \expref{Definition}{def:skes}, we introduce the concept of \textit{symmetric-key encryption schemes} $(\SKES)$, a setting in which two agents, say Alice and Bob, share a matching 
secret key prior to their communication.
In \expref{Definition}{def:comp_security}, we quantify limited computational power by introducing the notion of \textit{efficient} adversaries who
run algorithms with at most polynomial running time with regard to some security parameter relevant to the underlying cryptographic scheme.
A convenient security definition is one that formalizes the notion of indistinguishable encryptions.
The indistinguishability game introduces a game-based defintion of indistinguishable encryptions that takes place between an adversary and a challenger.
Here, the adversary prepares two plaintexts $m_0$ and $m_1$ and sends them to the challenger who chooses a bit $b$ uniformly at random and then
responds with an encryption of $m_b$. Thus, upon receiving a challenge ciphertext, the goal of the adversary is to output $b$. 
We say that an encryption scheme has indistinguishable encryptions
if no adversary wins the indistinguishability game with nonnegligible probability better than the trivial adversary who guesses $b$ at random.
We introduce the notion of indistinguishable encryptions under chosen-plaintext attacks (\expref{Definition}{def:ind-cpa}), as
well as under chosen-ciphertext attacks (\expref{Definition}{def:ind-cca}).
Another intuitive definition of security we consider is \textit{semantic security} (\expref{Definition}{def:sem_cca}), a notion of security that emphasizes
the possibility of an adversary attempting to compute something meaningful upon the encryption of a plaintext, such as a function of the plaintext.
In the semantic security game, the adversary is given an encryption of a plaintext $m$ and some side information $h(m)$, and the goal is to compute
a function $f(m)$ evaluated at the plaintext.
We say that an encryption scheme has semantic security
if every adversary is approximately identical to a \textit{simulator} who is given the side information $h(m)$ only. Therefore, semantic security
formalizes the intuition that even if the adversary has access to the ciphertext, essentially no advantage in computing anything
meaningful from it exists.
Furthermore, in \expref{Definition}{def:prf}, we define the concept of \textit{pseudorandom functions} (\PRF), a crucial building block in 
symmetric-key cryptography that 
allows for constructions of symmetric-key encryption schemes of precisely such security. The standard \PRF scheme is defined as follows:\\
\ \\
\textbf{\PRF scheme} (informal)\textit{ Given a family of pseudorandom functions $\mathcal F=\{f_k\}_k$, we define the scheme 
$\Pi[\mathcal F]=(\KeyGen,\Enc,\Dec,)$ which encrypts a plaintext $m$ using randomness $r$ via
$$\Enc_k(m; r) = (r, f_k(r) \oplus m).$$
To decrypt a ciphertext $(r,c)$, the procedure $\Dec_k(r,c)$ outputs $c \oplus f_k(r)=m$.
}\\
\ \\
Finally, we define the \LWE problem rigorously and discuss its applications in cryptography. 
We consider the standard \INDCPA-secure \LWE-based symmetric-key encryption scheme:\\
\ \\
\textbf{\LWE scheme} (informal)\textit{ The symmetric-key encryption scheme \LWE-\SKES$(n, q, \chi)$ is defined by an integer $n$, a modulus $q$ and 
a discrete error distribution $\chi$ over $\Z_q$ of certain bounded noise magnitude. 
The key for this scheme is a random vector $\vec k \rand \Z_q^n$. We encrypt a bit $b$ as follows:
\begin{enumerate}
  \item Sample a uniformly random vector $\vec a \rand \Z_q^n$ and an error $e \from \chi$;
  \item Output $(\vec a,\langle \vec a, \vec k \rangle + b \left \lfloor q/2\right \rfloor + e)$.
\end{enumerate}
To decrypt a ciphertext $(\vec a,c) \in \Z_q^{n+1}$, we output $0$ if and only if $|c - \langle \vec a, \vec k \rangle| \leq \left \lfloor q/4\right \rfloor$ (here we rely on the assumption that the error magnitude is bounded: $\abs{e} \leq \left \lfloor q/4 \right \rfloor$). 
}\\
\ \\
This scheme satisfies (classical) \INDCPA security under the \LWE assumption~\cite{Reg09}.
We then consider the \LWE-\SKES scheme to establish a separation between the previous notions of indistinguishable encryptions, both under chosen-plaintext attacks,
as well as under non-adaptive chosen-ciphertext attacks.

In \textsc{Chapter 4}, we present the most important developments in the theory of quantum computation to date.
To this end, we introduce the concept of qubits, unitary quantum operations and the quantum circuit model. We present a universal set of quantum gates
that enables a quantum computer to approximately perform any quantum operation (\expref{Theorem}{thm:universality}). 
Moreover, we give examples of quantum parallelism and show how to 
prepare a quantum state that evaluates a given function simultaneously over the range of its inputs. In this context,
we introduce the concept of quantum oracles, essentially a quantum gate that acts as a black box and grants an algorithm input access to a given function.
Finally, we turn to noise and decoherence in quantum computing architectures and give examples of elementary error correcting codes.

In \textsc{Chapter 5}, we review several of the well known quantum algorithms that solve certain computational tasks faster than any known classical algorithm
and provide the foundation for the algorithms of the later chapters. 
In particular, we introduce the \textit{Deutsch-Josza} algorithm, the earliest quantum speed-up
ever to be found in a black box model, as well as the \textit{Bernstein-Vazirani algorithm} as the original predecessor of the \textit{Extended Bernstein-Vazirani algorithm}.

In \textsc{Chapter 6}, we introduce the quantum Fourier transform (\QFT) over arbitrary finite abelian groups as a 
fundamental operation adopted in the majority of all the algorithms discussed in this thesis.
The Fourier transform (\expref{Definition}{def:qft}) is particularly useful in exploiting the symmetries of a 
given problem and allows us to generalize the Bernstein-Vazirani algorithm
over arbitrary cyclic groups. In \expref{Lemma}{lem:ortho}, we prove a widely used property on the orthogonality
of Fourier coefficients.
Finally, we discuss efficient quantum circuit implementations that compute the quantum Fourier transform.

In \textsc{Chapter 7}, we introduce useful language from \textit{computational learning theory} in which we frame the main algorithms in this thesis.
We consider a setting in which a learner (an algorithm) is requesting samples from a black box oracle whose inner workings are unknown.
The goal of the learner is to determine a hidden \textit{concept}, such as a Boolean function, based on the information that is being presented by the samples.
As each sample may be subjected to noise, potential errors are likely to get amplified and oftentimes lead to highly non-trivial tasks that are
computationally intractable for classical computers.
We consider the \textit{Learning Parity with Noise} (\LPN) problem, an early predecessor of the \LWE problem, 
as an instance of a computational learning problem. 
Once we define the analogous learning problem in a setting in which the oracle is providing quantum samples, we investigate how these
computational tasks become easy for quantum computers. We approach a quantum \LWE analogue by first proposing a new generalization 
of the Bernstein-Vazirani algorithm over
an arbitrary group under cyclic addition. We then prove \expref{Theorem}{thm:extended_BV} and show the following:\\
\ \\
\textbf{Theorem} (informal)\textit{ There exists a quantum algorithm for the Extended Bernstein-Vazirani problem
that can be amplified towards a success probability of $1-\delta$ by requesting
$O(\log1/\delta)$ many samples independently of $n$, whereas any classical algorithm requires $\Omega(n)$ many queries.}\\
\ \\
In addition, we compare our results to an independent 2017 proposal by Grilo and Kerenidis that proves that, in the quantum oracle setting, 
the extended Bernstein-Vazirani algorithm (in the special case where $q$ is prime) solves the \LWE problem given enough quantum samples.

In \textsc{Chapter 8}, we take a turn towards studying the limitations of quantum algorithms in order to
find secure constructions for post-quantum cryptography. While the previous chapter focused
on quantum speed-ups at solving learning problems by means of superposition samples, this chapter investigates
the limitations of quantum algorithms instead. We discuss the effects of relabeling in quantum algorithms,
a setting in which we relabel the function to which the algorithm is given oracle access at a subset of the domain
and study its subsequent output states, similar to the blinding of quantum algorithms proposed by Alagic et al.\cite{AMRS18}.
We introduce two variants of a new indistuinguishability game called \RelabelingGame,
a setting in which a quantum distinguisher receives quantum oracle access to a function and the goal is to detect
its modification as part of a game-based experiment. We distinguish between two variants,
a non-adaptive experiment in which the query phase takes place prior to the challenge, as well as an adaptive experiment in which
the query phase takes place during the challenge phase.
Thus, we define a non-adaptive relabeling game as an experiment in which a quantum algorithm first receives quantum oracle access
to a function and then, upon receiving a random input/output pair, the goal is to decide whether it is genuine (or modified)
based on the previous query phase.\\
\ \\
\textbf{Definition} (informal)\textit{ Given an arbitrary function $f:\{0,1\}^{n} \longrightarrow \{0,1\}^m$, 
we define the non-adaptive experiment $\RelabelingGame^{(1)}$ with a \QPT algorithm $\algo D$ as follows:
\begin{enumerate}
\item a bit $b \rand \bit$ and strings $r^*\rand \bit^n$, $s\rand \bit^m$ are generated;
\item \D receives quantum oracle access to $\mathcal{O}_f$;
\item depending on the random bit $b$, \D receives the following:
\begin{itemize}
\item $(b=0):$ \D receives a pair $(r^*,f(r^*))$;
\item $(b=1):$ \D receives a pair $(r^*, f(r^*)\oplus s)$.
\end{itemize}
Then, \D receives an example oracle that outputs classical random pairs $(r,f(r))$.
\item $\mathcal{D}$ outputs a bit $b'$ and wins the game if $b'=b$.
\end{enumerate}}
\ \\
We then prove \expref{Theorem}{def:non_relabeling_game} by controlling the success probability of \D in terms of the number of queries it makes. 
The proof uses a hybrid argument, adapting a variation of the standard quantum query lower bound technique, as well as the bound on the effects of blinding in \cite{AMRS18}, to give precise control over the success probability.\\
\ \\
\textbf{Theorem} (informal)\textit{ Given an arbitrary function $f: \{0,1\}^{n} \longrightarrow \{0,1\}^m$, any efficient quantum algorithm
making $T(n) = \poly(n)$ many oracle queries succeeds at the non-adaptive experiment $\RelabelingGame^{(1)}$ with
advantage at most $O(T(n)/\sqrt{2^{n}})$, except with negligible probability.}\\
\ \\
Next, we consider a stronger variant of the relabeling game (\expref{Definition}{def:adaptive_relabeling_game}), an adaptive setting in which a quantum algorithm first receives an arbitrary 
advice state (possibly even exponential-sized) for a function
and the goal is to detect whether it was relabeled at a random location.\\
\ \\
\textbf{Definition} (informal)\textit{ Given a function $f:\{0,1\}^{n} \longrightarrow \{0,1\}^m$, an arbitrary quantum advice state $\ket{\psi^f}$ (possibly depending on $f$) and
integer $0 \leq \mu \leq n$, we define the adaptive experiment $\RelabelingGame^{(2)}$ with a \QPT algorithm $\algo D$ as follows:
\begin{enumerate}
\item \D receives an advice state $\ket{\psi^f}$; \\
      a bit $b \rand \bit$ and strings $s\rand \bit^m$, $r^*\rand \bit^\mu$ are generated;
\item depending on the random bit $b$, \D receives the following:
\begin{itemize}
\item $(b=0):$ \D receives quantum oracle access to $\mathcal O_f$;
\item $(b=1):$ \D receives quantum oracle access to $\mathcal O_{f^*}$,
where $f^*$ is the relabeled function,
\begin{equation*}
f^*(x) :=
    \begin{cases}
      f(x) \oplus s &  \text{if the last $\mu$ bits of $x$ are equal to $r^*$}, \\
      f(x) & \text{otherwise}.
    \end{cases}
\end{equation*}
\end{itemize}
\item $\mathcal{D}$ outputs a bit $b'$ and wins the game if $b'=b$.
\end{enumerate}}
\ \\
Unlike in the previous non-adaptive variant, any distinguisher is able to adaptively make queries based on prior information on the target function from the pre-challenge phase.
Finally, we prove \expref{Theorem}{thm:adaptive_relabeling} on the success probability of the adaptive experiment.\\
\ \\
\textbf{Theorem} (informal)\textit{ Given an arbitrary function $f: \{0,1\}^{n} \longrightarrow \{0,1\}^m$ with arbitrary advice state $\ket{\psi^f}$ (possibly depending on $f$)
and integer $0 \leq \mu \leq n$, any efficient quantum algorithm
making $T(n) = \poly(n)$ oracle queries succeeds at the adaptive experiment $\RelabelingGame^{(2)}$ with
advantage $O(T(n)/\sqrt{2^{\mu}})$, except with negligible probability.}\\
\ \\
In choosing $\mu$ to be super-logarithmic in $n$, we can achieve a negligible advantage in the game.

In \textsc{Chapter 9}, we extend the notions of classical
indistinguishability from the earlier chapters to a quantum world. We make use of the rebeling result and propose secure
constructions under a quantum chosen-ciphertext attack. In this scenario, a quantum adversary
exercises control over the functionality of the scheme and is able to influence an honest party
into quantumly generating ciphertexts, as well as decrypting ciphertexts of the adversaries choice for some period in time.
We introduce several new quantum notions of security, such as indistinguishable encryptions under
non-adaptive quantum chosen-ciphertext attacks (\expref{Definition}{ind-qcca}), as follows:\\
\ \\
\textbf{Definition} (informal)\textit{ $\Pi$ is \INDQCCA if no quantum polynomial time algorithm $(\QPT)$ $\A$ can succeed at the following experiment with probability better than $1/2 + \negl(n)$.
\begin{enumerate}
\item A key $k \from \KeyGen(1^n)$ and a uniformly random bit $b \inrand \bit$ are generated;
\item $\algo A$ gets access to oracles $\Enc_k$ and $\Dec_k$, and outputs $(m_0, m_1)$;
\item $\algo A$ receives a challenge $\Enc_k(m_b)$ and access to $\Enc_k$ only; then \A outputs a bit $b'$;
\item $\algo A$ wins if $b = b'$.
\end{enumerate}}

We then introduce a quantum variant of semantic security under non-adaptive quantum chosen-ciphertext attacks. 
Finally, we prove that our proposed constructions based on quantum-secure pseudorandom functions
satisfy our definitions.\\
\ \\
\textbf{Theorem} (informal)\textit{ If $\mathcal{F}$ is a family of quantum-secure pseudorandom functions, then the \PRF scheme
$\Pi[\mathcal{F}]=(\KeyGen,\Enc,\Dec)$ is $\INDQCCA$-secure.}\\
\ \\
Moreover, we prove that quantum-secure pseudorandom functions are not strictly necessary to achieve \INDQCCA security of the \PRF scheme.
We consider a choice of post-quantum secure pseudorandom functions $\mathcal F'$, i.e. families of functions that are secure against quantum distinguishers
with classical access to the function, by equipping a \QPRF with a random large period. Note that due to quantum period finding, as observed in \cite{BZ13},
it follows that, if $\mathcal F$ is a family of \QPRF{s}, then $\mathcal F'$ is only post-quantum secure. Finally,
we prove that the \PRF scheme under $\mathcal F'$ achieves \INDQCCA security.\\
\ \\
\textbf{Theorem} (informal)\textit{ There exist families $\mathcal{F'}$ of post-quantum-secure pseudorandom functions for which the \PRF scheme
$\Pi[\mathcal{F'}]=(\KeyGen,\Enc,\Dec)$ is $\INDQCCA$-secure.}\\
\ \\
In \textsc{Chapter 10}, we discuss state-of-the-art quantum computing technology with a particular focus on the ion-trap architecture.
We give a detailed introduction to how qubits are realized in a physical system and
how quantum gates can be performed through the use of lasers. Furthermore, we discuss sources of noise and decoherence
in physical systems in order to investigate the effectiveness of noise models from the previous chapters. To this end, we
discuss the performance of recent implementations of quantum algorithms discussed in this thesis. Finally,
we discuss an experimental comparison between a five-qubit ion-trap implementation and the five-qubit IBM superconductor device.

%% file: cryptographic_primitives.tex
\newpage
\section{Cryptography}

The history of cryptography dates back to over two millenia. Ever since the birth of civilization and the invention of writing, people
required ways of transmitting secret messages using \textit{ciphers}, intended to be read only by the receiver and yet
difficult to decode for others. Since the 1970s, cryptography amounted to a well-established scientific discipline
by henceforth adopting a rigorous mathematical foundation. This crucial change marks the beginning of \textit{modern cryptography}.
Many of the popular encryption schemes still in use today, such as the \RSA encryption scheme,
were already developed in the early years of modern cryptography.
Typically, it is the hardness of certain computational problems that serves 
as a foundation for security. For example, as
in the case of \RSA, the security of the encryption scheme is related to the hardness 
of factoring large integers. In other words, we believe a scheme is secure, if no efficient adversary with limited computational recources 
is capable of breaking the scheme.
Peter Shor's discovery of an efficient quantum algorithm for the factoring of integers marked the beginning of an entirely new era of cryptography,
a so-called \textit{post-quantum cryptography}. It is from here on, that the search for quantum-secure cryptography began.
In the following sections, we provide an overview of selected topics in modern cryptography required for the main results in this thesis.

\subsection{Preliminaries}

Let us first introduce some necessary notation and formalism from theoretical computer science and cryptography. For additional reading,
we refer to \cite{KL15}.

For bit strings $x \in \bit^n$ of arbitrary length $n=|x|$, we associate a product space $\bit^*$ containing all such strings of finite length. 
A function $\negl: \mathbb{N} \rightarrow \mathbb{R}$ is called \textit{negligible} if, for every polynomial $p$, 
there exists an integer $N$ such that for all $n>N$, it holds that:
$ \negl(n) < \frac{1}{p(n)}.$ Typically, we adopt negligible functions in the context of a success probability that
decreases to an inverse-superpolynomial rate, hence cannot be amplified to a constant by a polynomial amount of repetitions.
An algorithm is a sequence of (possibly nondeterministic) operations that
terminates after a finite amount of steps upon any given input, say $x \in \bit^*$.
We say an algorithm is \textit{efficient} if it has polynomial running time
with respect to a size parameter of a given computational problem, i.e. if there exists a polynomial
$p(x)$ such that, for any input $x \in \bit^*$, the computation of $A(x)$ terminates after at most $p(|x|)$ steps. 
A probabilistic polynomial time $(\PPT)$ algorithm is a procedure with an additional random tape (such as a random number generator) that 
results in efficient, yet possibly nondeterministic, computations. We adopt the popular unary convention of representing the seed of
efficient randomized algorithms by $1^n=11...1$, highlighting a polynomial dependence with respect to the length of the input,
contrary to a polylog dependence in the general case where $\left \lceil{\log_2(n)}\right \rceil$ bits are needed to specify the length of a given input
(here, $\left \lceil{\cdot}\right \rceil$ denotes the ceiling function). With $x \rand X$, we denote a procedure
an outcome $x$ is sampled uniformly at random from a finite set $X$. If $D$ is a probability distribution,
we denote the sampling of an outcome according to $D$ by using the notation $x \leftarrow D$. Upon finite sets $X$ and $Y$,
we define the corresponding (finite) set of all possible functions from $X$ to $Y$ as $\{\mathcal{F}: \mathcal{X} \rightarrow \mathcal{Y}\}$.
An \textit{oracle} is a black box machine $\mathcal{O}$ that assists a given algorithm with a particular computational task
at unit cost, for example in an evaluation of an unknown function upon a given input or
the sampling from an unknown probability distribution. Typically, if \A is an algorithm, we denote oracle access to $\mathcal{O}$ using the
notation $\A^{\mathcal{O}}$.
Finally, throughout this thesis, we employ the usual asymptotic $O$-notation denoting an upper bound, where for a given function $g(n)$, we define
$O(g(n))= \{ f(n) \, : \, \exists c \in \mathbb{R}, \exists n \in \mathbb{N} \text{ such that } 0 \leq f(n) \leq c\,g(n), \, \forall n \geq n_0\}$.
Similarly, we denote an asymptotic lower-bound $\Omega(g(n))$ as the set of functions 
$\Omega(g(n))= \{ f(n) \, : \, \exists c \in \mathbb{R}, \exists n \in \mathbb{N} \text{ such that } 0 \leq c\,g(n) \leq f(n), \, \forall n \geq n_0\}$.\\
\ \\

\subsection{Symmetric-Key Cryptography}

Symmetric-key cryptography concerns the scenario in which two agents, say Alice and Bob, share a mutual secret key $k$ prior to their communication
and want to send messages to each other.
In order to encrypt messages, Alice first chooses a message $m$ and runs an encryption algorithm $\Enc_k(m)$ that 
requires the use of her key and later sends the resulting ciphertext $c$ over to Bob.  
Since Bob knows about the secret key, he can run a decryption algorithm $\Dec_k(c)$ upon Alice's ciphertext and decode the message.
In general, we consider randomized encryption in order to avoid \textit{replay attacks}, while only requiring decryption to be deterministic.
\begin{definition}\label{def:skes}
A symmetric-key encryption scheme $(\SKES)$ $\Pi =(\KeyGen,\Enc,\Dec)$ is a triple of \PPT algorithms on a 
finite key space
\K, message space \M and ciphertext space \C, where $\KeyGen: \mathbb{N} \rightarrow \K$, $\Enc: \K \times \M \rightarrow \C$,
$\Dec: \K \times \C \rightarrow \M$ and, for a security parameter $n$, we require:
\begin{enumerate}
\item \emph{(key generation)} $\KeyGen$: on input $1^n$, generate a key $k \from \KeyGen(1^n)$;
\item \emph{(encryption)} $\Enc_k$: on message $m\in \mathcal{M}$, output a ciphertext $\Enc_k(m)$;
\item \emph{(decryption)} $\Dec_k$: on cipher $c \in \C$, output a message $\Dec_k(c)$;
\item \emph{(correctness)} $(\Dec_k \circ \Enc_k)(m) = m$.
\end{enumerate}
\end{definition}
In order for communication under a given symmetric-key encryption scheme to be secure against eavesdroppers, 
we require that, without knowledge of the secret key, any ciphertext
must look sufficiently random and reveal little to no information about the actual message.\\
\\
In the next section, we provide several widely used notions of security for symmetric-key encryption. For further
reading, we refer to \cite{KL15}.

\newpage

\subsection{Security Notions}\label{ch:security}

\subsubsection{Computational Security}

Due to the well known \Poly-\NP problem, i.e. the seeming impossibility of finding efficient algorithms for certain computational problems whose solutions can be quickly verified,
and the fact that we consider adversaries who operate probabilistically, an important notion of security is provided by \textit{computational security}
based on the following principle:\\
\ \\
\textit{A successful cipher must be practically secure against adversaries with limited computational recources.}\\
\ \\
This brings us to the following standard definition of computational security:

\begin{definition}[Computational Security]\label{def:comp_security}
A scheme $\Pi =(\KeyGen,\Enc,\Dec)$ is computationally (or asymptotically) secure if every \PPT adversary succeeds
at breaking $\Pi$ with at most negligible probability with respect to the security parameter of $\Pi$.
\end{definition}
Since a negligible success probability is smaller than the inverse of any polynomial, no efficient algorithm is capable
of amplifying the success probability, i.e. capable of breaking the encryption scheme by sheer repetition. Therefore,
we regard any algorithm that breaks a particular scheme with at most negligible probability as not significant.

\subsubsection{Computational Indistinguishability.}

Another important notion of security for a given symmetric-key encryption scheme is \textit{indistinguishability of encryptions}, in particular
under a \textit{chosen-plaintext attack}.
In this model, an adversary has partial control over the encryption procedure and can generate encryptions of arbitrary messages.
This attack corresponds to a scenario in which an attacker is able to influence an honest party
into generating ciphertexts of the adversaries choice, thus potentially resulting in an advantage at decoding other ciphers of interest.
In the following, we specify this model in a security game between an adversary and a challenger: 

\begin{definition}[\INDCPA] \label{def:ind-cpa}\ \\
Let $\Pi = (\KeyGen, \Enc, \Dec)$ be a symmetric-key encryption scheme and consider the
$\INDGame$ between a \PPT adversary and challenger, defined as follows:
\begin{enumerate}
\item \emph{(initial phase)} the challenger chooses a key $k \from \KeyGen(1^n)$ and bit $b \inrand \bit$;
\item \emph{(pre-challenge phase)} as part of a learning phase, the adversary is given access to an encryption oracle $\Enc_k$ in order to generate encryptions. 
Upon each choice of message $m$, the adversary
receives a ciphertext $c \leftarrow \Enc_k(m)$. Finally, the adversary chooses two messages $m_0$ and $m_1$, and sends them to the challenger.
\item \emph{(challenge phase)} the challenger replies with $\Enc_k(m_b)$ and the adversary continues to have oracle access to $\Enc_k$;
\item \emph{(resolution)} the adversary outputs a bit $b'$ and wins the game if $b'=b$.
\end{enumerate}
We say that $\Pi$ has indistinguishable encryptions under a chosen-plaintext attack (or is \INDCPA-secure) if, for every \PPT \A, 
there exists a negligible function $\negl(n)$ such that:\\
$\Pr[\A \text{ wins } \INDGame] \leq 1/2 + \negl(n)$.
\end{definition}

An even stronger notion of security for a given symmetric-key encryption scheme is \textit{security under chosen-ciphertext attacks}.
In this variant of the \INDGame, an adversary not only exercises control over the encryption scheme as before,
but can also \textit{non-adaptively} decrypt messages unrelated to a ciphertext of interest (as highlighted in the pre-challenge and challenge phase).
Therefore, such an attack corresponds to a scenario in which an attacker is able to exercise control over an honest party
into generating ciphertexts, as well as decrypting ciphertexts of the adversaries choice for some period in time.
In the following, we specify this model in another security game between an adversary and a challenger: 

\begin{definition}[\INDCCA] \label{def:ind-cca}\ \\
Let $\Pi = (\KeyGen, \Enc, \Dec)$ be a symmetric-key encryption scheme and consider the
$\INDGame$ between a \PPT adversary and challenger, defined as follows:
\begin{enumerate}
\item \emph{(initial phase)} the challenger chooses a key $k \from \KeyGen(1^n)$ and bit $b \inrand \bit$;
\item \emph{(pre-challenge phase)} as part of a learning phase, the adversary is given access to 
both an encryption oracle $\Enc_k$ and decryption oracle $\Dec_k$. 
Upon each choice of message $m$, the adversary
receives a ciphertext $\Enc_k(m)$ and, upon each ciphertext $c$,
the adversary receives a plaintext $\Dec_k(c)$.
Finally, the adversary chooses two messages $m_0$ and $m_1$, and sends them to the challenger.
\item \emph{(challenge phase)} the challenger replies with $\Enc_k(m_b)$ and the adversary continues to have oracle access to $\Enc_k$ only;
\item \emph{(resolution)} the adversary outputs a bit $b'$ and wins the game if $b'=b$.
\end{enumerate}
We say that $\Pi$ has indistinguishable encryptions under a chosen-ciphertext attack (or is \INDCCA-secure) if, for every \PPT \A, 
there exists a negligible function $\negl(n)$ such that:\\
$\Pr[\A \text{ wins } \INDGame] \leq 1/2 + \negl(n)$.
\end{definition}
Finally, we can additionally extend the previous notion of \INDCCA security by also granting the adversary \textit{adaptive} decryption access
after the challenge phase. This model corresponds to \INDCCAA security,
a variant in which the adversary exercises full control over the encryption scheme, both before and after the challenge phase.
Remarkably, there exist classical symmetric-key encryption schemes that satisfy each of the security definitions provided in this chapter.
A major contribution of this thesis is to provide constructions that satisfy these notions, even in a setting in which the adversary
is granted quantum superposition access, again both to the encryption and decryption procedure. In the next section,
we introduce important tools to realize such cryptographic schemes.

\subsubsection{Semantic Security.}

In semantic security, the challenge phase corresponds to choosing a \emph{challenge template} instead of a pair of messages.
Contrary to the \INDGame, the intuition for this security game is that
the adversary seeks to compute something meaningful about the message of interest during the challenge phase.
Thus, we consider challenge templates consisting of a triple of classical circuits $(\Samp, h, f)$, where $\Samp$ outputs plaintexts from some distribution 
$\mathcal D_\Samp$, and $h$ and $f$ are functions over messages $m \from \Samp$.
Upon receiving an encryption of a randomly sampled message $m$ according to \Samp, the goal of the adversary is to
output some new information $f(m)$, given some side information $h(m)$ on the message.
In providing an adversary with a \CCA learning phase, we can consider the following notion of security.

\begin{definition}[\SEMCCA]\label{def:sem_cca} Let $\Pi = (\KeyGen, \Enc, \Dec)$ be an encryption scheme, and consider the experiment 
$\SEMGame$ with a \PPT $\algo A$, defined as follows.
\begin{enumerate}
\item \emph{(initial phase)} A key $k \from \KeyGen(1^n)$ and bit $b \rand \bit$ are generated;
\item \emph{(pre-challenge phase)} $\algo A$ receives access to oracles $\Enc_k$ and $\Dec_k$, then outputs a challenge template consisting of $(\Samp, h, f)$;
\item \emph{(challenge phase)} A plaintext $m \from \Samp$ is generated; $\algo A$ receives $h(m)$ and an oracle for $\Enc_k$ only; if $b = 1$, $\algo A$ also receives $\Enc_k(m)$.
\item \emph{(resolution)} $\algo A$ outputs a string $s$, and wins if $s = f(m)$.
\end{enumerate}
We say $\Pi$ is semantically secure under a non-adaptive chosen-ciphertext attack (or is \SEMCCA) if, for every \PPT $\algo A$, 
there exists a \PPT $\algo S$ such that the challenge templates output by $\algo A$ and $\algo S$ are identically distributed, and there exists a negligible
function $\negl(n)$ such that:
\begin{align*}
\left| \underset{k \rand \K}{\Pr}[\A (1^n,\Enc_k(m),h(m)) = f(m)] \, - \Pr[\Sim (1^n,|m|,h(m)) = f(m)] \right| \, \leq \, \negl(n),
\end{align*}
where, in both cases, the probability is taken over plaintexts $m \leftarrow \Samp$.
\end{definition}

Fortunately, as shown in \cite{KL15}, semantic security and indistinguishability are equivalent notions of security, in particular
under non-adaptive chosen-ciphertext attacks.

\begin{theorem}
Let $\Pi=(\KeyGen,\Enc, \Dec)$ be a symmetric-key encryption scheme. Then, $\Pi$ is $\INDCCA$-secure if and only if $\Pi$ is $\SEMCCA$-secure.
\end{theorem}

In Chapter \ref{ch:pq_crypto}, we introduce variants under quantum chosen-ciphertext attacks and prove the equivalence of both definitions.
While semantic security is a much more intuitive notion of security, it is oftentimes much harder to prove security in practice.
Therefore, it is convenient to provide security proofs under the notion of indistinguishable encryptions and to then refer to the equivalence
result for a more natural notion of security.

\subsection{Pseudorandom Functions}\label{ch:prf}

In this section, we turn to pseudorandom functions, a popular building block in symmetric-key cryptography.
Historically, the first instance of provably-secure pseudorandom functions was proposed in the Goldreich, Goldwasser and Micali
construction \cite{GGM86} using pseudorandom generators, which in turn rely on the existence of one-way functions.
The effectiveness of pseudorandom functions lies in the property of seeming indistinguishable from a perfectly random
function to any efficient distinguisher with limited computational power. The security properties
of pseudorandom functions are perhaps best explained in an indistinguishability game between a distinguisher (a \PPT algorithm)
and a challenger. Upon the start of the game, the challenger chooses a random bit $b$ whose outcome determines whether
the game is being played with a perfectly random function (sampled uniformly at random
from the finite set of all possible functions over given finite domain and range) of the challengers choice, or a pseudorandom function for
a freshly generated key. Next, the challenger presents the distinguisher with an oracle for the given function
who is then free to evaluate the function upon arbitrary inputs. Finally,
the distinguisher wins by outputting a bit $b'=b$. Since the distinguisher
is assumed to have limited computational recources, thus essentially running a \PPT algorithm, the claim of
pseudorandomness is that the outputs will look sufficiently random. Therefore, the probability that the
distinguisher makes a decision in a game against a pseudorandom function and outputs a bit, say $b'=1$, 
is negligibly close to a game in which the distinguisher is playing against a perfectly random function.
We formalize this observation in the following definition:

\begin{definition}(Pseudorandom Function)\label{def:prf}\ \\
Let $\mathcal F$ be an efficiently computable function $\mathcal F : \mathcal{K} \times \mathcal{X} \rightarrow \mathcal{Y}$ 
on a key-space $\mathcal{K}$, a domain $\mathcal{X}$ and a range $\mathcal{Y}$.
We say $\mathcal F= \{f_k\}_{k \in \K}$ is a family of pseudorandom functions $(\PRF)$
if, for all $k$ and \PPT distinguishers \D,
there exists a negligible function $\negl(n)$ such that:
\begin{align}
\left| \underset{k \rand \mathcal{K}}{\Pr}[ \mathcal{D}^{f_k}(1^n) = 1] \, - \underset{f \rand \{F: \mathcal{X} \rightarrow \mathcal{Y}\}}{\Pr}[ \mathcal{D}^{f}(1^n) = 1] \right| \, \leq \, \negl(n)
\end{align}
\end{definition}
Consider, for example, the following \SKES using a pseudorandom function, as found in Proposition 5.4.18 in \cite{Gol04}.
In this scheme, the pseudorandom function is used to both encrypt and decrypt messages using the same key.
\begin{construction}\label{cons:prf}

Consider a family $\mathcal F$ of keyed functions $f_k: \{0,1\}^n \longrightarrow \{0,1\}^n$, where $n$ is a security parameter and $\mathcal{K} = \bit^n$ is a key space.
Define a symmetric-key encryption scheme $\Pi[\mathcal F]=(\KeyGen,\Enc,\Dec)$ as follows:
\begin{enumerate}
\item \emph{(key generation)} $\KeyGen$: on input $1^n$, generate a key $k \rand \bit^n$;
\item \emph{(encryption)} $\Enc_k$: on message $m$, choose a randomness $r \rand\{0,1\}^n$ and output a ciphertext $\Enc_k(m;r) = (r,f_k(r) \oplus m)$;
\item \emph{(decryption)} $\Dec_k$: on cipher $(r,c)$, output $\Dec_k(r,c) = c \oplus f_k(r)$;
\item \emph{(correctness)} $(\Dec_k \circ \Enc_k)(m;r) = (f_k(r) \oplus m) \oplus f_k(r) = m$.
\end{enumerate}
\end{construction}
In fact, this scheme already satisfies the notion of \INDCPA security, for example as shown in \cite{KL15}.
In Chapter \ref{ch:pq_crypto}, we introduce a class of quantum-secure pseudorandom functions and prove the \INDCCA security of this scheme,
even in a setting in which the adversary is given quantum superposition access to the encryption oracle $\Enc_k$ and decryption procedure $\Dec_k$.

In the next section, we provide a formal definition of the Learning with Errors problem, as introduced in \cite{Reg09}.

\subsection{Learning with Errors}

The Learning with Errors problem can be stated in multiple variants, such as the search problem or the decision problem.
In the following, we begin by first defining the Learning with Errors search problem, as introduced in the introductory section.
 
\begin{definition}[\LWE Problem]\ \\
Let $n$ be a security parameter, let $q$ be a prime and let $\chi$ be a discrete probability distribution over errors in $\mathbb{Z}_q$. 
Let $\vec s \in \mathbb{Z}_q^n$ be a secret string and let $A_{s,\chi}$ be the probability distribution on $\mathbb{Z}_q^n \times \mathbb{Z}_q$ that performs the following:
\begin{enumerate}
\item Sample a uniformly random string $\vec a \in \mathbb{Z}_q^n$.
\item Sample an error $e \in \mathbb{Z}_q$ according to error distribution $\chi_q$.
\item Output $(\vec a,\braket{\vec a,\vec s} +e \pmod q)$.
\end{enumerate}
We say that a \PPT algorithm $\mathcal{A}$ solves the Learning with Errors problem $\LWE_{q,\chi}$ with modulus $q$ and error distribution $\chi$ if, 
for any $\vec s \in \mathbb{Z}_q^n$
and an arbitrary number of independent noisy samples from $A_{s,\chi}$, $\mathcal{A}$ outputs the secret $\vec s$ with nonegligible probability.
\end{definition}

Typically one chooses an error distribution $\chi_{\eta,q} \sim \mathcal{N}(0,\eta^2 q^2)$ that follows a discrete Gaussian distribution rounded to 
the nearest integer and reduced modulo $q$, where the noise magnitude $\eta >0$ is taken to be $1/poly(n)$. \textit{Chebyshev's inequality}
allows us to conveniently control the 
standard deviation $\eta q$ towards a sharply peaked error distribution around the origin for an appriopriate choice of 
parameters $\eta$ and $q$.
As Regev argues, there are several reasons that speak in favor of the hardness of the \LWE problem, particularly its
close relationship to lattice-problems and the \textit{Learning Parity with Noise} problem \cite{CSS14}, both studied
extensively and believed to be hard. Since \LWE can be thought of as a generalization of the \LPN problem, 
we believe that \LWE must also be hard.
Furthermore, the best known classical algorithms for solving the \LWE problem so far run in exponential time \cite{BKW03}.

\subsubsection{Decision Learning with Errors.}

A related variant of the \LWE problem is found in the task of determining whether a given sample results from a noisy linear equation
on a secret string, or a genuine uniform random sample.

\begin{definition}[Decision \LWE] \label{def:dlwe}\ \\
Let $\LWE_{q,\chi}$ be given by a sampling probability distribution $A_{s,\chi}$ for a string $\vec s \in \mathbb{Z}_q^n$ and
let $\mathcal{O}$ be the uniform distribution over $\mathbb{Z}_q^n \times \mathbb{Z}_q$.
We say that $\LWE_{q,\chi}$ satisfies the decisional \LWE assumption $(\DLWE_{q,\chi})$ with modulus $q$ and error distribution $\chi$ if, 
for all \PPT distinguishers \D, there exists a negligible function $\negl(n)$ such that:
\begin{align}
\left| \underset{s \rand \mathbb{Z}_q^n}{\Pr}[ \mathcal{D}^{A_{s,\chi}}(1^n) = 1] \, - \underset{}{\Pr}[ \mathcal{D}^{\mathcal{O}}(1^n) = 1] \right| \, \leq \, \negl(n),
\end{align}
where $\mathcal{O}$ outputs uniform samples $(\vec a,u) \rand \mathbb{Z}_q^n \times \mathbb{Z}_q$.
\end{definition}
Remarkably, as Oded Regev showed, there exists a simple reduction of the \LWE search problem towards the decisional \LWE problem.
While it is clear that an efficient algorithm for the search \LWE problem implies the existence of an algorithm for the decisional \LWE problem,
the opposite implication is guaranteed by the following lemma:

\begin{lemma}[\cite{Reg09}, Decision \LWE to Search \LWE] \label{lem:dlwe}\ \\
Let $n$ be a security parameter, $q$ be a prime and let $A_{s,\chi}$ be a sampling probability distribution $A_{s,\chi}$
for a string $\vec s \in \mathbb{Z}_q^n$ and discrete probability distribution $\chi_{q,\eta}$ over errors in $\mathbb{Z}_q$. 
If \A is an algorithm that solves the $\DLWE_{q,\chi}$ problem with nonegligible probability
over a uniform choice of strings $\vec s$,
then there exists an efficient algorithm $\A'$ that receives samples from $A_{s,\chi}$ and solves the
search \LWE problem with probability exponentially close to $1$.
\end{lemma}

\subsubsection{Symmetric-Key Constructions and Security.}
Let us now consider the following symmetric-key encryption scheme motivated by the \LWE hardness assumption, as suggested in \cite{Reg05}.
In this example of an encryption scheme, the secret string acts as a key and we encrypt a single bit by computing an \LWE sample in a suitable way that can be
detected by a receiver in possession of the key. 

\begin{construction}[\LWE scheme]\label{cons:LWE_sym}\ \\
Let $n$ be an integer, let $q$ be a modulus and let $\chi_q$ be a discrete error distribution over $\mathbb{Z}_q$ and 
consider the following symmetric-key encryption scheme $\LWE-\SKES(n,q,\chi)=(\KeyGen,\Enc,\Dec)$, defined as follows:
\begin{enumerate}
\item \emph{(key generation)} run $\KeyGen(1^n)$ and generate a key $\vec k \rand \mathbb{Z}_q^n$;
\item \emph{(encryption)}
upon a bit $b \in \bit$, sample a string $\vec a \rand \mathbb{Z}_q^n$ and error $e \from \chi_{q}$, and output
$\Enc_k(b)=(\vec a,\langle \vec a,\vec k \rangle + b \cdot \left \lfloor{\frac{q}{2}}\right \rfloor + e)$;
\item \emph{(decryption)} upon cipher $(\vec a,c)$, apply rounding to output $\Dec_k(\vec a,c)=0$ if and only if $|c - \langle \vec a, \vec k \rangle| \leq \left \lfloor{\frac{q}{4}}\right \rfloor$,
else output $1$.
\item \emph{(correctness)} $(\Dec_k \circ \Enc_k)(b) = b \,$ $(\text{with high probability})$.
\end{enumerate}
\end{construction}

Using the decisional \LWE assumption, we can easily show that \LWE-\SKES indeed satisfies a notion of indistinguishability
under a chosen-plaintext attack.

\begin{theorem}
Let \LWE-\SKES$(n,q,\chi)=(\KeyGen,\Enc,\Dec)$ be the symmetric-key encryption scheme from \expref{Construction}{cons:LWE_sym}. Then \LWE-\SKES is \INDCPA-secure.
\end{theorem}
\begin{proof}
We introduce a hybrid game by modifying the security game
in a way that is indistinguishable (to any \PPT adversary) from the original game in order to arrive at a security game in which the challenge
is perfectly hidden.
\begin{description}
\item[\textsc{Game 0}:] In the standard hybrid, the adversary is playing the \INDCPA security game 
for the original scheme $\Pi$ in \expref{Construction}{cons:LWE_sym}. Prior to the challenge, the adversary chooses message bits $b_0,b_1$ and is given
access to an encryption oracle $\Enc_k$. Upon receiving a challenge cipher $(\vec{a}^*,c^*) \leftarrow \Enc_k(b)$,
the adversary may perform additional queries to the encryption oracle and the goal is to decide whether the challenge corresponds to an encryption
of $b_0$ or $b_1$.

\item[\textsc{Game 1}:] In the this hybrid, the challenger instead responds with 
uniformly random samples $(\vec a,c) \rand \mathbb{Z}_q^n \times \mathbb{Z}_q$ upon each encryption query, as
well as with a challenge $(\vec {a}^*,c^*) \rand \mathbb{Z}_q^n \times \mathbb{Z}_q$.
From the decisional \LWE assumption in \expref{Definition}{def:dlwe} and \expref{Lemma}{lem:dlwe}, it follows that no \PPT adversary can distinguish between genuine
\LWE samples or uniformly random samples (both with $ b  \left \lfloor{\frac{q}{2}}\right \rfloor$ added to them).
\end{description}
Since adopting this hybrid game only negligibly affects the success probability of any \PPT adversary, we arrive at a security game in which
the adversary cannot win, except with at most negligible probability better than guessing at random.\qed
\end{proof}

\subsubsection{Separation Result.}

In preparation for the sections on post-quantum cryptography in which we study quantum access to decryption,
let us now conclude this chapter with a simple separation between the two notions of security from \expref{Section}{ch:security} and show
that there exist schemes that are \INDCPA-secure but not \INDCCA-secure.
Using the \LWE-\SKES scheme, we can easily prove such a separation. The intuition is that decryption oracle access in this scheme
allows the adversary to evaluate the noisy inner product upon arbitrary inputs.
As a result, the adversary can determine parts of the secret key using only a few queries to its decryption oracle.

\begin{lemma}
Let \LWE-\SKES$(n,q,\chi)=(\KeyGen,\Enc,\Dec)$ be the symmetric-key encryption scheme from \expref{Construction}{cons:LWE_sym}. Then \LWE-\SKES does not satisfy \INDCCA-security.
\end{lemma}
\begin{proof}
The following algorithm recovers the key using close to a linear number of classical decryption queries.\\

\IncMargin{1em}
\begin{algorithm}[H]
\SetKwInOut{Input}{input}\SetKwInOut{Output}{output}
\Input{Classical decryption oracle $\mathcal{O}_{\Dec_k}$ for $\LWE$-$\SKES(n,q,\chi)$\\ Parameter: $M \in \mathbb{N}$}
\Output{$\tilde k \in \Z_q^n$}
\BlankLine
\For{$i\leftarrow 1$ \KwTo $n$}{
initialize a list $X[1 \,.. \,M]$ of size $M$.\\
\For{$m\leftarrow 1$ \KwTo $M$}{
\emph{sample $c_m \rand \Z_q$}\;
\emph{query $b_m \leftarrow \Dec_k(\vec e_i, c_m)$, where $\vec e_i = (0,..., 1, ...,0)$}\;
\emph{let $X[m] = c_m-b_m\left \lfloor{\frac{q}{2}}\right \rfloor$}\;
}\textbf{end}\\
\emph{choose $\tilde k_i = \frac{1}{M} \displaystyle\sum_{m = 1}^{M} X[m]$}\;
}\textbf{end}
output: $\tilde k \in \Z_q^n$.
\caption{Classical Decryption-Access Key-Recovery}\label{alg:KeyRec}
\end{algorithm}\DecMargin{1em}\ \\

A standard analysis using Hoeffding's bound (\expref{Lemma}{lem:hoeffding}) and the union bound guarantees that
the adversary can now amplify the success probability by simply controlling for
the probability of failure through enough repetitions $M$.
Therefore, any polynomial time adversary with access to a decryption oracle can output $\tilde k = k$ with high probability given enough repetitions
$M$, hence can recover the key and obtain a crucial advantage in the security game. In order to win at the challenge and to distinguish between $b_0$ and $b_1$,
the adversary first computes the key during the pre-challenge phase and then uses it to evaluate an inner product between the key and the 
challenge randomness $a^*$. As a result, the adversary now succeeds at the indistinguishability game with nonnegligible probability.\qed
\end{proof}

%% file: quantum_computation.tex
\newpage
\section{Quantum Computation}

Quantum information processing is concerned with the storage and manipulation of information in a quantum system.
The fundamental unit of information is the \textit{qubit}, a quantum two-level system of states $\ket{0}$ and $\ket{1}$. 
Fortunately, nature presents us with many ways of realizing a qubit in a physical system. Typical representations of a qubit
are found in the two spin $1/2$ states of a particle, the vertical or horizontal polarization of a photon or simply a ground and excited state 
in the energy spectrum of an atom.
In this chapter, we give a brief overview of the most important concepts in the theory of quantum computing to date. 
For further reading, we refer to \cite{NC10}.
Finally, with regard to the physical realization of quantum computers, we refer to Chapter \ref{ch:physical}.

\subsection{Formalism}

A quantum system lives in a Hilbert space $\mathcal{H}$, a complex vector space together with an inner product $\braket{\cdot|\cdot}$. 
A qubit is a quantum system $\ket{\psi}$ of mutually orthogonal basis states $\ket{0}$ and $\ket{1}$, given by a normalized state 
vector of amplitudes $|\alpha|^2 + |\beta|^2 = 1$, where
\begin{equation}
\ket{\psi} = \alpha \ket{0} + \beta \ket{1}.
\end{equation}
Contrary to classical bits of information that carry definite states of either $0$ or $1$, a qubit can be represented as a continuous superposition 
of two basis states.
By introducing angular degrees of freedom $\phi$ and $\theta$, a qubit can be visualized as a point on the \textit{Bloch sphere}, as in \expref{Figure}{fig:bloch}, 
and written as\footnote{Note that we ignore the contributions from an overall phase as it produces no observable effects.}
\begin{equation}
\ket{\psi} = \cos\left(\frac{\theta}{2}\right) \ket{0} + e^{i \phi} \sin\left(\frac{\theta}{2}\right) \ket{1}.
\end{equation}
Given two quantum systems $\mathcal{H}_A$ and $\mathcal{H}_B$, the composition results in a joint quantum system given by
$\mathcal{H} = \mathcal{H}_A \otimes \mathcal{H}_B$, the tensor product of the two systems. Thus, for $\ket{\psi}_A \in \mathcal{H}_A$ 
and $\ket{\phi}_B \in \mathcal{H}_B$, the product state is given by $\ket{\psi}_A \otimes \ket{\phi}_B$.
For example, if $\ket{\psi}_A = \alpha \ket{0} + \beta \ket{1}$ and 
$\ket{\phi}_B = \delta \ket{0} + \gamma \ket{1}$, then:
\begin{align}
\ket{\psi}_A \otimes \ket{\phi}_B &= \alpha \delta \ket{0} \otimes \ket{0} + \alpha \gamma \ket{0} \otimes \ket{1} +\beta \delta \ket{1} \otimes \ket{0} +
 \beta \gamma \ket{1} \otimes \ket{1}.
\end{align}
\begin{figure}[tbp]
\centering 
\includegraphics[width=.40\textwidth,origin=c]{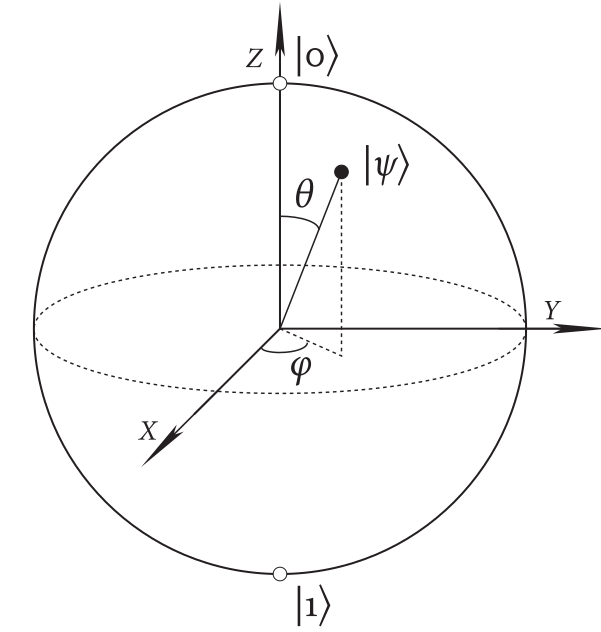}
\caption{\label{fig:i}(\cite{Wil13}) The Bloch sphere.}
\label{fig:bloch}
\end{figure}\ \\
For the sake of brevity, we often write $\ket{\psi} \ket{\phi} = \ket{\psi}_A \otimes \ket{\phi}_B$. 
Furthermore, we shall also frequently adopt the notation $\ket{00}$ instead of $\ket{0}\ket{0}$, as well as $\ket{01}$, $\ket{10}$ and $\ket{11}$.
This allows us to conveniently represent $\ket{\psi} \ket{\phi}$ using a decimal instead of a binary expression:
\begin{equation}
\alpha_0 \ket{0} + \alpha_1 \ket{1} + \alpha_2 \ket{2} + \alpha_3 \ket{3}.
\end{equation}
In general, a collection of $n$ qubits forms a \textit{register} of size $n$:
\begin{equation}
\ket{\Psi} = \sum_{x \in \{0,1\}^n} \alpha_x \ket{x_1} \ket{x_2}... \ket{x_n},
\end{equation}
where, due to normalization, we require $\sum_x |\alpha_x|^2 = 1$.
Equivalently, we can also consider the above as a superposition of $2^n$ different states in a decimal expression:
\begin{equation}
\ket{\Psi} = \sum_{x=0}^{2^n-1} \alpha_x \ket{x}.
\end{equation}
Excluding the overall phase, the description of an $n$-qubit state already requires an enormous amount of $2^n -1$ many complex numbers, 
as $\mathcal{H}_2^n \cong \mathbbm{C}^{2n}$.
This fact can be exploited in quantum parallelism, which we discuss in the subsequent chapters.
More generally, for $d \geq 2$, it is also useful to consider \textit{qudits}, a quantum system of computational states $\ket{0},\ket{1}, ..., \ket{d-1}$ in a 
register of size $n$:
\begin{equation}
\ket{\Psi} = \sum_{x\in \mathbb{Z}_d^n} \alpha_x \ket{x_1} \ket{x_2}... \ket{x_n}.
\end{equation}
Similarly, by adopting a decimal expression, we can write:
\begin{equation}
\ket{\Psi} = \sum_{x=0}^{d^n-1} \alpha_x \ket{x}.
\end{equation}
In this case, the computational space $\mathcal{H}_d^n \cong \mathbbm{C}^{dn}$ features an enormous amount of $q^n -1$ different states.
The use of qudits is particularly useful in the context of the \LWE problem of the later sections.
While qudits are certainly more difficult to realize in a physical system, they can easily be emulated with qubits by using a block encoding in
which each qudit is packed into $\left \lceil{\log_2(d)}\right \rceil$ many qubits.

A quantum system with a well-defined state vector $\ket{\psi}$ in $\mathcal{H}$ is said to be \textit{pure}.
The most general state of a quantum system, however, is a \textit{mixed} state described by a \textit{density operator} $\rho \in \mathcal{D}(\mathcal{H})$, 
the set of positive semidefinite Hermitian matrices of unit trace. 
We can interpret the density operator as a statistical ensemble of pure states $\ket{\psi_i}$, 
where $\sum_i p_i = 1$, $p_i \geq 0$ and
\begin{equation}
\rho = \sum_{i} p_i \ket{\psi_i}\bra{\psi_i}.
\end{equation}
If $\rho$ is pure, then $\rho$ has rank $1$ and we can conveniently write $\rho = \ket{\psi}\bra{\psi}$. 
Furthermore, we can distinguish between pure and mixed states by using the fact that 
tr$(\rho^2) = 1$, if and only if $\rho$ is pure, whereas tr$(\rho^2) < 1$, if and only if $\rho$ is mixed.

\subsection{Unitary Evolution}\label{ch:unitary}

In the previous section, we introduced the concept of a \textit{qubit}, a quantum system $\ket{\psi}$ described by a continuous superposition of states $\ket{0}$ and $\ket{1}$. 
Computation, understood as simply the manipulation of encoded information such as bits, requires a notion of what transformations are possible within
a certain model of computation. Just as in Turing's abstract model of computation, it is necessary to define a model together
with a set of rules on how to operate symbols stored on the equivalent of a tape by a set of instructions. 
In order to define what computation means in the quantum model of computation, we require one of the postulates of quantum mechanics:\\
The time evolution of a closed quantum system is governed by the \textit{Schr{\"o}dinger equation},
\begin{align}
i \hbar \frac{d\ket{\psi} }{dt } = \mathcal{H} \ket{\psi},
\end{align}
where $\hbar$ is \textit{Planck's constant} and $\mathcal{H}$ is the Hamiltonian operator of the system.
If the Hamiltonian is time-independent, the Schr{\"o}dinger equation gives rise to the following dynamics of the state vector:
\begin{equation}
\label{unitary_evolution}
\ket{\psi(t)} = \exp{\left(\frac{-i \mathcal{H} t}{\hbar}\right)} \ket{\psi(0)}.
\end{equation}
The associated time-evolution operator,
\begin{align}
U = \exp{\left(\frac{-i \mathcal{H} t}{\hbar}\right)},
\end{align}
is a \textit{unitary} evolution operator, i.e. a norm-preserving operation that satisfies $U^\dagger U = \mathbbm{1}$ such that
\begin{equation}
\braket{\psi(t) | \psi(t)} = \bra{\psi(t)}U^\dagger U \ket{\psi(t)} = \braket{\psi(0) | \psi(0)} = 1.
\end{equation}
Consequently, we can also write the unitary evolution of a density operator as
\begin{equation}
\rho(t) = \sum_{i} p_i \ket{\psi_i(t)}\bra{\psi_i(t)}= \sum_{i} p_i \, U\ket{\psi_i(0)}\bra{\psi_i(0)}U^\dagger = U \rho(0)\,  U^\dagger.
\end{equation}
Since an ideal qubit is required to be a closed quantum system, any unitary time-evolution describing a computation
corresponds to a rotation on the Bloch sphere.
Furthermore, the time-evolution of a quantum system under a given stationary Hamiltonian is reversible through its Hermitian adjoint $U^\dagger$.
Consequently, all unitary quantum gates must be inherently reversible. 
As we discuss in the next sections, this fact has important consequences for many elementary operations.

\subsection{Quantum Measurement}

The measurement postulate of quantum mechanics specifies how information is retrieved in the quantum model of computation. 
Thus, in accordance with the laws of quantum mechanics, a measurement of a quantum state translates into classical measurement outcomes
according to a set of rules.
In this section, we highlight the most relevant notions of measurement required for
the work contained in this thesis.

Quantum measurements are described by a set of \textit{measurement operators} $\{M_m\}$ acting on the state space of a given 
system. These operators 
obey a completeness relation $\sum_m M_m^\dagger M_m = \mathbbm{1}$, where
$m$ labels the measurement outcome of the associated measurement operator.
Let $\ket{\psi}$ be the state vector of the system prior to measurement.
Then, the probability that outcome $m$ occurs is:
\begin{align}
p(m) = \bra{\psi} M_m^\dagger M_m \ket{\psi}.
\end{align}
The post-measurement state is subsequently renormalized and given by:
\begin{align}
\frac{ M_m \ket{\psi}}{\sqrt{\bra{\psi} M_m^\dagger M_m \ket{\psi}}}.
\end{align}
For example, given the qubit from the previous sections,
\begin{equation}
\ket{\psi} = \alpha \ket{0} + \beta \ket{1},
\end{equation}
a measurement in the \textit{computational basis} is defined by two measurement operators, where $M_0 = \ket{0}\bra{0}$ and $M_1 = \ket{1}\bra{1}$.
Each measurement operator is Hermitian, since $M_0^2=M_0$ and $M_1^2=M_1$, so that the completeness relation is obeyed.
The probabilities of the respective outcomes are given by:
\begin{align}
p(0) &= \bra{\psi} M_0^\dagger M_0 \ket{\psi} =\braket{\psi | 0} \braket{0 | \psi} = |\alpha|^2\\
p(1) &= \bra{\psi} M_1^\dagger M_1 \ket{\psi} = \braket{\psi | 1} \braket{1 | \psi}  = |\beta|^2.
\end{align}
Consequently, a measurement results in $\ket{0}$ with probability $|\alpha|^2$, and $\ket{1}$ with probability $|\beta|^2$.
This brings us to a special case of measurements, the class of \textit{projective measurements}. Here, the measurement operators are given
by hermitian operators $\{P_m\}$, so-called \textit{projectors}, that obey a completeness relation $\sum_m P_m = \mathbbm{1}$ and satisfy $P_n P_m = \delta_{n,m} P_m$.\\
The probability to observe the outcome $m$ is given by:
\begin{align}
p(m) = \bra{\psi} P_m \ket{\psi},
\end{align}
whereas the post-measurement state is
\begin{align}
\frac{ P_m \ket{\psi}}{\sqrt{\bra{\psi} P_m \ket{\psi}}}.
\end{align}
A final class of more general measurements we consider is that of \POVM \textit{measurements} (Positive-Operator-Valued Measure) \cite{NC10}, 
where the post-measurement state is of little
interest and the concern lies on the outcome probabilities corresponding to a set of measurement operators.
In this context, a set of complete positive semidefinite measurement operators $\{E_m\}$ is employed
such that $\sum_m E_m = \mathbbm{1}$.

\subsection{Universal Quantum Gates}\label{ch:universal}

In this section we introduce elementary quantum gates, in particular those that allow for universal quantum computation.
In \expref{Section}{ch:unitary}, we observed that all quantum gates must correspond to unitary transformations, and are thus inherently reversible.
While classical universality of logic gates is achieved by using only a NAND gate, we require a certain universal set of
at least three elementary gates for quantum computation.

Let us begin with a few examples of single-qubit quantum gates.
A simple set of single-qubit gates are the $X,Y,Z$-gates, resembling the \textit{Pauli matrices} $\sigma_x,\sigma_y$ and $\sigma_z$:
\begin{equation}
X =
  \begin{bmatrix}
  0\, &\hphantom{ll} 1\\
  1\, &\hphantom{ll} 0
  \end{bmatrix}
\phantom{kkkk}
Y = 
  \begin{bmatrix}
  0\, &\hphantom{} -i\\
  i\, &\hphantom{-} 0
  \end{bmatrix}
\phantom{kkkk}
Z =
  \begin{bmatrix}
  1\, &\hphantom{-} 0\\
  0\, &\hphantom{} -1
  \end{bmatrix}
\end{equation}
Consider, for example, a qubit $\ket{\psi} = \alpha \ket{0} + \beta \ket{1}$. In vector representation, we compute the action of the $X$-gate as follows:
\begin{equation}
X  \ket{\psi} =
  \begin{bmatrix}
  \alpha\\
  \beta
  \end{bmatrix}
\cdot 
\begin{bmatrix}
  0\, &\hphantom{ll} 1\\
  1\, &\hphantom{ll} 0
  \end{bmatrix}
  = \beta \ket{0} + \alpha \ket{1}.
\end{equation}
One of the most frequently used operations in quantum computing is that of the \textit{Hadamard gate}, which is given by
\begin{equation}
H = \frac{1}{\sqrt{2}}
  \begin{bmatrix}
  1\, &\hphantom{-} 1\\
  1\, &\hphantom{}-1
  \end{bmatrix}.
\end{equation}
The Hadamard gate, often described as a square root of the $X$-gate, completes only half of a $180^\circ$ rotation on the Bloch sphere
and maps the basis states onto an equal superposition and back:
\begin{equation}
\ket{0} \overset{H}{\longleftrightarrow} \, \, \ket{+} = \frac{\ket{0}+\ket{1} }{\sqrt{2} } \hspace{15mm} \ket{1} \overset{H}{\longleftrightarrow} \, \, \ket{-} = \frac{\ket{0}-\ket{1} }{\sqrt{2} }
\end{equation}
Another important single-qubit gate is the \textit{phase-shift gate} in which $\phi$ denotes the angle of rotation. The special
case where $\phi = \pi/4$ is often referred to as the T-gate:
\begin{equation}
\phi = 
  \begin{bmatrix}
  1\, &\hphantom{i\phi} 0\\
  0\, &\hphantom{\phi}e^{i\phi}
  \end{bmatrix},
\quad \quad \quad
T = 
  \begin{bmatrix}
  1\, &\hphantom{i\phi} 0\\
  0\, &\hphantom{\phi}e^{i \pi/4}
  \end{bmatrix}.
\end{equation}
In addition, we consider the rotation operators around the $x,y$ and $z$ axis:
\begin{equation}
R_x(\theta) =
\begin{bmatrix} 
  \cos \frac{\theta}{2}     & \, -i \sin\frac{\theta}{2}\\ 
  -i \sin\frac{\theta}{2}  & \cos\frac{\theta}{2}
\end{bmatrix}
\phantom{kk}
R_y(\theta) =
\begin{bmatrix} 
  \cos\frac{\theta}{2}  & \, -\sin\frac{\theta}{2}\\ 
  \sin\frac{\theta}{2}  & \cos\frac{\theta}{2}
\end{bmatrix}
\phantom{kk}
R_z(\theta) =
\begin{bmatrix} 
  e^{-i \frac{\theta}{2}}  & \, 0\\ 
  0  & \, e^{i \frac{\theta}{2}}
\end{bmatrix}.
\end{equation}
In fact, \textit{any} unitary single-qubit operation $U$ can be decomposed using the rotation operators above. 

\begin{theorem}[\cite{NC10}, Theorem $4.10$]\label{th:single}\ \\
The two rotation operations $R_x$ and $R_y$ comprise a basis for all single-qubit operations: For every $2 \times 2$ unitary operation $U$,
there exist real numbers $\alpha,\beta,\gamma$ and $\delta$ such that:
$$ U = e^{i \alpha} R_x(\beta)R_y(\gamma)R_x(\delta).$$
\end{theorem}
Let us now conclude our discussion on quantum gates with two-qubit gates, perhaps
the most striking class of operations found in quantum computers. Early work by Deutsch, Eckert and Josza suggests
that this class of gates is precisely the set of operations that entangles qubits with one another, thereby providing the 
foundation for most quantum computations.
The most important two-qubit gate is the \textit{controlled-NOT} (CNOT) gate, an operation that performs a bit flip on a target bit if and only if 
the control qubit is $\ket{0}$.
Another important gate is the Toffoli (CCNOT) gate, a three-qubit gate that flips the last qubit only if and only if all three inputs correspond to $\ket{1}$.
The matrix representations are given by:
\begin{equation}
\text{CNOT } = 
  \begin{bmatrix}
  1 &\,\, 0 &\,\, 0 &\,\, 0\\
  0 & 1 & 0 & 0\\
  0 & 0 & 0 & 1\\
  0 & 0 & 1 & 0
  \end{bmatrix},
\quad \quad
\text{Toffoli } = 
  \begin{bmatrix}
  1 & 0 & 0 & 0 & 0 & 0 & 0 & 0 & 0\\
  0 & 1 & 0 & 0 & 0 & 0 & 0 & 0 & 0\\
  0 & 0 & 1 & 0 & 0 & 0 & 0 & 0 & 0\\
  0 & 0 & 0 & 1 & 0 & 0 & 0 & 0 & 0\\
  0 & 0 & 0 & 0 & 1 & 0 & 0 & 0 & 0\\
  0 & 0 & 0 & 0 & 0 & 1 & 0 & 0 & 0\\
  0 & 0 & 0 & 0 & 0 & 0 & 1 & 0 & 0\\
  0 & 0 & 0 & 0 & 0 & 0 & 0 & 0 & 1\\
  0 & 0 & 0 & 0 & 0 & 0 & 0 & 1 & 0\\
  \end{bmatrix}
  \phantom{kkkkkk}
\end{equation}
Equivalently, these two-qubit and three-qubit gates can also be written as the following operations:
\begin{align}
&\text{CNOT: } \ket{x}\ket{y} \longrightarrow \ket{x}\ket{x \oplus y}\\
&\text{Toffoli: } \ket{x}\ket{y} \ket{z} \longrightarrow \ket{x}\ket{y}\ket{z \oplus x \wedge y },
\end{align}
where $\oplus$ denotes addition modulo 2 and $\wedge$ denotes the AND operation (\expref{Table}{AND_NAND}).
Finally, we also consider the \textit{controlled-Z} (CZ) gate, an operation that features an additional minus sign and has the following matrix representation:
\begin{equation}
\text{CZ } = 
  \begin{bmatrix}
  1 &\,\, 0 &\,\, 0 &\,\, 0\\
  0 &\,\, 1 &\,\, 0 &\,\, 0\\
  0 &\,\, 0 &\,\, 1 &\,\, 0\\
  0 &\,\, 0 &\,\, 0 &\,  -1
  \end{bmatrix}
\end{equation}
The following theorem states ensures that quantum computation is indeed universal using only a limited set of gates.

\begin{theorem}[\cite{Deu89} \cite{Kit97}, Universal set of quantum operations] \label{thm:universality} \ \\
The Hadamard gate, the Toffoli gate and phase-shift gate comprise a universal basis for any quantum operation:
For every $D\geq3$, there exists $l\leq 100(D \log\frac{1}{\epsilon})^3$ such that every unitary $D \times D$ matrix $U$
can be approximated by a sequence of the above unitary gates to $\epsilon > 0$ degree of accuracy:
$$ | U_{i,j} - (U_l \cdot \cdot \cdot U_1)_{i,j} | \, < \, \epsilon,$$
where the index $(i,j)$ denotes the entries of the respective matrices.
\end{theorem}
While a basis for quantum computation is not limited to precisely the set as given in \expref{Theorem}{thm:universality},
it nevertheless provides a convenient choice of elementary gates.

\subsection{The Quantum Circuit Model}

A quantum computation typically starts out in some initial state $\ket{00...0}$, then performs a sequence
of single and two-qubit gates and finally ends with a measurement. The \textit{quantum circuit model} provides us with
a convenient way of representing any computation pictorially. 
\begin{figure}
	\centering
	\begin{subfigure}[t]{1.5in}
		\centering
		\includegraphics[width=8mm]{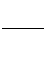}
		\caption{quantum wire}\label{fig:1a}		
	\end{subfigure}
	\quad 
	\begin{subfigure}[t]{1.5in}
		\centering
		\includegraphics[width=16mm]{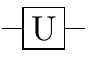}
		\caption{unitary gate}\label{fig:1b}
	\end{subfigure}
		\quad 
	\begin{subfigure}[t]{1.5in}
		\centering
		\includegraphics[width=16mm]{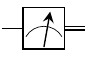}
		\caption{measurement}\label{fig:1b}
	\end{subfigure}
	\caption{The quantum circuit model. Quantum wires (a) represent the history of single qubit in time from left to right.
	Single-qubit unitary gates (b) act on a single wire. Measurements (c) of qubits are given in the computational basis. }\label{fig:1}
\end{figure}
\begin{figure}
	\centering
	\begin{subfigure}[t]{1.5in}
		\centering
		\includegraphics[width=30mm]{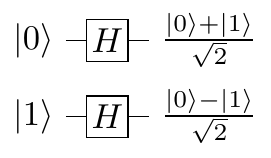}
		\caption{Hadamard gates}\label{fig:1a}		
	\end{subfigure}
	\quad \,\,\,
	\begin{subfigure}[t]{1.5in}
		\centering
		\includegraphics[width=30mm]{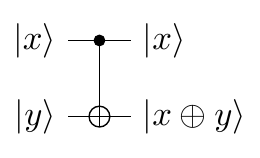}
		\caption{CNOT-gate}\label{fig:1b}
	\end{subfigure}
	\quad \,
	\begin{subfigure}[t]{1.5in}
		\centering
		\includegraphics[width=25.3mm]{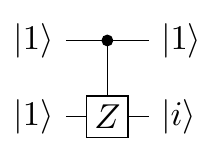}
		\caption{CZ-gate}\label{fig:1b}
	\end{subfigure}
	\caption{Single-qubit gates and two-qubit gates. The Hadamard gates (a) each act on a single quantum wire.
	CNOT (b) adds the value of the control qubit into the target qubit. CZ (c) performs a phase flip only if both the
	control and target qubit are $\ket{1}$.
	}\label{fig:gates}
\end{figure}
\begin{figure}[ht!]
\centering
\includegraphics[width=80mm]{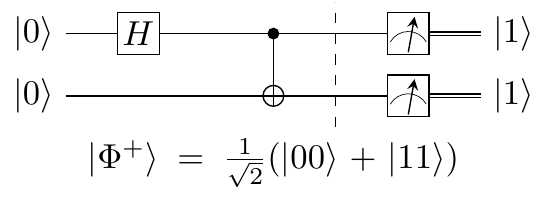}
\caption{A quantum circuit that prepares an entangled state $\ket{\Phi^+}$ (as indicated by the dashed line), the famous Einstein-Podolsky-Rosen (EPR)-pair $\ket{\Phi^+} = \frac{1}{\sqrt{2}}(\ket{00} + \ket{11})$. 
Here, a final measurement in the computational basis results in the outcome $\ket{1}\ket{1}$.}\label{fig:1}
\end{figure}

\subsection{Quantum Parallelism}

In the previous sections, we recognized that unitary quantum gates are inherently reversible. Can we, nevertheless, still 
simulate classical computation using only
reversible gates? Consider for example, the classical NAND gate, as shown in \expref{Table}{AND_NAND}.
The NAND-gate, a universal logic gate for classical computation, is inherently irreversible. 
Knowing that the output is $1$, we cannot conclude with certainty whether the input was in fact $00,01$ or $10$.
More generally, consider the problem of computing the following transformation:
\begin{equation}
x \overset{f}{\longrightarrow} f(x)
\end{equation}
If $f$ is a bijective operation, we can reverse this transformation and recover the input. However, if $f$ is irreversible, we can attach the input
and still achieve an overall reversible operation, as follows:
\begin{equation}
(x,y) \overset{f}{\longrightarrow} (x,y \oplus f(x)).
\end{equation}
Note that, when performing this operation twice, we obtain the original input pair.
A well known trick attributed to Charles Bennet allows us to compute irreversible transformations using only reversible quantum gates at the expense of a 
few additional registers. 
By simply attaching additional input registers prior to the evaluation of the function, a reversible unitary transformation $U_f$ is possible in which
later registers can be uncomputed. As a result, this allows us to define operations, such as:
\begin{equation}
\label{oracle_f}
\ket{x}\ket{y} \overset{U_f}{\longrightarrow} \ket{x}\ket{y \oplus f(x)},
\end{equation}
where $U_f^\dagger U_f = \mathbbm{1}$.
This subsumes an evaluation of $f$, as we can initialize the second qubit to $y=0$ and compute the output of $f$ as follows:
\begin{equation}
\ket{x}\ket{0} \overset{U_f}{\longrightarrow} \ket{x}\ket{f(x)}
\end{equation}
\begin{table}[tbp]
\center
\begin{tabular}{|lr|c|}
\hline
\,\,\,\,Inputs: && Outputs:\\
\hline 
 \,A \, & \, B \,& A \,\textbf{AND}\, B\, \\
\hline 
 \,0 \, & \, 0 \,& 0 \,\\
 \,1 \, & \, 0 \,& 0 \,\\
 \,0 \, & \, 1 \,& 0 \,\\
 \,1 \, & \, 1 \,& 1 \,\\
\hline
\end{tabular}
\quad \quad \quad \quad
\begin{tabular}{|lr|c|}
\hline
\,\,\,\, Inputs: && Outputs:\\
\hline 
 \,A \, & \, B \,& A \,\textbf{NAND}\, B\, \\
\hline 
 \,0 \, & \, 0 \,& 1 \,\\
 \,1 \, & \, 0 \,& 1 \,\\
 \,0 \, & \, 1 \,& 1 \,\\
 \,1 \, & \, 1 \,& 0 \,\\
\hline
\end{tabular}
\caption{\label{tab:i} Two classical logic gates: AND$(\wedge)$ and NAND$(\uparrow)$.}
\label{AND_NAND}
\end{table}
One of the essential features of quantum algorithms is quantum parallelism, the ability to prepare a superposition of states for simultaneous evaluation.
Consider a simple Boolean function $f:\{0,1\} \rightarrow \{0,1\}$ and suppose we have access to a unitary gate
that evaluates $f$ onto the presented inputs, such as in \eqref{oracle_f}. By preparing two initial states $\ket{0}\ket{0}$ and
applying a single Hadamard gate onto the first register, we can evaluate $f$ and exploit quantum parallelism (\expref{Figure}{fig:parallelism}).
\begin{figure}[ht!]
\centering
\includegraphics[width=60mm]{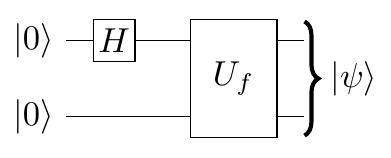}
\caption{A quantum circuit that prepares a uniform superposition $\ket{\psi}$ which simultaneously evaluates $f$ on both $0$ and $1$. }\label{fig:parallelism}
\end{figure}
The result of such a transformation is an output state:
\begin{equation}
 \ket{\psi} = \frac{1}{\sqrt{2}}\sum_{x \in \{0,1\}} \ket{x} \ket{f(x)} = \frac{1}{\sqrt{2}}(\ket{0}\ket{f(0)}+\ket{1}\ket{f(1)}).
\end{equation}
Remarkably, the transformation achieves a superposition that contains information on both $f(0)$ and $f(1)$. 
Simultaneous evaluation of a function is what gives power to quantum parallelism: a single quantum evaluation of a function can
result in a state that features evaluations of $f$ in superposition. Research on quantum algortithms involves sophisticated techniques that exploit
such information hidden in the superposition to one's advantage.

\subsection{Decoherence}\label{sec:quantum_noise}

As with any physical implementation of a computing device, not all operations can be done perfectly and there remains an unavoidable risk of error.
Unlike the closed quantum system evolution we discussed in \expref{Section}{ch:unitary}, all physically realizable quantum systems evolve as an open system
in constant interaction with its environment. In fact, due to coupling and entanglement with the environment, artificially engineered closed quantums system tend to lose their \textit{quantumness}
and become intrinsically random. This process is commonly known as \textit{decoherence}.   

Typically, one distinguishes between two classes of errors. 
We encounter both \textit{memory errors} that occur during storage of information,
as well as \textit{operational errors} that occur during manipulation of stored information. In the section on error correcting codes,
we provide further discussion on how to correct for such errors.

\subsubsection{Quantum Noise Models.}
In order to develop a successful noise model, we follow \cite{NC10} and adopt a theory of quantum channels and the operator-sum-representation.
In this framework, we consider noise models as \textit{discrete state changes} without reference to time.\\
For the remainder of the section, we let $\rho$ be a quantum system of computational states $\ket{0}$ and $\ket{1}$ and define the
action of a noisy quantum channel by a completely positive and trace preserving (\CPTP) operation $\mathcal{E}$, where for operation elements $E_0$ and $E_1$, we have:
\begin{equation}
\rho \longrightarrow \mathcal{E}(\rho) = E_0 \rho E_0^\dagger + E_1 \rho E_1^\dagger.
\end{equation}
A simple quantum noise channel is the bit-flip channel that, with probability $\eta>0$, maps the state $\ket{0}$ to the state $\ket{1}$ and vice-versa:
\begin{equation}
E_0 = \sqrt{1-\eta} \cdot 
  \begin{bmatrix}
  1\, &\hphantom{\phi} 0\\
  0\, &\hphantom{\phi} 1
  \end{bmatrix},
  \quad \quad
E_1 = \sqrt{\eta} \cdot 
  \begin{bmatrix}
  0\, &\hphantom{\phi} 1\\
  1\, &\hphantom{\phi} 0
  \end{bmatrix}.
\end{equation}
Thus, in the operator-sum-representation, we can now write the bit-flip channel as:
\begin{equation}
\label{bit_flip_channel}
\rho \longrightarrow \mathcal{E}_X(\rho) = (1-\eta) \rho + \eta \, X\rho X^\dagger,
\end{equation}
where $X$ corresponds to the bit-flip gate from the earlier sections.\par
Consider now the case of a single qubit, a quantum system that starts out in a pure state $\rho = \ket{\psi}\bra{\psi}$, where 
$\ket{\psi} = \alpha \ket{0} + \beta \ket{1}$. 
Under the bit-flip channel $\mathcal{E}$, the state evolves into a statistical mixture according to \eqref{bit_flip_channel}. 
Therefore, with probability $1-\eta$, we recover the original state,
\begin{equation}
\ket{\psi} = \alpha \ket{0} + \beta \ket{1},
\end{equation}
and, with probability $\eta$, we find the system in a state:
\begin{equation}
\ket{\psi^\bot} = \beta \ket{0} + \alpha \ket{1}.
\end{equation}
This noise model is often called \textit{classification noise} and will be highly relevant for the quantum learning algorithms of the later sections.\par
Another elementary quantum noise channel is the phase-flip channel that, with probability $\eta>0$, maps the state $\ket{1}$ to the state $-\ket{1}$, where
\begin{equation}
E_0 = \sqrt{1-\eta} \cdot 
  \begin{bmatrix}
  1\, &\hphantom{\phi} 0\\
  0\, &\hphantom{\phi} 1
  \end{bmatrix},
  \quad \quad
E_1 = \sqrt{\eta} \cdot 
  \begin{bmatrix}
  1\, &\hphantom{-\phi} 0\\
  0\, &\hphantom{\phi} -1
  \end{bmatrix}.
\end{equation}
Thus, in the operator-sum-representation, we can also write the phase-flip channel as:
\begin{equation}
\label{phase_flip_channel}
\rho \longrightarrow \mathcal{E}_Z(\rho) = (1-\eta) \rho + \eta \, Z\rho Z^\dagger,
\end{equation}
where $Z$ corresponds to the phase-flip gate from the earlier sections.\par
A much more ideal noise model for quantum computation is the amplitude damping channel. 
This noise channel $\mathcal{E}_{AD}$ corresponds to a scenario in which a photon is spontaneously emitted
with some probability $\gamma$.
\begin{equation}
\label{ampl_channel}
\rho \longrightarrow \mathcal{E}_{AD}(\rho) = E_0 \rho E_0^\dagger + E_1 \rho E_1^\dagger,
\end{equation}
where the transition matrix operators are given by:
\begin{equation}
E_0 =
  \begin{bmatrix}
  1\, &\hphantom{\phi} 0\\
  0\, &\hphantom{\phi} \sqrt{1-\gamma}
  \end{bmatrix},
  \quad \quad
E_1 = 
  \begin{bmatrix}
  1\, &\hphantom{\phi} \sqrt{\gamma}\\
  0\, &\hphantom{\sqrt{\gamma}} 0
  \end{bmatrix}.
\end{equation}
Thus, $E_1$ corresponds to the emission of a photon and a quantum of energy
is lost to the environment. The operator $E_0$, corresponds to the case where the state remains unchanged and
a photon is not yet lost but the amplitudes are adjusted appropriately.

Finally, we consider the $\textit{depolarizing channel}$, a devastating type of noise model in which all quantum information is lost to the environment and 
the quantum state gets replaced by a maximally mixed state with some probability $\eta$:
\begin{equation}
\label{ampl_channel}
\rho \longrightarrow \mathcal{E}_{D}(\rho) = (1-\eta) \rho + \eta \frac{\mathbbm{1}}{2},
\end{equation}
where $\frac{\mathbbm{1}}{2}$ is the maximally mixed state.

\subsubsection{Independent Noise Models.}

The most commonly adopted noise model in computational learning theory is that of independent noise~\cite{BJ95}\cite{GK17}.
In this work, we specifically focus on noise that corrupts the final output register that
evaluates the input registers for a choice of function. 
As in the \LWE probblem, the final register is independently corrupted over all input registers by some probability,
for example when receiving a quantum state over possible outcomes over an arbitrary set $X$, such as
\begin{equation}
\frac{1}{\sqrt{|X|}}\sum_{x \in X} \ket{x}\ket{f(x)}.
\end{equation}
We can quantum mechanically generate noisy samples towards an outcome $f(x) + e(x,r)$, where
$x \in X$ and $r \in R$, for some set $R$ of randomness, by simply tracing out
the randomness register of a superposition
\begin{equation}
\frac{1}{\sqrt{|R|}}\sum_{r \in R} \ket{r} \left(\frac{1}{\sqrt{|X|}}\sum_{x \in X} \ket{x} \ket{f(x)+ e(x,r)}\right).
\end{equation}
This allows us to generate noisy samples, where with probability $1/|R|$ over all values $r \in R$,
\begin{equation}
\frac{1}{\sqrt{|X|}}\sum_{x \in X} \ket{x}\ket{f(x) + e(x,r)}.
\end{equation}
In the context of qubits, we consider independent noise by extending the noise models from 
this section onto quantum registers. 
Consider a register of $n$ qubits $\ket{\Psi}$, where
\begin{equation}
\ket{\Psi} = \sum_{x \in \{0,1\}^n} \alpha_x \ket{x_1} \ket{x_2}... \ket{x_n}.
\end{equation}
For example, by extending the bit-flip channel introduced in this chapter independently onto the final register, the result is a state
\begin{equation}
\ket{\tilde \Psi} = \sum_{x \in \{0,1\}^n} \alpha_x \ket{x_1}\dots \ket{x_n}\ket{f(x) \oplus e_x},
\end{equation}
where the error $e$ is sampled from a Bernoulli distribution of noise rate $\eta > 0$:
\begin{equation}
Bern(\eta) =
    \begin{cases}
     1, &  \text{with probability }  \eta\\
     0, & \text{with probability }  1-\eta.
    \end{cases}
\end{equation}
Note that, upon a choice of error distribution, an independent noise model also translates naturally in the context of qudits instead of qubits.

\subsection{Error Correcting Codes}\label{sec:error_corr}

In his seminal $1948$ paper \textit{A Mathematical Theory of Communication} \cite{Sha48}, Claude Shannon put forward a revolutionary view on the concept
information and errors in communication. Instead of investing tedious effort to avoid them on a technical level, it is not only possible,
but oftentimes even favorable, to simply correct them. In the 1950s, John von Neumann developed very successful error correcting codes 
in order to address noise originating in the relay architecture of present computer technology. Today's transistors achieve near-perfect fault tolerance,
hence error correcting codes often do not even have to be applied.
In a nutshell, by adding additional redundant information to each bit of information, one can realize error correction. 

Suppose the task is to store a single
bit of information for some desired time interval $T \geq 0$. Regardless of whether an operation takes place, we consider \textit{memory errors} 
that can occur spontaneously.
We denote the probability that such an error occurs after time $T$ by $p$. A simple way to protect the storage of information against these affects 
is by adding redundance. For example, consider the following three copies of each instance of the bit:
\begin{align}
 0 &\longrightarrow 000\\
 1 &\longrightarrow 111.
\end{align}
The probability that there are no errors is given by $(1-p)^3$. Hence, after time $T$, the three bits $000$ remain the same.
The probability that there is an error in one of the bits is $3p(1-p)^2$, resulting in either $001$, $010$ or $100$.
Finally, the probability that there are two or more errors is $3p^2(1-p)+p^3$. The error correction scheme now works as follows.
By measuring all the bits and taking the majority vote, we can easily rule out single bit errors. Our correction method thus assigns
the measurement outcomes of the bitstrings as follows:
\begin{align}
\{000, 001, 010, 100\} &\longrightarrow 0\\
\{111, 110, 101, 011\} &\longrightarrow 1.
\end{align}
The error correction method above is correct with probability $p_c = 1 - 3p^2+2p^3$.
Compared to the previous error probability of $p$, we gain as long as $p_c \geq 1-p$. This translates into
an error probability of at most $p < \frac{1}{2}$. By dividing the time interval into slices $\Delta t = T/N$ and applying the error correcting
code repeatedly after reach slice, one can reach arbitrarily close success probabilities which grow in $N$ \cite{NC10}\cite{CZ01}.

Consider now the problem of quantum error correction by analogy to classical error correction.
Due to the quantum nature of information, a new framework is needed to correct for errors. This becomes apparent as we consider a number of
differences as compared to classical error correction. According to the \textit{no-cloning theorem}, there is no machine that can make
a copy of an unknown quantum state. Therefore, the naive attempt of simply copying quantum information in order to achieve redundancy is not possible.
Moreover, a signifcant feature of quantum states is that the parameters describing them are continuous.
Consequently, quantum noise is also continuous and requires correction to reach up to infinite precision, hence demanding unbounded recources. 
And finally, classical error correction requires read-out or state detection of the bit sequence in order to detect errors and correct them.
In particular, measuring a quantum state generally destroys the state and makes recovery impossible.
The \textit{quantum fault-tolerant threshold theorem} \cite{Pre98}\cite{AB08} states that, as long as the noise level is below a 
certain threshold (typically around $10^{-4}-10^{-2}$),
any quantum computation can be performed with arbitrarily small error by adopting error correction. Most notably,
\textit{Shor's code} \cite{Sho95} achieves error correction for arbitrary single-qubit errors, including bit-flips and phase errors.
In the following, we give a simple example of a three qubit quantum error correcting code, known as the bit-flip code (see \cite{NC10}).

Suppose the task is to store a single qubit $\ket{\psi} = c_0 \ket{0} + c_1 \ket{1}$ in a state that is unknown to us.
Analogous to the classical case, we assume that after some time $T$
a bit flip occurs with probability $p$ such that the state $\ket{\psi}$ is taken to the corrputed state $X \ket{\psi}$. The 
bit flip error thus results in a new state $\ket{\psi^\bot} = c_0 \ket{1} + c_1 \ket{0}$.
First, we begin by preparing the following encoding into a sequence of three logical qubits
\begin{align}
\ket{0} &\longrightarrow \ket{0}_L \equiv \ket{000}\\
\ket{1} &\longrightarrow \ket{1}_L \equiv \ket{111}.
\end{align}
We can realize this encoding in a quantum state $\ket{\psi}_L = c_0 \ket{000} + c_1\ket{111}$.
A circuit that performs this operation upon $\ket{\psi}$ is given by:
\begin{figure}[ht!]
\centering 
\includegraphics[width=.26\textwidth,origin=c]{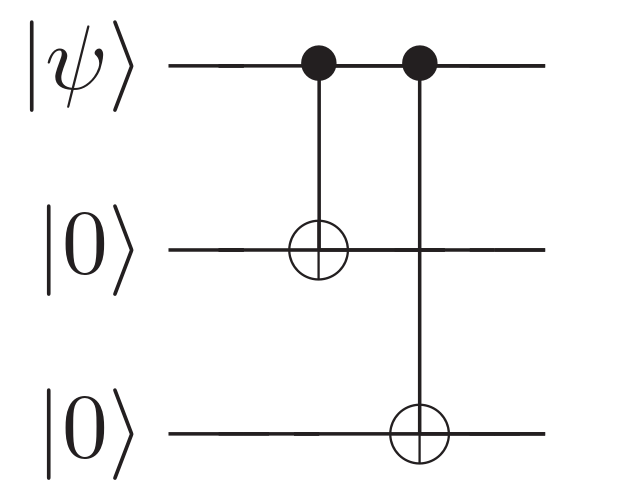}
\caption{\label{fig:i} (\cite{NC10}). A quantum circuit that prepares the state $c_0 \ket{000} + c_1\ket{111}$.}
\end{figure}

The probability that there are no errors in the state $\ket{\psi}_L$ is given by $(1-p)^3$. 
Hence, after time $T$, the three qubits remain the same.
The probability that there is an error in just one of the qubits is $3p(1-p)^2$, resulting in either $(X \otimes \mathbbm{1} \otimes \mathbbm{1}) \ket{\psi}_L$, 
$(\mathbbm{1} \otimes X \otimes \mathbbm{1}) \ket{\psi}_L$ or
$(\mathbbm{1} \otimes \mathbbm{1} \otimes X) \ket{\psi}_L$.
Finally, the probability that there are two or more errors is $3p^2(1-p)+p^3$.
The quantum error correcting code now works in two steps, as in the classical code from the previous section.
First, errors are being detected and then subsequently being corrected for in the second step using a recovery procedure.
To this end, consider the following sets of (incomplete) projection operators:
\begin{align}
P_0 &= \ket{000}\bra{000} + \ket{111}\bra{111}\\
P_1 &= \ket{100}\bra{100} + \ket{011}\bra{011}\\
P_2 &= \ket{010}\bra{010} + \ket{101}\bra{101}\\
P_3 &= \ket{001}\bra{001} + \ket{110}\bra{110}.
\end{align}
The first projection operator corresponds to the case where there is no error, and the other operators correspond
to a single bit-flip on one of the qubits, respectively.
In any error correcting code, the goal is to infer information on what error has occured without
destroying the superposition $\ket{\psi}$ altogether. The design of the code should therefore aim at projecting $\ket{\psi}_L$ into
mutually orthogonal spaces in which we can detect the type of the error and reversibly restore the original state
without at any point destroying the information.
We begin by first measuring the operators above for the state $\ket{\psi}_L$. Starting with
$P_0$, whenever we obtain $1$, we leave the state as it is knowing no error occured.
If we obtain $0$, we continue by measuring the next projector $P_1$. If we obtain $1$, we know that a bit flip occured on
the first qubit, and we correct by applying the $(X \otimes \mathbbm{1} \otimes \mathbbm{1})$ operation. Similarly, we can
correct for all single bit-flip errors as in the classical error correcting code and gain an advantage as long as $p < \frac{1}{2}$.

\subsection{Quantum Oracles}\label{section:oracles}

An important abstraction in computational complexity theory and the study of decision problems is the use of oracle machines, or \textit{oracles}.
First introduced in the context of Turing machines, oracles act as a black box that assist a Turing machine in a given computational task.
When presented an input $x$, such as an integer or a string, the oracle solves an instance of a decision problem at unit cost and returns
the output $\mathcal{O}(x)$ (typically a YES/NO answer) 
back to the Turing machine.
While this computational setting is certainly highly abstract, it does however contribute enormously to our understanding of complexity classes 
and is commonplace
in theoretical cryptography, particularly when providing arguments for security. 
Oracles also turn out to be highly useful in quantum computing, especially
in the study of quantum algorithms. Most notably, \textit{Grover's search algorithm} \cite{Gro96} relies
on the use of a quantum oracle that recognizes solutions to 
a given search problem. Many of the earliest quantum algorithms, such as the Deutsch-Josza algorithm \cite{DJ92}, the Bernstein-Vazirani algorithm \cite{BV93} 
or Simon's algorithm \cite{Sim97} are also devised in an oracle model.
In order to make an oracle evaluation reversible, we rely on the technique from the previous section.

For our purposes, a quantum oracle $\mathcal{O}$, is a unitary operation,
\begin{equation}
\ket{x} \ket{y} \longrightarrow \ket{x}\ket{y \oplus \mathcal{O}(x)},
\end{equation}
acting as a black box that can be accessed by a quantum computation. 
The inner workings of the quantum oracle are unknown to the computation, but can evaluate upon inputs in a reversible manner, see (\expref{Figure}{fig:quantum_oracle}).
\begin{figure}[ht!]
\centering
\includegraphics[width=60mm]{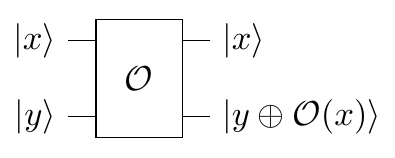}
\caption{Quantum oracle. }\label{fig:quantum_oracle}
\end{figure}
Consider now a collection of $n$-qubits, a quantum register of size $n$.
As a unitary gate, the oracle obeys the linearity of quantum mechanics and can therefore be queried on a superposition over all inputs:
\begin{equation}
\sum_{x,y \in \{0,1\}^n} \alpha_{x,y}\ket{x} \ket{y} \longrightarrow \sum_{x,y \in \{0,1\}^n} \alpha_{x,y} \ket{x}\ket{y \oplus \mathcal{O}(x)}
\end{equation}
More generally, we also further differentiate between two variants of oracles, so-called \textit{membership oracles} and \textit{example oracles}.

\subsubsection{Membership Oracles.}
In a membership oracle model, the oracle provides direct unitary input access to a function $f:\{0,1\}^{n} \longrightarrow \{0,1\}^m$ which is to be evaluated,
precisely as introduced before.
For example, upon input $x \in \{0,1\}^n$ and additional register $y \in \{0,1\}^m$, we define a membership oracle $\mathcal{O}_f$ 
for the function $f$ as an operation:
\begin{equation}
\mathcal{O}_f: \ket{x}\ket{y} \longrightarrow \ket{x}\ket{y \oplus f(x)}.
\end{equation}
The oracle can thus be queried at unit cost as a quantum gate at any step of a quantum computation. 

\begin{figure}[tbp]
	\centering
	\begin{subfigure}[t]{2in}
		\centering
		\includegraphics[width=60mm]{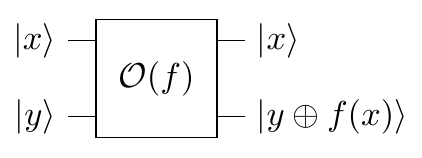}
		\caption{Quantum Membership Oracle}\label{fig:1a}		
	\end{subfigure}
	\quad \,\,\,
	\begin{subfigure}[t]{2in}
		\centering
		\includegraphics[width=48mm]{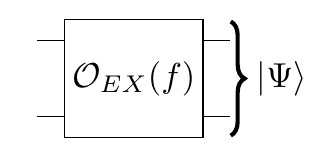}
		\caption{Quantum Example Oracle}\label{fig:1b}
	\end{subfigure}
		\caption{The quantum membership oracle (a) can be queried on arbitrary inputs during a quantum computation.
		The quantum example oracle (b) responds to queries by outputting a uniform superposition $\ket{\Psi}$ over all evaluations of the function.
	}\label{fig:gates}
\end{figure}\noindent

\subsubsection{Example Oracles.}
In the setting of computational learning theory from the subsequent chapters, we consider uniform example oracles $\mathcal{O}_{EX}(f)$ as
a black box that only outputs uniform samples for a given function. Thus, upon each query, the oracle replies with a state:
\begin{equation}
\ket{\Psi} \, = \frac{1}{\sqrt{2^n}}\sum_{x \in \{0,1\}^n} \ket{x}\ket{f(x)}.
\end{equation}
In the context of quantum example oracles, we also consider samples that are corrupted by noise. 
In particular, as the example oracle consists of a quantum circuit that
evaluates $f$ on a superposition of all inputs, this operation is naturally prone to errors. 
We consider two models of noise in this setting.

First, we consider uniform example states $\ket{\Psi}$ that suffer from a bit-flip error in the final \textit{result register} through a noise channel 
$\mathcal{E}_X$ of magnitude $\eta >0$, see \eqref{bit_flip_channel}.
Equivalently, we can also express the fact that $\ket{\Psi}$ is turning into a mixture 
by again sampling an error from a Bernoulli distribution of noise parameter $\eta$, resulting in a state
\begin{equation}
\ket{\Psi} \, = \frac{1}{\sqrt{2^n}}\sum_{x \in \{0,1\}^n} \ket{x_1} \ket{x_2}... \ket{x_n} \ket{f(x) \oplus e}.
\end{equation}
Moreover, in a model resembling the \LWE problem, we consider independent noise in the result register for each element in the superposition:
\begin{equation}
\ket{\Psi} = \sum_{x \in \{0,1\}^n} \alpha_x \ket{x_1} \ket{x_2}... \ket{x_n} \ket{f(x) \oplus e_x}.
\end{equation}
Upon a choice of error distribution, such an independent noise model also translates naturally in the context of qudits instead of qubits.

%% file: quantum_algorithms.tex
\clearpage
\section{Quantum Algorithms}
\label{quantum_algorithms}

Ever since the dawn of quantum computation, it was speculated that quantum computers could solve certain computational problems
faster than any conventional classical computer. Historically, the first abstract model of universal computation was proposed by 
Alan Turing in his seminal 1936 paper, a discovery
that henceforth greatly shaped the field of theoretical computer science. The \textit{Church-Turing thesis}
famously suggests that any model of computation appears at most as powerful as a Turing machine:\\

\noindent \textit{Any intuitively computable algorithmic process can be simulated efficiently by a Turing machine.}\\

Efficient algorithms are especially relevant to \textit{computational complexity} and concern
only computations of polynomial amounts of elementary operations, thus highlight the set of problems that can be solved with a \textit{feasible} 
use of computational recources. On the contrary, inefficient algorithms require superpolynomial amounts of recources (typically exponential)
and become computationally infeasible as the size of the problem increases.
The observation that the laws of physics are fundamentally quantum mechanical ultimately led David Deutsch to speculate on the prospect of computing devices
that behaved inherently quantum mechanical. Deutsch's insights into universal quantum computation \cite{Deu85} and the discovery
of the first quantum algorithm outperforming classical counterparts \cite{DJ92} provided unforseen challenges for the Church-Turing principle.
Shor's factoring algorithm \cite{Sho94} provided further evidence that quantum computers could indeed solve computational problems for which no efficient
classical algorithm is known. Still today, it is not clear whether a quantum model of computation is indeed capable of efficiently simulating any physical system
in nature, i.e. whether a quantum extension of the original thesis, the \textit{Quantum Church-Turing principle}, holds.

In this section, we review some of early quantum algorithms that offer substantial quantum speed-ups, 
such as the Deutsch-Josza algorithm \cite{DJ92}, the Bernstein-Vazirani algorithm \cite{BV93} 
and Simon's algorithm \cite{Sim97}, each providing the basis for the algorithms of the later sections.
In this thesis, we regard a quantum algorithm as a sequence of unitary operations, i.e. \textit{computations}, operating
on a product state space, for example $ \H = \H_{input} \otimes \H_{work} \otimes \H_{output}$, in analogy to the tape of a Turing machine.
Upon an input state $\ket{\psi_0} \in \H$, the quantum polynomial time (\QPT) algorithm runs an efficient quantum circuit 
consisting of a sequence of unitary operations:
\begin{align}
\ket{\psi} = U_T U_{T-1} U_{T-2} ...U_1 \, \ket{\psi_0},
\end{align}
where $T$ denotes a polynomial (in terms of the dimension of $\H$) amount of operations.
Hence, the algorithm generates an output state $\ket{\psi}$, typically followed by a measurement in the computational basis.
In this chapter, we consider problems in a membership oracle model, as discussed in Section \ref{section:oracles}.
Here, the goal of the algorithm is determine a secret property of a given function $f$ by making queries to an oracle at unit cost.
Any \QPT algorithm $\A$ can thus be written as:
\begin{align}
\ket{\psi} = U_T \mathcal{O}_f U_{T-1} \mathcal{O}_{f} U_{T-2} \dots U_2 \mathcal{O}_{f} U_1 \ket{\Psi_0}.
\end{align}
In order to highlight that $\A$ has quantum access to a membership oracle for a function $f$, we adopt the notation $\A^{\ket{f}}$.

\subsection{Deutsch-Josza Algorithm}

Let us now consider the Deutsch's problem, one of the earliest known problems to be solved
more efficiently by a quantum algorithm in a black box model, as appeared in \cite{DJ92}:\\

 \begin{mdframed}[backgroundcolor=green!5] 
\textbf{Deutsch's Problem:}\\
\textit{
Determine whether a Boolean function $f: \bit^n \rightarrow \bit$ is either constant or balanced, i.e.
$0$ for half of the inputs and $1$ else.}
\end{mdframed}
\begin{algorithm}
\label{deutsch_algorithm}
\caption{Deutsch-Josza Algorithm}
\begin{description}
\item[Input:] A quantum black box oracle $\mathcal{O}_{f}$ for a Boolean function $f$ which is either constant or balanced.
\item[Output:] Outcome $\ket{0^n}$ if and only if $f$ is constant, else $f$ is balanced.

\item[Procedure:]
\item[1.] Initialize $(n+1)$-qubits $\ket{0^n}\ket{1}$ and apply the Hadamard gate $H^{\oplus (n+1)}$:
\begin{equation*}
\longrightarrow  \frac{1}{\sqrt{2^n}} \sum_{x\in \{0,1\}^n} \ket{\vec x}\ket{-} \phantom{kklkkkkkkkklkkkk}
\end{equation*}
\item[2.] Query the quantum oracle, resulting in a \textit{phase-kickback}:
\begin{equation*}
\longrightarrow \frac{1}{\sqrt{2^n}} \displaystyle\sum_{x\in \{0,1\}^n} (-1)^{f(\vec x)} \ket{\vec x} \ket{-} \phantom{kklkkkkkk}
\end{equation*} 
\item[3.] Throw away the last register and then apply another Hadamard gate $H^{\otimes n}$:
\begin{equation*}
\longrightarrow \frac{1}{\sqrt{2^n}}\sum_{x,y \in \{0,1\}^n} (-1)^{\braket{\vec x,\vec y}} (-1)^{f(\vec x)}\ket{\vec y} \phantom{lkkkk}
\end{equation*}
\item[4.] Measure the entire output state.
\end{description}
\end{algorithm}

\begin{figure}[ht!]
\centering
\includegraphics[width=60mm]{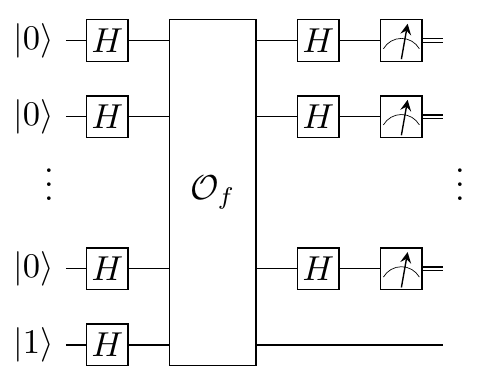}
\caption{\label{fig:i} A quantum circuit whose outcome determines whether a Boolean function
is constant or balanced using only a single query to the membership oracle.}
\end{figure}

\noindent We can verify the correctness of the algorithm as follows:
If we measure the final ouput state of the algorithm in the computational basis for the outcome $\ket{0^n}$, we observe that
\begin{eqnarray}
p(0^n) \,=\, \left\|\frac{1}{\sqrt{2^n}}\sum_{x \in \{0,1\}^n} (-1)^{f(\vec x)} \ket{0^n} \right\|^2 \,=\,\frac{1}{2^n}\sum_{x \in \{0,1\}^n} (-1)^{f(\vec x)}.
\end{eqnarray}
Thus, the amplitudes interfere constructively towards a probability of $1$ if $f$ is constant, whereas they interfere destructively around
a probability of $0$ whenever $f$ is balanced. Note that the quantum algorithm only required as much as a single query to the oracle.
Classically, in the worst-case setting, any algorithm requires $\Omega(2^{n-1}+1)$ classical queries to the oracle in order to learn more than 
half of the evaluations of the function.

\subsection{Bernstein-Vazirani Algorithm}

Another potentially interesting problem in complexity theory is the task of determining a hidden string from inner product of bit strings.
In 1993, Bernstein and Vazirani \cite{BV93} initiated the field of quantum complexity theory and proposed a
quantum algorithm achieving a superpolynomial speed-up over classical algorithms. In brief,
the problem can be stated as the following learning problem of determining a secret string:\\
\ \\

 \begin{mdframed}[backgroundcolor=green!5] 
\textbf{Bernstein-Vazirani Problem:}\\
\textit{
Learn a string $\vec s \in \{0,1\}^n$ by querying an oracle for a Boolean function $f_s: \{0,1\}^n \rightarrow \{0,1\}$ given by
\begin{equation}
f_s(\vec x) = s_1 \cdot x_1 \oplus ... \oplus s_n \cdot x_n \, =\, \braket{\vec s,\vec x}  \pmod 2.
\end{equation}}
\end{mdframed}

In the classical membership oracle setting, we observe that a single query to the function can only ever reveal as much as one bit of information 
about the secret string $s$. 
In fact, this can easily be done by considering queries on strings $e_i = (0,...\,, 1, ...\,,0)$,
where the $i$-th index is $1$ and $e_i$ is $0$ everywhere else. An algorithm performing such queries achieves an overall query complexity 
of $\Omega(n)$ when determining the secret, as each iteration reveals only a single bit of the hidden string by querying
\begin{equation}
f_s(\vec e_i) = \, \braket{\vec s,\vec e_i} \pmod 2 \, = \, s_i,
\end{equation}
so that $s$ is fully determined after a total of $n$ queries to the function.

In the quantum membership oracle model, Bernstein and Vazirani showed that only a single quantum query to the oracle is sufficient \cite{BV93}:

\begin{algorithm}
\label{bernstein_algorithm}
\caption{Bernstein-Vazirani Algorithm}
\begin{description}
\item[Input:] A quantum black box oracle $\mathcal{O}_{f_s}$ for a Boolean function $f_s$, where $f_s(\vec x)=\braket{\vec s,\vec x}$. The task is to determine 
$\vec s \in \{0,1\}^n$.
\item[Output:] The secret string with only a single query to the oracle.

\item[Procedure:]
\item[1.] Initialize $(n+1)$-qubits to $\ket{0^n}\ket{1}$ and apply the Hadamard transform $H^{\oplus (n+1)}$:
\begin{equation*}
\longrightarrow  \frac{1}{\sqrt{2^n}} \sum_{x\in \{0,1\}^n} \ket{\vec x}\ket{-} \phantom{kklkkkkkkkklkkkk}
\end{equation*}
\item[2.] Query the quantum oracle, resulting in a \textit{phase-kickback}:
\begin{equation*}
\longrightarrow \frac{1}{\sqrt{2^n}} \displaystyle\sum_{x\in \{0,1\}^n} (-1)^{\braket{\vec s,\vec x}} \ket{\vec x} \ket{-} \phantom{kklkkkkkk}
\end{equation*} 
\item[3.] Throw away the last register and then apply a Hadamard gate $\text{H}^{\otimes n}$:
\begin{equation*}
\longrightarrow \frac{1}{\sqrt{2^n}}\sum_{x,y \in \{0,1\}^n} (-1)^{\braket{\vec s,\vec x}} (-1)^{\vec x\cdot \vec y}\ket{\vec y} \phantom{kkkkk}
\end{equation*}
\item[4.] Measure the entire output state.
\end{description}
\end{algorithm}

\begin{figure}[ht!]
\centering
\includegraphics[width=60mm]{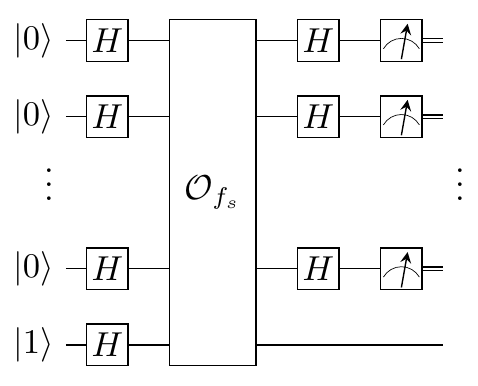}
\caption{\label{fig:i} A quantum circuit for the Bernstein-Vazirani problem. The secret string is determined after only a single query to the membership oracle.}
\end{figure}

\noindent Thus, if we now measure the final ouput state of the algorithm in the computational basis for a particular outcome $\vec m \in \{0,1\}^n$, we observe that
\begin{eqnarray}
p(m) &=& \left\|\frac{1}{\sqrt{2^n}}\sum_{x \in \{0,1\}^n} (-1)^{\braket{\vec s,\vec x}}(-1)^{\vec x\cdot \vec m} \ket{\vec m} \right\|^2\\
             &=& \frac{1}{2^{n}}\sum_{x \in \{0,1\}^n} (-1)^{(s_1\oplus m_1)x_1} (-1)^{(s_2\oplus m_2)x_2} \, \cdots (-1)^{(s_n\oplus m_n)x_n}\\
             &=& \frac{1}{2^{n}} \left(\sum_{x_1 \in \{0,1\}} (-1)^{(s_1\oplus m_1)x_1}\right) \left(\sum_{x_2 \in \{0,1\}} (-1)^{(s_2\oplus m_2)x_2}\right) \cdots \left(\sum_{x_n \in \{0,1\}} (-1)^{(s_n\oplus m_n)x_n}\right) \nonumber \\
             &=& \frac{1}{2^n}\prod_{j=1}^n\sum_{x_j =0}^{1}(-1)^{(s_j\oplus m_j)x_j}\\
             &=& \label{prob_bernstein} \frac{1}{2^n}\prod_{j=1}^n\left(1+(-1)^{s_j\oplus m_j}\right).
\end{eqnarray}
Note that the probabilities of measuring any output states $\vec m \neq \vec s$ vanish due to \eqref{prob_bernstein}. 
Consequently, the amplitudes interfere constructively towards a probability of $1$ if every bit of the output state $\vec m$ is equal to $s$.

%% file: quantum_fourier_transform.tex
\newpage
\section{The Quantum Fourier Transform}

The quantum Fourier transform (\QFT) is arguably the most important and most widely used
tool in the design of quantum algorithms. Many of the known algorithms, such as the Deutsch-Josza algorithm,
Shor's factoring algorithm, quantum phase estimation, quantum order finding, Simon's algorithm \cite{Sim97} or the Bernstein-Vazirani algorithm 
rely substantially on its use. 
In the previous section, we encountered the Hadamard transform
as a single qubit gate $H$ that performs the following operation:
\begin{equation}
H \ket{x} = \frac{1}{\sqrt{2}} \displaystyle\sum_{y=0}^1 { (-1)^{x\cdot y} \ket{y}}.
\end{equation}
Implicitly, we have already encountered the single-qubit quantum Fourier transform of order $n=1$.
In fact, the Hadamard transform can be thought of as a \QFT over the group $\mathbb{Z}/{2}\mathbb{Z}$
that takes single qubits in the computational basis and maps them to the Hadamard basis.
The underlying principle of the Fourier transform already starts to show, namely, that the \QFT acts
as a change of basis in which the amplitudes of individual states of the computational basis are \textit{related} to the 
amplitudes of the entire computational space.
The amplitudes of the transformed states are the so-called \textit{characters} of the Fourier transform.\\
In this chapter, our goal is to introduce a general \textit{qudit} extension of the Hadamard transform,
the \QFT on the cyclic group $\mathbb{Z}/{q}\mathbb{Z}$, where $q$ is any integer.
By extension, we also obtain the \QFT over any finite Abelian group.
More importantly, we discuss how the \QFT can be efficiently implemented on a quantum computer.
Historically, the earliest non-trivial variant known to be efficiently computable on a quantum computer is the Fourier transform over the group 
$\mathbb{Z}/{2^n}\mathbb{Z}$, due to Deutsch \cite{Deu85}. 
The general variant of the quantum Fourier transform where $q$ is an arbitrary integer is less common, but also mentioned as a side note in \cite{NC10}.
Techniques due to Kitaev \cite{Kit95} first allowed to generalize this variant of the Fourier transform using quantum phase estimation.
Evidently, this breakthrough immediately led to the generalization over finite Abelian groups, as we will discuss in
detail in the next section. 
In both cases, we review efficient quantum circuit implementations and thus
provide the basis for addressing the \QFT in the extended Bernstein-Vazirani problem.\\

\newpage
\subsection{The Quantum Fourier Transform over Finite Abelian Groups}
The Fourier transform can be defined on arbitrary groups and we can extend the same principle of basis change into the language of groups
and group algebras. For further reading on the Fourier transform on groups, we refer to the survey \cite{CD10} or the 
supplementary chapters in \cite{NC10}. In this thesis, we concern ourselves with the case of finite Abelian groups.\\
Consider a finite Abelian group $(G,*)$ of order $|G| = N$ and let $\mathbb{C}G = span\{\ket{g}:g \in G\}$
be the associated group algebra of $G$ over $\mathbb{C}$. 
Each element $\ket{x} \in \mathbb{C}G$ can be uniquely expressed as
a linear combination of basis vectors and complex coefficients:
\begin{equation}
\ket{x} = \sum_{g \in G} \alpha_g \ket{g}, \, \, \, \, \text{where } \alpha_g \in \mathbb{C}.
\end{equation} 
Since $G$ is finite Abelian, there exists a basis $\hat G \subseteq \mathbb{C}G$ of dimension $|\hat G| = N$
that consists of $N$ distinct irreducible characters $\chi \in \hat G$,
where $\chi: G \longrightarrow \mathbb{C}$ and $\chi(a * b)=\chi(a)\cdot \chi(b)$ (\cite{CD10}, Appendix B).
Moreover, if $\chi,\chi' \in \hat G$ are two irreducible characters, then:
\begin{equation}
\label{orthogonality}
\sum_{x \in G} \chi(x) * \overline{\chi'(x)} \, = N \, \delta_{\chi,\chi'}.
\end{equation}
\begin{definition}\label{def:qft}
Let $(G,*)$ be a finite Abelian group of order $|G| = N$ and let $\hat G$ be the basis containing the set of $N$ distinct characters
$\chi: G \longrightarrow \mathbb{C}$.\\
The quantum Fourier transform $\mathcal{F}_G$ on the group $G$ is defined as the operation:
\begin{equation}
\label{fourier_character}
\ket{x} \longrightarrow \frac{1}{\sqrt{N}} \sum_{\chi \in \hat G} \chi(x) \ket{\chi}.
\end{equation}
\end{definition}
Consider now the case where $G$ is given by the group $G=(\mathbb{Z}/{N}\mathbb{Z},+)$ and $N$ is any integer. In this case, the irreducible characters
$\chi: \mathbb{Z}/{N}\mathbb{Z} \longrightarrow \mathbb{C}$ are given precisely
by the primitive $N^{\text{th}}$ roots of unity $\omega_N = e^{ \frac{2\pi i}{N}}$.
Thus, for every $y \in \mathbb{Z}/{N}\mathbb{Z}$, the character $\chi_y(x) = \omega_N^{xy}$ is uniquely determined.
By choosing an orthonormal basis $\ket{0},\ket{1}, ..., \ket{N}$ of $\hat G$ in Fourier space, we can identify \eqref{fourier_character} as:
\begin{equation}
\ket{x} \longrightarrow \frac{1}{\sqrt{N}} \sum_{y \in \mathbb{Z}/{N}\mathbb{Z}} \omega_N^{x\cdot y} \ket{y}.
\end{equation}
Furthermore, we can associate the Fourier transform $\mathcal{F}_G$ with the operator:
\begin{equation}
\mathcal{F}_G = \frac{1}{\sqrt{N}} \sum_{x,y \in \mathbb{Z}/{N}\mathbb{Z}} \omega_N^{x\cdot y} \ket{y}\bra{x}.
\end{equation}
Let us briefly state the following fact, analogous to as in Eq.\eqref{orthogonality}, regarding orthogonality of the roots of unity:
\begin{proposition}
\label{lem:ortho}
\begin{equation}
\sum_{y \in \mathbb{Z}/{N}\mathbb{Z}} \omega_N^{x \cdot y} \omega_N^{- x'\cdot y} = N \, \delta_{x,x'}.
\end{equation}
\end{proposition}
\begin{proof}
Consider first the case when $x = x'$. Then, for all $y \in \mathbb{Z}/{N}\mathbb{Z}$, we find $\omega_N^{(x-x')y}=1$ and the above sum clearly adds up to $N$.
If $x\neq x'$, we can apply the partial sum formula of the geometric series and compute:
\begin{equation}
1 + \omega_N^{x-x'} + (\omega_N^{x-x'})^2 + ... + (\omega_N^{x-x'})^{N-1} = \frac{1-(\omega_N^{x-x'})^{N}}{1- \omega_N^{x-x'}} = 0.
\end{equation}\qed
\end{proof}
We can apply \expref{Proposition}{lem:ortho} in order to show that $\mathcal{F}_G$ is indeed a well defined unitary operation:
\begin{align}
\mathcal{F}_G \mathcal{F}^\dagger_G &= \frac{1}{N} \sum_{x,y \in \mathbb{Z}/{N}\mathbb{Z}} \omega_N^{x\cdot y} \ket{x}\bra{y}
                                        \sum_{x',y' \in \mathbb{Z}/{N}\mathbb{Z}} \omega_N^{-y'\cdot x'} \ket{y'}\bra{x'}\\                                        
                                       &= \frac{1}{N} \sum_{x,x',y,y' \in \mathbb{Z}/{N}\mathbb{Z}} \omega_N^{x\cdot y - x'\cdot y'} \delta_{y,y'} \ket{x}\bra{x'}\\
                                       &= \frac{1}{N} \sum_{x,x',y \in \mathbb{Z}/{N}\mathbb{Z}} \omega_N^{(x - x')\cdot y}\ket{x}\bra{x'} \, \, \, = \,  \sum_{x,x' \in \mathbb{Z}/{N}\mathbb{Z}}\delta_{x,x'} \ket{x}\bra{x'} \, \, = \,  \mathbbm{1}.\qed                                
\end{align}
According to the \textit{fundamental classification of finite abelian groups} \cite{CD10}, any finite Abelian group $G$ is structurally equivalent,
i.e. isomorphic, to a direct product of cyclic factors whose orders are prime powers. 
Let $|G|=N$ and let $N = p_1^{r_1} \hdots p_k^{r_k}$ be the unique prime factorization of $N$,
then:
\begin{equation}
\label{fundamental_theorem_abelian}
\large{G \, \cong \, \mathbb{Z}/{p_1^{r_1}}\mathbb{Z} \times \hdots \times \mathbb{Z}/{p_k^{r_k}}\mathbb{Z}}.
\end{equation}
Moreover, the basis $\hat G$ of irreducible characters is given by products of irreducible characters of the respective factors in 
\eqref{fundamental_theorem_abelian}.
Consequently, following \cite{CD10}, the quantum Fourier transform on finite Abelian groups $G$ is given by:
\begin{equation}
\large{\mathcal{F}_G = \mathcal{F}_{\mathbb{Z}/{p_1^{r_1}}\mathbb{Z}} \otimes \hdots \otimes \mathcal{F}_{\mathbb{Z}/{p_k^{r_k}}\mathbb{Z}}.}
\end{equation}
For example, let $G = (\left(\mathbb{Z}/{N}\mathbb{Z}\right)^n,+)$. Then, for any $y \in \left(\mathbb{Z}/{N}\mathbb{Z}\right)^n$, we can associate
a unique irreducible character $\chi_y: \mathbb{Z}/{N}\mathbb{Z} \longrightarrow \mathbb{C}$ such that for all $x \in \left(\mathbb{Z}/{N}\mathbb{Z}\right)^n$:
\begin{equation}
\chi_y(x) = \chi_y(x_1) \cdot \cdot \cdot \chi_y(x_n) = \omega_N^{ x_1 \cdot y_1 + \hdots + x_n \cdot y_n}.
\end{equation}
Hence, the quantum Fourier transform $\mathcal{F}_G$ over the group $G = \left(\mathbb{Z}/{N}\mathbb{Z}\right)^n$ is given by:
\begin{equation}
\ket{x_1}\ket{x_2}\hdots\ket{x_n} \longrightarrow  \frac{1}{\sqrt{N^n}} \sum_{y \in \left(\mathbb{Z}/{N}\mathbb{Z}\right)^n} \omega_N^{ x_1 \cdot y_1 + \hdots + x_n \cdot y_n} \ket{y_1}\ket{y_2}\hdots\ket{y_n},
\end{equation}
where, from now on, $ \braket{x,y} = x_1 \cdot y_1 + \hdots + x_n \cdot y_n$.
In the following sections, we consider different variants of groups $G$ and derive a quantum circuit implementation that realizes the corresponding Fourier transformations.
Before we give efficient circuit implementations, we require additional remarks.\\
Let $G$ be the group $G=(\mathbb{Z}/{N}\mathbb{Z},+)$ and consider the \textit{shift operator} $U(1)$ that performs the the following operation:
\begin{equation}
U(1): \, \ket{x} \longrightarrow \ket{x+1},
\end{equation}
where $x+1$ is the cyclic addition \textit{mod} $N$. We can verify that $U(1) = \sum_{x \in \mathbb{Z}/{N}\mathbb{Z}} \ket{x+1}\bra{x}$ is unitary, since:
\begin{align}
U(1)U(1)^\dagger &= \sum_{x \in \mathbb{Z}/{N}\mathbb{Z}} \ket{x+1}\bra{x} \sum_{x' \in \mathbb{Z}/{N}\mathbb{Z}} \ket{x'}\bra{x'+1}\\
                 &= \sum_{x,x' \in \mathbb{Z}/{N}\mathbb{Z}} \ket{x+1}\braket{x|x'}\bra{x'+1} = \sum_{x \in \mathbb{Z}/{N}\mathbb{Z}} \ket{x+1}\bra{x+1} \, = \, \mathbbm{1}.
\end{align}
Therfore, we can prove the following statement that connects our previous discussion on the quantum Fourier transform with the shift operator, as
appeared in the work of Kitaev \cite{Kit95}:
\begin{proposition}\label{lem:fourier_basis}
The shift operator $U(1)$ is diagonal in the Fourier basis:
\begin{equation}
\mathcal{F}_G \,U(1) \mathcal{F}_G^\dagger = \sum_{y \in \mathbb{Z}/{N}\mathbb{Z}} \omega_N^{y} \ket{y}\bra{y}.
\end{equation}
\end{proposition}
\begin{proof}
Let the operators be represented by:
\begin{align*}
\mathcal{F}_G &=\frac{1}{\sqrt{N}} \sum_{x,y \in \mathbb{Z}/{N}\mathbb{Z}} \omega_N^{x\cdot y} \ket{y}\bra{x}, \hspace{6mm} U(1)\, =\sum_{z \in \mathbb{Z}/{N}\mathbb{Z}} \ket{z+1}\bra{z} \phantom{kkk}\text{and}\\
\mathcal{F}_G^\dagger &= \frac{1}{\sqrt{N}} \sum_{x',y' \in \mathbb{Z}/{N}\mathbb{Z}} \omega_N^{-y'\cdot x'} \ket{y'}\bra{x'}.
\end{align*}
Using \expref{Proposition}{lem:ortho}, we compute:
\begin{align}
\mathcal{F}_G \,U(1) \mathcal{F}_G^\dagger &= 
           \frac{1}{N} \sum_{x,y,x',y' \in \mathbb{Z}/{N}\mathbb{Z}} \omega_N^{x\cdot y} \omega_N^{-y'\cdot x'} \ket{y}\braket{x\,|\,y'+1}\bra{x'} \phantom{kkkkkkk}\\
           &= \frac{1}{N} \sum_{y,x',y' \in \mathbb{Z}/{N}\mathbb{Z}} \omega_N^{(y'+1)\cdot y} \omega_N^{-y'\cdot x'} \ket{y}\bra{x'} \phantom{kkkkkkk}\\
           &=  \sum_{y,x' \in \mathbb{Z}/{N}\mathbb{Z}} \omega_N^{y} \left[\frac{1}{N}\sum_{y' \in \mathbb{Z}/{N}\mathbb{Z}}\omega_N^{y'\cdot y} \omega_N^{-y'\cdot x'} \right]\ket{y}\bra{x'} \phantom{kkkkkkk}\\
           &=  \sum_{y,x' \in \mathbb{Z}/{N}\mathbb{Z}} \omega_N^{y} \, \delta_{y,x'} \,\ket{y}\bra{x'} \phantom{kkkkkkk}\\
           &=  \sum_{y \in \mathbb{Z}/{N}\mathbb{Z}} \omega_N^{y} \ket{y}\bra{y}. \phantom{kkkkkkk}
\end{align}
\end{proof}

\newpage
\subsection{Efficient Circuit Implementations}

One of the earliest efficient circuit implementations was found for the \QFT of order $N=2^n$,
as provided by the following theorem:

\begin{theorem}[\cite{BV93}]
For any integer $n$ and order $N=2^n$, there exists an efficient quantum circuit that uses $O(n^2)$ elementary gates
and performs the quantum Fourier transform on any of the orthonormal basis states orthonormal basis $\ket{0},\ket{1}, ..., \ket{N-1}$:
\begin{equation}
\ket{x} \longrightarrow \frac{1}{\sqrt{2^n}} \sum_{y \in \mathbb{Z}/{2^n}\mathbb{Z}} \omega_N^{x\cdot y} \ket{y}.
\end{equation}
\end{theorem}

\begin{figure}[ht!]
\centering
\includegraphics[width=130mm]{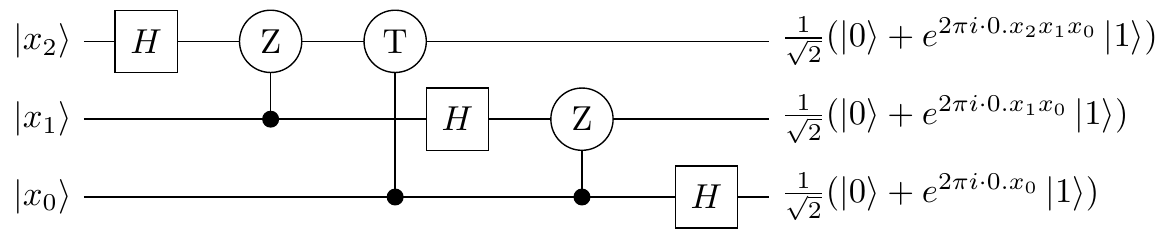}
\caption{A quantum circuit that performs the \QFT for three qubits using the following elementary gates: 
Hadamard (H), Controlled-Z (Z) and the Controlled-$\pi/4$ Phase-shift (T).}\label{fig:qft_2}
\end{figure}

In particular, in the case of $N=2^n$, we can use a binary representation $x= \sum_{j=0}^{n-1} 2^j x_j$ so that $x=x_1 x_2 ... x_n$.
Moreover, it is also helpful to introduce the binary fraction notation $[0.x_1...x_m] = \sum_{i=1}^{m} 2^i x_i.$
This allows us to write the \QFT in terms of a separable product of $n$-qubits \cite{NC10}:
\begin{align}
\ket{x} &\rightarrow \frac{1}{\sqrt{2^n}} \sum_{y = 0}^{2^n-1} \omega_N^{x\cdot y} \ket{y}\\
        &= \frac{1}{\sqrt{2^n}} \sum_{y \in \{0,1\}^n} \omega_N^{x\sum_{j=0}^{n-1} 2^j y_j} \ket{y_1}...\ket{y_n}\\
        &= \frac{1}{\sqrt{2^n}} \bigotimes_{j=0}^{n-1}\sum_{y_j \in \{0,1\}} e^{2 \pi i x y_j/2^{n-j}} \ket{y_j}\\
        &= \bigotimes_{j=0}^{n-1} \frac{ \ket{0} + e^{2 \pi i \sum_{k=0}^{n-1} 2^{j+k-n} x_k} \ket{1}}{ \sqrt{2}  }
\end{align}
Next, we would like to extend the \QFT onto an arbitrary cyclic group $\mathbb{Z}/{N}\mathbb{Z}$, by using a technique due to Kitaev \cite{Kit95}.
Following \cite{CD10}, we can derive this transformation using \textit{quantum phase estimation}, an efficient quantum procedure 
for the estimation of eigenvalues for a given unitary operator \cite{NC10}.
The goal is to perform the \QFT over $\mathcal{F}_{\mathbb{Z}/{N}\mathbb{Z}}$ and map a state $\ket{x}$ into the Fourier basis $ \ket{\hat x}$, as follows:
\begin{equation}
\label{kitaev_ft}
\ket{x} \longrightarrow \ket{\hat x} = \frac{1}{\sqrt{N}} \sum_{y \in \mathbb{Z}/{N}\mathbb{Z}} \omega_N^{x\cdot y} \ket{y}.
\end{equation}
Let us first note that by additionally attaching the input state in Eq.\eqref{kitaev_ft}, it is straightforward to realize the above operation
as a two-qubit operation using elementary gates.
This can be verified as follows: First prepare the state $\ket{x}\ket{0}$ and create a uniform superposition in the second register:
\begin{equation}
\ket{x}\ket{0} \longrightarrow \frac{1}{\sqrt{N}} \sum_{y \in \mathbb{Z}/{N}\mathbb{Z}}\ket{x} \ket{y}.
\end{equation}
Consider now applying a controlled phase-shift gate $\ket{x}\ket{y} \longrightarrow  \omega_N^{x\cdot y}  \ket{x}\ket{y}$.
As a result, the output is thus transformed into:
\begin{equation}
\frac{1}{\sqrt{N}} \sum_{y \in \mathbb{Z}/{N}\mathbb{Z}} \ket{x}\ket{y} \longrightarrow \frac{1}{\sqrt{N}} \sum_{y \in \mathbb{Z}/{N}\mathbb{Z}} \omega_N^{x\cdot y}  \ket{x}\ket{y}.
\end{equation}
However, due to entanglement of the registers, straightforward erasure of the first register is not possible.
At this point, however, we can make use of the quantum phase estimation procedure that allows us to
efficiently approximate the eigenvalues of a given unitary operator with $n=O(\log N)$ bits of precision \cite{NC10}. 
According to \expref{Lemma}{lem:fourier_basis},
the eigenvalues of the shift operator $U(1)$ are precisely given by the roots of unity $\omega_N$. Thus, we can approximately perform the following 
unitary operation
$\mathcal{P}$:
\begin{equation}
\ket{\hat x}\ket{0} \longrightarrow \ket{\hat x}\ket{x}
\end{equation}
By reversing the above operation and applying $\mathcal{P}^\dagger$, we arrive at the desired outcome:
\begin{equation}
\ket{\hat x}\ket{x} \longrightarrow \ket{\hat x}\ket{0}.
\end{equation}
Finally, we refer to the following highly efficient realization of the \QFT due to Hales and Hallgren:
\begin{theorem}[\cite{HH00}]
For arbitrary integers $N$, where $n=\log(N)$, and any $\epsilon >0$, there exists an efficient quantum circuit that uses 
$O(n \log \frac{n}{\epsilon} + \log^2\frac{1}{\epsilon})$ many gates
and approximately performs the quantum Fourier transform on orthonormal basis states $\ket{0},\ket{1}, ..., \ket{N-1}$
up to a fidelity of $\epsilon$:
\begin{equation}
\ket{x} \longrightarrow \frac{1}{\sqrt{N}} \sum_{y \in \mathbb{Z}/{N}\mathbb{Z}} \omega_N^{x\cdot y} \ket{y}.
\end{equation}
\end{theorem}

%% file: bernstein_problem.tex
\newpage
\section{Quantum Learning Algorithms}

Quantum computers can indeed solve certain problems faster than classical computers, as demonstrated in \expref{Section}{quantum_algorithms}.
The tasks we considered so far all concerned static learning tasks with well-defined entities free of noise and error.
As decoherence still poses a major threat to current quantum computing architectures, the promise of successfully running quantum algorithms
is still largely dependent on the extent to which fault-tolerant computing is currently realized. 
The theory of quantum error correction has been crucial in establishing the prospect
of fault-tolerant quantum computing in the near future. Currently, noise in quantum computing architectures
is regarded as fatal and believed to substantially slow down most quantum improvements over classical algorithms. In this chapter,
we review recent work by Cross, Smith and Smolin \cite{CSS14}, showing that quantum algorithms can indeed solve certain tasks in the presence of 
certain classes of noise,
much to the contrary of their classical counterparts. We introduce tasks from computational learning theory, such as the
\textit{Learning Parity with Noise} (\LPN) problem which concerns the decoding of random linear binary codes. The \LPN problem
is conjectured to be classically intractable, as the best known algorithms require sub-exponential numbers of recources \cite{BFKL94}. 
However, learning in the quantum setting remains easy despite the presence of noise. Furthermore, we discuss its consequences
in a related setting in which we consider the \LWE problem with quantum samples.
To this end, we propose a new generalization of the Bernstein-Vazirani algorithm of \expref{Section}{quantum_algorithms} and present
recent results by Grilo and Kerenidis \cite{GK17}, demonstrating a successful amplification of the success probability in the presence of noise.\\

\subsection{Computational Learning Theory}

We begin by introducing a few basic notions from computational learning theory, following a recent survey on quantum learning theory by 
Arunachalam and de Wolf \cite{AdW17}. Let us first remark that
all definitions in this section may consider arbitrary computational spaces $X$, hence translate naturally in the context of qubits and qudits, i.e. for $X=\mathbb{Z}/{q}\mathbb{Z}$,
where $q$ is some positive integer.
We start with a few relevant definitions regarding the objectives of learning.

A \textit{concept class} $\mathfrak{C} = \bigcup\limits_{n \geq 1} C_n$ is a collection of \textit{concepts} (typically Boolean functions)
in which each set $C_n$ is to contains all \textit{concepts} $f: X^n \longrightarrow X$. We consider \textit{learning problems} as a setting in which
a \textit{learner} $\mathcal{A}$, i.e. an algorithm, is given access to either a membership or example oracle for a 
\textit{target concept} $f \in C_n$ and the task is
to then find a \textit{hypothesis} $h\in C_n$ that agrees with the concept $f$ upon some measure of accuracy. In other words,
having access to a black box oracle, the goal of the learner is to correctly identify the oracle that corresponds to the target concept.
Let us now specify variants of learning models, both in the classical, as well as in the quantum setting.

\subsubsection{Exact Learning.}

\begin{definition}[Classical Exact Learning] \ \\
In a classical exact learning model, a learner $\mathcal{A}$ for a concept class $C_n$ is given access to a membership oracle $\mathcal{O}_f$ for a 
\textit{target concept} $f \in C_n$ and the task is to find a hypothesis $h \in C_n$ that agrees with the target concept $f$
on all the inputs in $X$. Upon input $x \in X$, the membership oracle $\mathcal{O}_f$ outputs a label $f(x)$.\\
We say an efficient algorithm $\mathcal{A}$ is an exact learner for $C_n$ if, for every $f \in C_n$, there exists $\delta > 0$
such that, with probability $1-\delta$, $\mathcal{A}$ outputs a hypothesis $h$ where for all $x \in X: \,h(x) = f(x)$.\\
We refer to the query complexity of $\mathcal{A}$ as the maximum number of requests to the membership oracle, over all $f \in C_n$, 
as well as over the internal randomness needed to achieve the desired success probability of $1-\delta$.
\end{definition}

\begin{definition}[Quantum Exact Learning] \ \\
In a quantum exact learning model, a learner $\mathcal{A}$ for a concept class $C_n$ is given access to a quantum membership oracle $\mathcal{O}_f$ for a 
\textit{target concept} $f \in C_n$ and the task is to find a hypothesis $h \in C_n$ that agrees with the target concept $f$
on all the inputs in $X$. Upon input $x \in X$
and $y \in Y$, the membership oracle performs the operation:
$$\mathcal{O}_f: \ket{x}\ket{y} \longrightarrow \ket{x}\ket{y \oplus f(x)}.$$
We say an efficient quantum algorithm $\mathcal{A}$ is an exact learner for $C_n$ if, for every $f \in C_n$, there exists $\delta > 0$
such that, with probability $1-\delta$, $\mathcal{A}$ outputs a hypothesis $h$ where for all $x \in X: \,h(x) = f(x)$.\\
Similarly, we now refer to the quantum query complexity of $\mathcal{A}$ as the maximum number of quantum queries to the membership oracle, over all $f \in C_n$, 
as well as over the internal randomness needed to achieve the desired success probability of $1-\delta$.
\end{definition}

\subsubsection{PAC Learning.}

In this section, we introduce a variant called \textit{probably approximately correct} (\PAC) learning,
a model in which we consider uniform example oracles contrary to membership oracles. We begin by specifying
the learning model in the classical, as well as quantum setting.

\begin{definition}[Classical \PAC Learning] \ \\
In a \PAC learning model, a learner $\mathcal{A}$ for a concept class $C_n$ is given access to a uniform example oracle $\mathcal{O}_{EX}(f)$ for a 
\textit{target concept} $f \in C_n$ and the task is to find a hypothesis $h \in C_n$ that agrees with the target concept $f$
on at least a $1-\epsilon$ fraction of the inputs in $X$.\\
Upon each query, the example oracle $\mathcal{O}_{EX}(f)$ samples a label $f(x)$ uniformly at random.\\
We say an algorithm $\mathcal{A}$ is a PAC learner for $C_n$ if, for every $f \in C_n$, there exists an $\epsilon > 0$ and $\delta > 0$
such that, with probability $1-\delta$, $\mathcal{A}$ outputs a hypothesis $h$, where:
\begin{enumerate}
\item $\displaystyle\Pr_{x \in X}[h(x) = f(x)] \geq 1-\epsilon$.
\item $\mathcal{A}$ runs in time and uses a number of queries that is poly$(n,1/\epsilon,1/\delta)$.
\end{enumerate}
We refer to the query complexity of $\mathcal{A}$ as the maximum number of requests to the example oracle, over all $f \in C_n$, 
as well as over the internal randomness needed to achieve the desired success probability of $1-\delta$.
The $(\epsilon, \delta)$-\PAC sample complexity of a concept class C is given by the minimum sample complexity over all 
$(\epsilon, \delta)$-\PAC learners for $C_n$.
\end{definition}

\begin{definition}[Quantum \PAC Learning] \ \\
In a quantum \PAC learning model, a learner $\mathcal{A}$ for a concept class $C_n$ is given access to a quantum example oracle $\mathcal{O}_{EX}(f)$ for a 
\textit{target concept} $f \in C_n$ and the task is to find a hypothesis $h \in C_n$ that agrees with the target concept $f$
on at least a $1-\epsilon$ fraction of the inputs in $X$.\\
When queried, the example oracle $\mathcal{O}_{EX}(f)$ responds with a quantum state:
$$ \frac{1}{\sqrt{|X|}}\sum_{x \in X} \ket{x_1} \hdots \ket{x_n}\ket{f(x)}.$$
We say a quantum algorithm $\mathcal{A}$ is a quantum \PAC learner for $C_n$ if, for every $f \in C_n$, there exists an $\epsilon > 0$ and $\delta > 0$
such that, with probability $1-\delta$, $\mathcal{A}$ outputs a hypothesis $h$, where:
\begin{enumerate}
\item $\displaystyle\Pr_{x \in X}[h(x) = f(x)] \geq 1-\epsilon$.
\item $\mathcal{A}$ runs in time and uses a number of queries that is poly$(n,1/\epsilon,1/\delta)$.
\end{enumerate}
We refer to the query complexity of $\mathcal{A}$ as the maximum number of requests (at unit cost) to the example oracle, over all $f \in C_n$, 
as well as over the internal randomness needed to achieve the desired success probability of $1-\delta$.
The $(\epsilon, \delta)$-\PAC sample complexity of a concept class C is given by the minimum sample complexity over all 
$(\epsilon, \delta)$-\PAC learners for $C_n$.
\end{definition}

\begin{definition}[\PAC Learnable Classes] \ \\
We say a concept class $\mathfrak{C} = \bigcup\limits_{n \geq 1} C_n$ is classically (or quantumly) \PAC learnable if, given an example oracle for any 
target concept $f \in \mathfrak{C}$,
there exists a \PAC algorithm such that, for any $\epsilon,\delta \in (0,1/2)$, the algorithm
\begin{enumerate}
\item outputs an $\epsilon$-approximation $h$ of $f$ with probability $1-\delta$.
\item runs in time and uses a number of queries that is poly$(n,1/\epsilon,1/\delta)$.
\end{enumerate}
\end{definition}\ \\
In the next section, we apply these definitions to the \textit{learning party with noise} problem and consider
classical, as well as quantum algorithms.

\subsection{Learning Parity With Noise.}\label{quantum_lpn}

Consider the following well known computational problem resembling a noisy variant of the Bernstein-Vazirani problem in \expref{Chapter}{quantum_algorithms}:\\

 \begin{mdframed}[backgroundcolor=green!5] 
\textbf{Learning Parity With Noise Problem:}\ \\
Recover the secret $\vec s \in \{0,1\}^n$ by making queries to a uniform example oracle of Bernoulli 
noise rate $\eta < 1/2$ for the class of parity functions $f_s: \{0,1\}^n \rightarrow \{0,1\}$, where
\begin{equation}
\label{bernstein}
f_s(\vec x) = s_1 \cdot x_1 \oplus ... \oplus s_n \cdot x_n \, \,  \text{mod }2 = \, \braket{\vec s,\vec x} \,  \text{mod }2.
\end{equation}
\end{mdframed}

In the noiseless case, this problem amounts to Gaussian elimination given
enough linearly independent samples. Following \cite{CSS14}, the probability that $n$ queries to the example oracle $\mathcal{O}_{f_s}$
produce a set of linearly independent examples is given by:
\begin{align}
\left(1-2^{-n}\right)\cdot\left(1-2^{-n+1}\right)\cdot\cdot\cdot \left(1-\frac{1}{2}\right) 
= \prod_{j=0}^{n-1} \left(1-2^{j-n}\right).
\end{align}
A simple proof by induction shows that this probability is in fact greater than $1/4$ for any integer $n>1$. In the noiseless case,
the class of parity functions is clearly \PAC-learnable. In fact,
any algorithm that fails with constant probability less than some $p \in (0,1)$ can be repeated in the order of
$O(\log_{1/p}1/\delta)$ to reduce the probability of failure below $\delta >0$.
In the case of a noise rate $\eta < 1/2$, it is known that the \LPN problem is an average-case version of the \NP-hard problem of decoding 
a linear code, hence the \LPN problem is thus classically intractable.

In the quantum setting, this problem remains easy even in the presence of noise, as shown in \cite{CSS14}. 
Our goal is to first define the \LPN problem in a quantum oracle model using uniform quantum samples
and show that the \LPN problem is quantumly \PAC-learnable. First note its resemblance to the Bernstein-Vazirani 
problem based on queries from a uniform example oracle.\\
\ \\

 \begin{mdframed}[backgroundcolor=green!5] 
\textbf{Learning Parity With Noise:}\ \\
Recover the secret $\vec s \in \{0,1\}^n$ from the class of parity functions $f_s: \{0,1\}^n \rightarrow \{0,1\}$, where
$f_s(\vec x) = s_1 \cdot x_1 \oplus ... \oplus s_n \cdot x_n \, \,  \text{mod }2 = \, \braket{\vec s,\vec x} \,  \text{mod }2,$ 
by querying a quantum example oracle $\mathcal{O}_{EX}(f_s,\eta)$ of noise rate $\eta < 1/2$. Upon each query, $\mathcal{O}_{EX}(f_s,\eta)$ outputs uniform 
quantum sample given by:
\begin{equation}
\ket{\Psi_s} = \frac{1}{\sqrt{2^n}}\sum_{x\in \bit^n} \ket{x_1}\hdots \ket{x_n} \ket{\braket{\vec x,\vec s} \oplus e},
\end{equation}
where the error follows $e \sim Bernoulli(\eta)$.
\end{mdframed}
\ \\
\ \\
Let us first treat the problem in the \textit{noiseless} case. Consider the following algorithm, as in \cite{CSS14}:

\begin{algorithm} 
\caption{\emph{Quantum Parity Learning}}
\begin{description}
\item[Input:] A quantum example oracle $\mathcal{O}_{EX}(f_s)$ acting as a black box that outputs ideal uniform quantum samples. The task is to determine 
$\vec s \in \{0,1\}^n$.
\item[Output:] The secret string $\vec s \in \{0,1\}^n$ with probability $1/2$.
\ \\
\item[Procedure:]
\ \\
\item[1.] Query $\mathcal{O}_{EX}(f_s)$ and receive a uniform quantum example state $\ket{\psi_{f_s}}$, where
      $$\ket{\psi_{f_s}} = \frac{1}{\sqrt{2^n}} \sum_{x\in \{0,1\}^n} \ket{x_1}\ket{x_2}\hdots\ket{x_n}\ket{f_s(\vec x)}$$
\item[2.] Perform a Hadamard gate onto all $n+1$ registers:
\begin{equation*}
 \longrightarrow \, \, \frac{1}{\sqrt{2}}(\ket{0^n}\ket{0} + \ket{\vec s}\ket{1}) \phantom{kkkkkkk}
\end{equation*} 
\item[3.] Measure the entire output state. Read out $s$ if the last register is $\ket{1}$, else output $\bot$.   
\end{description}
\label{parity_algorithm}
\end{algorithm}

\begin{algorithm}
\caption{\emph{Learning Parity With Noise}}
\begin{description}
\item[Input:] A quantum example oracle $\mathcal{O}_{EX}(f_s,\eta)$ acting as a black box that outputs quantum states 
prone to a parity bit flip error with probability $\eta$. The task is to determine 
$\vec s$.
\item[Output:] The secret string $\vec s \in \{0,1\}^n$ with probability $1/2$, independent of $\eta$.

\item[Procedure:]
\item[1.] Query $\mathcal{O}_{EX}(f_s,\eta)$ and receive a uniform quantum state $\ket{\psi_{f_s}}$, where $e \sim Bern(\eta)$:
      $$\ket{\psi_{f_s}} = \frac{1}{\sqrt{2^n}} \sum_{x\in \{0,1\}^n} \ket{x_1}\ket{x_2}\hdots\ket{x_n}\ket{f_s(\vec x)\oplus e}$$
\item[2.] Perform a Hadamard gate onto all $n+1$ registers.\\
\begin{align*}
\longrightarrow \frac{1}{\sqrt{2}}(\ket{0^n}\ket{1} + \ket{\vec s}\ket{0})  \phantom{kkk} &(\text{with probability } \eta)\\
\longrightarrow \frac{1}{\sqrt{2}}(\ket{0^n}\ket{0} + \ket{\vec s}\ket{1})  \phantom{kkk} &(\text{with probability } 1-\eta)
\end{align*}
\item[3.] Measure the entire output state. Read out any nonzero string, else output $\bot$.     
\end{description}
\end{algorithm}

\begin{figure}
\centering
\includegraphics[width=80mm]{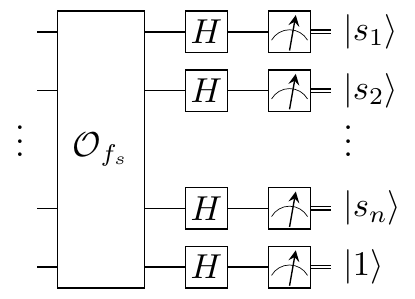}
\caption{\label{fig:i} A quantum circuit for the quantum parity learning algorithm. With probability $1/2$, the final register is measured to be in the state
$\ket{1}$, yielding the secret string. }
\end{figure}

\noindent
The second step in \expref{Algorithm}{parity_algorithm} can easily be verified using \expref{Proposition}{lem:ortho}, as follows:
\begin{equation*}
\begin{split}
     H^{\otimes (n+1)}\ket{\psi_{f_s}} &= { }\frac{1}{\sqrt{2}} \frac{1}{2^n}\sum_{y_{n+1} \in \{0,1\}} \sum_{x,y\in \{0,1\}^n} (-1)^{\braket{\vec x,\vec y}} (-1)^{\braket{\vec x,\vec s}\cdot y_{n+1}} \ket{y_1}\ket{y_2}\hdots\ket{y_n} \ket{y_{n+1}} \\
     &=\frac{1}{\sqrt{2}} \left( \frac{1}{2^n}\sum_{x,y\in \{0,1\}^n} (-1)^{\braket{x,y}} (-1)^{\braket{\vec x, 0}} \ket{y_1}\ket{y_2}\hdots\ket{y_n} \ket{0} \right. \\
      &  \phantom{i + j + k  kkkkkkk} \left. + \, \frac{1}{2^n} \sum_{x,y\in \{0,1\}^n} (-1)^{\braket{\vec x,\vec y}} (-1)^{\braket{\vec x,\vec s}} \ket{y_1}\ket{y_2}\hdots\ket{y_n} \ket{1}               \right)\\
          &=\frac{1}{\sqrt{2}} \left( \sum_{y\in \{0,1\}^n} \delta_{y,0} \ket{y_1}\ket{y_2}\hdots\ket{y_n} \ket{0} \, +\sum_{y\in \{0,1\}^n} \delta_{y,s} \ket{y_1}\ket{y_2}\hdots\ket{y_n} \ket{1}               \right)\\
     &= \, \, \frac{1}{\sqrt{2}} \, (\ket{0^n}\ket{0} + \ket{s}\ket{1}). 
\end{split}
\end{equation*} 
In the \LPN problem, any learning algorithm is given access to quantum samples that are described as a mixture of both noisy and noiseless samples. Surprisingly,
even in this model, the amplitudes interefere constructively once again after the use of Hadamard gates, as discussed in \cite{CSS14}.

\subsection{Extended Bernstein-Vazirani Algorithm.}

In this chapter, we analyze a \textit{qudit} extension of the well known Bernstein-Vazirani problem in a computational learning setting by considering
the problem over the group $\mathbb{Z}/{q}\mathbb{Z}$ under cyclic addition,
where $q$ is any integer.
In solving this problem, we provide the basis for the \LWE problem in a quantum example oracle model.\\

 \begin{mdframed}[backgroundcolor=green!5] 
\textbf{Extended Bernstein-Vazirani Problem:}\ \\
Learn the secret string $s \in \mathbb{Z}_q^n$ by making queries to a uniform example oracle for the concept class of
inner product functions $f_s: \mathbb{Z}_q^n \rightarrow \mathbb{Z}_q$, where $q$ is any positive integer, and
\begin{equation}
\label{bernstein}
f_s(\vec x) = s_1 \cdot x_1 + ... + s_n \cdot x_n  \pmod q = \, \braket{\vec s,\vec x} \pmod q.
\end{equation}
\end{mdframed}

Similar to the noiseless \LPN problem, the classical query complexity of the above problem is given by $\Omega(n)$.
In the quantum setting, we are given a quantum example oracle $\mathcal{O}_{f_s}$ and the goal is to solve the extended Bernstein-Vazirani problem.
In the following, we show that the above problem is exactly learnable by \expref{Algorithm}{extended_bernstein_algorithm}
and discuss its applications for the \LWE problem.

\begin{algorithm}
\caption{\emph{Extended Bernstein-Vazirani Algorithm}}
\begin{description}
\item[Input:] A quantum example oracle $\mathcal{O}_{EX}(f_s)$ acting as a black box 
for the inner product function $f_s(\vec x) = \braket{\vec s,\vec x} \,\pmod q$, where $\vec s \in \mathbb{Z}_q^n$ is to be determined.
\item[Output:] $\vec s \in \mathbb{Z}_q^n$ with probability $\varphi(q)/q$, where $\varphi(q) = |\left(\mathbb{Z}/q\mathbb{Z}\right)^\times|$, else $\bot$.

\item[Procedure:]
\item[1.] Query $\mathcal{O}_{EX}(f_s)$ and receive a quantum example:
      $$  \ket{\Psi_s} = \frac{1}{\sqrt{q^n}}\sum_{x \in \mathbb{Z}_q^n} \ket{x_1}\hdots \ket{x_n} \ket{\braket{\vec s,\vec x} \pmod q}.$$
\item[2.] Apply the quantum Fourier transform $\QFT_{\Z_q}^{\otimes n+1}$.\par
\ \\
\item[3.] Measure in the computational basis, yielding an outcome $\ket{z_1}\ket{z_2}\hdots \ket{z_{n+1}}$.\\
\ \\
\item[4.] \textbf{If} $\text{gcd}(z_{n+1},q)=1$, then output $ \tilde s = ( \frac{-z_1}{z_{n+1}} (\text{mod } q), \frac{-z_2}{z_{n+1}} (\text{mod } q), ...,  \frac{-z_n}{z_{n+1}} (\text{mod } q)),$\\
      \textbf{else} output $\bot$.      
\end{description}
\label{extended_bernstein_algorithm}
\end{algorithm}

\begin{theorem}\label{thm:extended_BV}
The procedure Extended Bernstein-Vazirani in \expref{Algorithm}{extended_bernstein_algorithm}
succeeds at exactly learning the secret string with probability $1-\delta$ by requiring $O(\log(1/\delta))$ samples and running in time poly$(n,\log(1/\delta))$.
\end{theorem}
\begin{proof}
Let $\vec s \in \mathbb{Z}_q^n$ be a secret string and fix the corresponding member $f_s(x) = \braket{s,x} \,\pmod q$ of the class of inner product functions.
Upon receiving uniform quantum examples $\ket{\psi_s}$, where
$$  \ket{\psi_s} = \frac{1}{\sqrt{q^n}}\sum_{x \in \mathbb{Z}_q^n} \ket{x_1}\hdots \ket{x_n} \ket{\braket{\vec s,\vec x} \,\pmod q},$$
the goal is to output $\vec s$.
The procedure Extended Bernstein-Vazirani in \expref{Algorithm}{extended_bernstein_algorithm} applies the quantum Fourier transform, which results in a state
$$
\QFT_{\Z_q}^{\otimes n+1} \ket{\psi_s} = \frac{1}{q^{n+1/2}} \sum_{\vec x \in \Z_q^{n}}\sum_{\vec y \in \Z_q^{n+1}}
  \omega^{\braket{\vec x,\vec y + y_{n+1}\vec s}} \ket{y_1}\dots \ket{y_n}\ket{ y_{n+1}}.
$$
Then, the probability of measuring the outcome $\vec z = -y_{n+1}\vec s \pmod q$ is given by
\begin{align}
\abs{\bra{\vec z}\QFT_{\Z_q}^{\otimes n+1}\ket{\psi_s}}^2 &= \left\| \frac{q^n}{q^{n+1/2}} \sum_{z_{n+1} \in \Z_q} \sum_{x \in \mathbb{Z}_q^n} \omega^0 \ket{z_1}\dots\ket{z_n}\ket{z_{n+1}} \right\|^2\\
                                                          &=  \frac{q^{2n}}{q^{2n+1}} \sum_{z_{n+1}} \left(\sum_{x \in \mathbb{Z}_q^n} 1 \right)^2 \, = \, 1.
\end{align}
Consequently, with probability $\varphi(q)/q$, an outcome $z_{n+1}$ is measured such that $\gcd(z_{n+1},q)=1$, hence the procedure Extended Bernstein-Vazirani in \expref{Algorithm}{extended_bernstein_algorithm} correctly outputs $\vec s$.
Let us consider Euler's product formula\footnote{
The curious quotient $\varphi(q)/q$ is of deep importance to number theory and has been studied for many decades.
In the $1950$ies, Schinzel and Sierpi{\'n}ski proved that $\{\varphi(n)/n : n=1,2,...\}$ is dense in the interval $(0,1)\subset \mathbb{R}$, 
highlighting that the ratio
is highly nontrivial. Therefore, it is not possible to find a unique limit as $q$ approaches infinity.
Euler's product formula, Eq. \eqref{euler_product}, gives us an intuition on how large $\varphi(q)/q$ is, depending on the prime factorization of $q$.
If $q$ is prime, we observe a simple ratio of $\frac{q-1}{q}$, hence a high probability for \expref{Algorithm}{extended_bernstein_algorithm}
to succeed.}
that allows us to write the probability that the algorithm succeeds as:
\begin{equation}
\label{euler_product}
 \frac{\varphi(q)}{q} = \prod_{\text{primes } p | q} \left(1- \frac{1}{p}\right).
\end{equation}
Using a result due to Rosser and Schoenfeld \cite{RS62}, we can also bound the success probability of \expref{Algorithm}{extended_bernstein_algorithm}
in the case where $q>2$:\footnote{In the case of $q=2$, the problem is the noiseless variant of the \LPN problem from \expref{Section}{quantum_lpn}.}
\begin{equation}
 \frac{\varphi(q)}{q} > \frac{1}{e^\gamma \log\log(q) + \frac{3}{\log\log(q)}},
\end{equation}
where $e^\gamma = 1.7810724...$ is Euler's constant. For our purposes, this ratio is still constant for a fixed modulus $q$ and the algorithm
can be repeated to amplify the success probability as it fails with constant probability less than some $p \in (0,1)$.
Thus, if \expref{Algorithm}{extended_bernstein_algorithm} is to succeed
after $m$ repetitions with probability $1-\delta$, we require $O(\log_{1/p}1/\delta)$ samples and time poly$(m,n,\log\frac{1}{\delta})$.
Surprisingly, the sample complexity is independent of $n$, whereas the classical query complexity is at least $\Omega(n)$.
\qed
\end{proof}

\subsection{Learning with Errors with Quantum Examples.}

Finally, we can state the algorithm for the \LWE problem with quantum samples. While the Extended-Bernstein-Vazirani algorithm we introduce
considers any integer modulus $q$, Grilo and Kerenidis \cite{GK17} independently proposed a similar algorithm, specifically for the case when $q$ is prime, in order to solve
the \LWE problem using quantum samples. In addition, it is shown in \cite{GK17} that the prime modulus Bernstein-Vazirani Algorithm can be amplified to solve
the \LWE problem with quantum samples up to arbitrarily high success probability.

The following algorithm runs the Extended-Bernstein-Vazirani procedure.\ \\

\IncMargin{1em}
\begin{algorithm}[H]
\SetKwData{Left}{left}\SetKwData{This}{this}\SetKwData{Up}{up}
\SetKwInOut{Input}{input}\SetKwInOut{Output}{output}
\Input{A quantum example oracle $\mathcal{O}_{EX}(\vec s,\chi)$ that outputs uniform \LWE samples
      $$  \ket{\psi_s} = \frac{1}{\sqrt{q^n}}\sum_{x \in \Z_q^n} \ket{x_1}\hdots \ket{x_n} \ket{ \braket{\vec x,\vec s} + e_x (\text{mod } q) }\phantom{-----}$$
where $2 \leq q \leq \exp(n)$ and the errors $e_x$ are i.i.d. random variables drawn according to a symmetric error distribution $\chi_{\eta,q}$ centered around $0$
of noise magnitude at most $\eta$.}
\Output{$\tilde{\vec{s}} \in \Z_q^n$ with probability at least $\varphi(q)/(24\eta q)$, else $\bot$.}
\BlankLine
\begin{enumerate}
\item Query a sample $\ket{\psi_s}$, where for unknown i.i.d. errors $e_x \from \chi_{\eta,q}$,
      $$  \ket{\psi_s} =  \displaystyle\frac{1}{\sqrt{q^n}}\sum_{\vec x \in \Z_q^n} \ket{x_1}\hdots \ket{x_n}\ket{\langle \vec x, \vec s \rangle + e_x \, (\text{mod } q) } \phantom{-------}$$
\item Apply the quantum Fourier transform $\QFT_{\Z_q}^{\otimes n+1}$.\par
\ \\
\item Measure in the computational basis, yielding an outcome $\ket{z_1}\ket{z_2}\hdots \ket{z_{n+1}}$.\\
\ \\
\item \textbf{If} $\text{gcd}(z_{n+1},q)=1$, then output $ \tilde s = ( \frac{-z_1}{z_{n+1}} (\text{mod } q), \frac{-z_2}{z_{n+1}} (\text{mod } q), ...,  \frac{-z_n}{z_{n+1}} (\text{mod } q)),$\\
      \textbf{else} output $\bot$.   
\end{enumerate}
\caption{Extended Bernstein-Vazirani for \LWE
}\label{alg:ExtendedBV_LWE}
\end{algorithm}\DecMargin{1em}

\begin{theorem}
Let $\vec s \in \mathbb{Z}_q^n$ be a secret string and $\mathcal{O}_{EX}(\vec s,\chi)$ be a quantum example oracle that outputs uniform \LWE samples upon a modulus
$2 \leq q \leq \exp(n)$ and i.i.d. errors drawn according to a symmetric error distribution $\chi_{\eta,q}$
of noise magnitude at most $\eta$ centered around $0$.
Then, after querying a single sample, the procedure in \expref{Algorithm}{alg:ExtendedBV_LWE} recovers the secret string $\vec s$ with probability at least $\varphi(q)/(24\eta q)$.
\end{theorem}
\begin{proof}
Let $\vec s \in \mathbb{Z}_q^n$ be a secret string as sampled by $\LWE(n,q,\chi)$.
The example oracle $\mathcal{O}_{EX}(\vec s,\chi)$ generates samples $\ket{\psi_s}$, where for unknown i.i.d. errors $e_x \from \chi_{\eta}$ of noise magnitude $|e_x| \leq \eta$,
$$
\ket{\psi_s} =  \displaystyle\frac{1}{\sqrt{q^n}}\sum_{\vec x \in \Z_q^n} \ket{x_1}\hdots \ket{x_n}\ket{\langle \vec x, \vec s \rangle + e_x \, (\text{mod } q) }.
$$
The procedure Extended Bernstein-Vazirani for \LWE in \expref{Algorithm}{alg:ExtendedBV_LWE} applies the quantum Fourier transform, resulting in a state
$$
\QFT_{\Z_q}^{\otimes n+1} \ket{\psi_s} = \frac{1}{q^{n+1/2}} \sum_{\vec x \in \Z_q^{n}}\sum_{\vec y \in \Z_q^{n+1}}
  \omega^{\braket{\vec  x,\vec y + y_{n+1}\vec k} + e_x y_{n+1}} \ket{y_1}\dots \ket{y_n}\ket{ y_{n+1}}.
$$
Then, the probability of measuring the outcome $\vec z = -y_{n+1}\vec s \pmod q$ is given by
\begin{align*}
\abs{\bra{\vec z}\QFT_{\Z_q}^{\otimes n+1}\ket{\psi_s}}^2 &= \left\| \frac{q^n}{q^{n+1/2}} \sum_{z_{n+1} \in \Z_q} \sum_{x \in \mathbb{Z}_q^n} \omega^{e_x z_{n+1}} \ket{z_1}\dots\ket{z_n}\ket{z_{n+1}} \right\|^2\\
&=\frac{1}{q^{2n+1}} \sum_{z_{n+1} \in \Z_q}
 \left( \sum_{\vec x \in \Z_q^{n}} \Re{\omega^{e_x z_{n+1}}}\right)^2 +
  \left( \sum_{\vec x \in \Z_q^{n}} \Im{\omega^{e_x z_{n+1}}}\right)^2\\
&  \geq \frac{1}{q^{2n+1}} \sum_{z_{n+1} \in \Z_q \atop z_{n+1} \leq \frac{q}{6\eta}}
 \left( \sum_{\vec x \in \Z_q^{n}} \cos \left({\frac{2 \pi e_x z_{n+1}}{q}}\right)\right)^2\\
 & \geq \frac{1}{q^{2n+1}} \sum_{z_{n+1} \in \Z_q \atop z_{n+1} \leq \frac{q}{6\eta}}
 \left( \sum_{\vec x \in \Z_q^{n}} \frac{1}{2}\right)^2  \phantom{---} \left(\text{since }\left|\frac{2 \pi}{q} e_x z_{n+1}\right| < \frac{\pi}{3}, \text{ if  }\, z_{n+1} \leq \frac{q}{6\eta}\right)\\
& \geq \frac{1}{24\eta}.
\end{align*}
Moreover, with probability $\varphi(q)/q$, an outcome $z_{n+1}$ is measured such that $\gcd(z_{n+1},q)=1$. In this case, the procedure Extended Bernstein-Vazirani for \LWE in \expref{Algorithm}{alg:ExtendedBV_LWE} correctly outputs the string $\vec s$.
Therefore, \expref{Algorithm}{alg:ExtendedBV_LWE} succeeds with probability at least $\varphi(q)/(24\eta q)$.
\end{proof}

Finally, we cite the main result in \cite{GK17} that completes the above analysis for the special case where $q$ is a large prime. It is shown that the success probability of the prime modulus variant of the Bernstein-Vazirani Algorithm can indeed be amplified 
to solve \LWE with quantum samples. 

\begin{theorem}[\cite{GK17}, Quantum Algorithm for \LWE]
For symmetric error distributions $\chi_{\eta,q}$ of noise magnitude $\eta = n/24$ around $0$ and prime modulus $q$ in $[2^{n^\gamma},2\cdot 2^{n^\gamma}]$, where $\gamma \in (0,1)$,
the Extended Bernstein-Vazirani algorithm for \LWE can be amplified to solve \LWE$(n,q,\chi)$ towards a success probability of $1-\delta$ by requesting
$O(n \log\frac{1}{\delta})$ many quantum examples and running in time poly$(n,\log\frac{1}{\delta})$.
\end{theorem}

%% file: blinding.tex
\section{Relabeling Games}\label{sec:relabeling}

In this section, we introduce a technique that exploits the effects of relabeling in quantum algorithms, a variation on
the blinding of quantum algorithms considered in a lemma by Alagic et al.\cite{AMRS18},
as well as the 'simulation lemma', the standard query lower bound technique for proving optimality.
The results on the relabeling game offer a useful limitation of all quantum query algorithms,
particularly in a cryptographic setting when proving the security of our proposed constructions.
In fact, as a direct consequence of relabeling, we can prove the indistinguishability of several hybrid games from
the section on post-quantum cryptography.

\subsection{Classical Relabeling.}

First, we consider a toy example to illustrate the effects of relabeling in a classical experiment.
\begin{definition}[\ClassicalRelabeling]\label{def:classical_relabeling} \ \\
Let $f:\{0,1\}^{n} \longrightarrow \{0,1\}^m$ be a function and consider the experiment
$\ClassicalRelabeling(f,\D)$ between a \PPT algorithm $\algo D$ and a challenger as follows:
\begin{enumerate}
\item \emph{(setup)} the challenger generates a bit $b \rand \bit$ and strings $r^*\rand \bit^n$ and $s\rand \bit^m$ ;
\item \textit{(pre-challenge)} \D receives classical oracle access to $\mathcal{O}_f$.
\item \textit{(challenge phase)} depending on the random bit $b$, \D receives the following:
\begin{itemize}
\item $(b=0):$ \D receives a pair $(r^*,f(r^*))$;
\item $(b=1):$ \D receives a pair $(r^*, f(r^*)\oplus s)$.
\end{itemize}
\item \textit{(resolution)} \D outputs $b'$ and wins if $b'=b$.
\end{enumerate}
\end{definition}

\begin{proposition}\label{thm:classical_relabeling}
Let $f: \{0,1\}^{n} \longrightarrow \{0,1\}^m$ be any function. Then, any \PPT algorithm
$\mathcal{D}$ making $T(n) = \poly(n)$ many oracle queries succeeds at $\ClassicalRelabeling(f,\D)$ with
advantage at most $O(T(n)/2^n)$.
\end{proposition}
\begin{proof}
Let $T(n) = \poly(n)$ be an upper bound on the number of queries to $\mathcal{O}_f$.
Note that \D can at best win the game with probability $1/2$, unless the challenge value $r^*$ was previously queried during the pre-challenge phase.
We let \textsc{Quer} denote the event in which the actual challenge pair is queried, and let $\overline{\textsc{Quer}}$ be the event when it is not. 
Then it is easy to see that
\begin{align*}
\Pr[\D \text{ wins } \ClassicalRelabeling] &= \Pr[\D \text{ wins } \ClassicalRelabeling | \textsc{Quer}] \cdot \Pr[\textsc{Quer}]\\
&\phantom{----} +\Pr[\D \text{ wins } \ClassicalRelabeling |\overline{\textsc{Quer}}] \cdot \Pr[\overline{\textsc{Quer}}]\\
&\leq \frac{T(n)}{2^n} + \frac{1}{2}\left(1-\frac{T(n)}{2^n}\right) \, = \, \frac{1}{2} + O\left(\frac{T(n)}{2^n}\right).
\end{align*}
\qed
\end{proof}

\subsection{Relabeling in Quantum Algorithms}

Let us now consider the effects of relabeling with respect to quantum oracles and prove lower bounds on appropriate quantum variants of the previous classical experiment.
In the following, we distinguish between non-adaptive and adaptive variants of relabeling.

\subsubsection{Non-adaptive Relabeling.}

We now define a non-adaptive relabeling game, an experiment in which a quantum algorithm first receives quantum oracle access
to a function and then, upon receiving a random input/output pair, the goal is to decide whether it is genuine (or modified)
based on the previous query phase.

\begin{definition}[Non-adaptive Relabeling Game]\label{def:non_relabeling_game} \ \\
Let $f:\{0,1\}^{n} \longrightarrow \{0,1\}^m$ be a function and let $n,m$ be integers, where $n$ is the security parameter.
We define the non-adaptive experiment $\RelabelingGame^{(1)}(f,\D)$ between a \QPT algorithm $\algo D$ and a challenger as follows:
\begin{enumerate}
\item \emph{(setup)} the challenger generates a bit $b \rand \bit$ and strings $r^*\rand \bit^n$ and $s\rand \bit^m$;
\item \emph{(pre-challenge phase)} \D receives quantum oracle access to $\mathcal{O}_f$;
\item \textit{(challenge phase)} depending on the random bit $b$, \D receives the following:
\begin{itemize}
\item $(b=0):$ \D receives a pair $(r^*,f(r^*))$;
\item $(b=1):$ \D receives a pair $(r^*, f(r^*)\oplus s)$.
\end{itemize}
Then, \D receives an example oracle that outputs classical random pairs $(r,f(r))$.
\item \emph{(resolution)} $\mathcal{D}$ outputs a bit $b'$ and wins the game if $b'=b$.
\end{enumerate}
\end{definition}

We now show how to control the success probability of \D in terms of the number of queries it makes. 
The proof uses a hybrid argument, adapting a standard quantum query lower bound technique to give precise control over the success probability.

\begin{theorem}\label{thm:nonadaptive_relabeling}
Let $f: \{0,1\}^{n} \longrightarrow \{0,1\}^m$ be an arbitrary function. Then, any \QPT algorithm
$\mathcal{D}$ making $T(n) = \poly(n)$ many quantum oracle queries succeeds at the non-adaptive experiment $\RelabelingGame^{(1)}(f,\D)$ with
advantage at most $O(T(n)/\sqrt{2^{n}})$, except with at most negligible probability.
\end{theorem}
\begin{proof}
The view of \D during the game is the following.
In the pre-challenge phase, \D is allowed to make at most $\poly(n)$ many queries to an oracle for $f$.
Then, in the challenge phase, \D receives a random input/output pair for $f$ (possibly relabeled with random $s$) and,
after a final query phase to a random example oracle, \D has to decide whether the challenge is genuine or relabeled.
Note that \D can only truly win $\RelabelingGame^{(1)}$ in the following sense:
\begin{enumerate}
\item \D recovers $(r^*,f(r^*))$ during the pre-challenge phase while quantumly querying $\mathcal O_f$.
\item \D gets lucky and receives $(r^*,f(r^*))$ as one of the random classical queries to the example oracle. 
\end{enumerate}
Note that the second possibility only occurs with at most negligible probability. If $\D$ makes $\poly(n)$ queries to the example oracle,
the probability of receiving the challenge pair is at most $O(\poly(n)/2^n)$, hence negligible.

In order to rule out the first possibility, we consider the relabeled function,
\begin{equation*}
f^*(x) :=
    \begin{cases}
      f(x) \oplus s, &  \text{if } x=r^*, \\
      f(x), & \text{if } x\neq r^*,
    \end{cases}
\end{equation*}
and show that we can replace the pre-challenge oracle with $\mathcal O_{f^*}$ and only negligibly affect the output states generated by \D in the game. In fact, any measurement
of the state produced during the pre-challenge results in negligibly close outcome distributions, irrespective of whether the oracle is relabeled at $r^*$.
Consequently, \D then cannot directly observe a mismatch in the challenge value $r^*$ and only win with probability a most $1/2 + \negl(n)$. We will now prove this claim.

We can write any quantum algorithm $\mathcal{D}$ interacting with $\mathcal O_f$ as a sequence of $T(n)=\poly(n)$ oracle queries and
unitary operations $U_0,U_1,U_2, \dotsc, U_T$ upon an initial state $\ket{\psi_0}$ that result in an output state
\begin{equation}\label{query_sequence_phi}
\ket{\psi} = U_T \mathcal{O}_{f} U_{T-1}  \dotsb
                   U_{1} \mathcal{O}_{f} U_{0} \, \ket{\psi_0}.
\end{equation}
We adopt a hybrid approach and show that the output states $\ket{\psi}$ remain close, irrespective of whether the oracle queries are answered with $f$ or $f^*$.
First, we argue that replacing the functionality of only a single oracle query results in statistically close output distributions.
To this end, we define the $k$-th hybrid state as:
\begin{equation}
\ket{\psi^{(k)}} = U_T \mathcal{O}_{f^*} U_{T-1} \dotsb \mathcal{O}_{f^*} U_k \mathcal{O}_f \dotsb \mathcal{O}_f U_0 \ket{\psi_0}.
\end{equation}
Now we can bound two successive hybrids by using the invariance of the trace distance with respect to simultaneous unitary transformations:
\begin{eqnarray*}
\delta( \ket{\psi^{(k)}}, \ket{\psi^{(k-1)}})  &=& \delta(U_T \mathcal{O}_{f^*} \dotsb \mathcal{O}_{f^*} U_k \mathcal{O}_f \dotsb \mathcal{O}_f U_0 \ket{\psi_0},U_T \mathcal{O}_{f^*} \dotsb \mathcal{O}_{f^*} U_{k-1} \mathcal{O}_f \dotsb \mathcal{O}_f U_0 \ket{\psi_0} ) \\
 &=& \delta( \mathcal{O}_{f} U_{k-1}  \dotsb \mathcal{O}_f U_0 \ket{\psi_0}, \mathcal{O}_{f^*} U_{k-1} \mathcal{O}_f \dotsb \mathcal{O}_f U_0 \ket{\psi_0} ) \\
 &=& \delta( \mathcal{O}_{f} \ket{\psi^{(k-1)}}, \mathcal{O}_{f^*} \ket{\psi^{(k-1)}} )\\
  &=& \delta( \ket{\psi^{(k-1)}}, \mathcal{O}_{f}\mathcal{O}_{f^*}\ket{\psi^{(k-1)}}).
\end{eqnarray*}
We adopt the notation $\ket{\psi_f}$ and $\ket{\psi_{f^*}}$ to denote output states that are generated consistently by
the respective oracles throughout the pre-challenge phase.
Using the triangle inequality, we can bound the total expected distance between the output states over
all hybrids as follows:
\begin{equation}
 \mathop{\displaystyle\mathbb{E}} \left[ \delta( \ket{\psi_f}, \ket{\psi_{f^*}}) \right] \leq
 \mathop{\displaystyle\mathbb{E}} \left[ \sum_{k=1}^T \delta( \ket{\psi^{(k)}}, \ket{\psi^{(k-1)}}) \right] =
 \sum_{k=1}^T  \mathop{\displaystyle\mathbb{E}} \left[ \delta( \ket{\psi^{(k)}}, \ket{\psi^{(k-1)}}) \right].
\end{equation}
Let us now consider an oracle $\mathcal{O}_{S^*}$ for the relabeled function $S^*: \{0,1\}^{n} \longrightarrow \{0,1\}^{m}$ such that,
for fixed $r^* \rand  \{0,1\}^n$ and $s \rand \{0,1\}^m$,
\begin{equation}
S^*(x) :=
    \begin{cases}
      s, &  \text{if } x=r^*, \\
      0, & \text{if } x\neq r^*,
    \end{cases}
\end{equation}
Using the fact that any state in the $k$-th hybrid computed prior to $U_k$, in particular during the pre-challenge state, is completely independent
of $S^*$, we can bound any successive hybrids as follows:
\begin{eqnarray}
\delta( \ket{\psi^{(k)}}, \ket{\psi^{(k-1)}})
  &=& \delta( \ket{\psi^{(k-1)}}, \mathcal{O}_{f}\mathcal{O}_{f^*}\ket{\psi^{(k-1)}})\\
  &=& \delta( \ket{\psi^{(k-1)}}, \mathcal{O}_{S^*}\ket{\psi^{(k-1)}} ) \\
  &\leq& \mathop{\max_{\ket{\psi}}} \delta( \ket{\psi}, \mathcal{O}_{S^*} \ket{\psi}).
\end{eqnarray}
Therefore, we can find the following upper bound for the total expected trace distance:
\begin{eqnarray}
 \mathop{\displaystyle\mathbb{E}} \left[ \delta( \ket{\psi_f}, \ket{\psi_{f^*}}) \right] &\leq& T(n) \mathop{\max_{\ket{\psi}}} \mathop{\displaystyle\mathbb{E}} \left[  \delta( \ket{\psi}, \mathcal{O}_{S^*} \ket{\psi}) \right]  \\
 &= & T(n) \mathop{\max_{\ket{\psi}}} \mathop{\displaystyle\mathbb{E}} \left[  \sqrt{ 1 - |\bra{\psi} \mathcal{O}_{S^*} \ket{\psi}|^2} \right] \\
 & \leq & T(n) \mathop{\max_{\ket{\psi}}}  \sqrt{ 1 - \mathop{\displaystyle\mathbb{E}} \left[ |\bra{\psi} \mathcal{O}_{S^*} \ket{\psi}| \right]^2}. \label{max_state}
\end{eqnarray}
Consider a projection operator $\Pi_{*}$ onto the relevant subspace with respect to $S^*$, as
given by $\supp \mathcal{O}_{S^*} = \Span\{ \, \ket{r^*}\otimes\ket{y} \, | \, y\in\{0,1\}^{m}\}$. We can
write the oracle $\mathcal{O}_{S^*}$ as the identity operator, except on the range of $\Pi_{*}$. Using the reverse triangle inequality, we find:
\begin{eqnarray}
|\bra{\psi} \mathcal{O}_{S^*}\ket{\psi}| & = &  |\bra{\psi} \mathcal{O}_{S^*} \Pi_{*} \ket{\psi} + \bra{\psi} \mathcal{O}_{S^*} (\one - \Pi_{*})\ket{\psi}| \\
 & \geq &  -|\bra{\psi} \mathcal{O}_{S^*} \Pi_{*}\ket{\psi}| + |1 - \bra{\psi} \Pi_{*} \ket{\psi}| \\
    & \geq & 1 - 2 \bra{\psi} \Pi_{*} \ket{\psi}.
\end{eqnarray}
Consequently, irrespective of the output state produced by algorithm \D, we can now bound the expectation of the trace distance as follows:
\begin{eqnarray}
 \mathop{\displaystyle\mathbb{E}} \left[ \delta( \ket{\psi_f}, \ket{\psi_{f^*}}) \right]
&\leq& T(n) \mathop{\max_{\ket{\psi}}}  \sqrt{ 1 - \mathop{\displaystyle\mathbb{E}} \left[|\bra{\psi} \mathcal{O}_{S^*} \ket{\psi}| \right]^2} \\
&\leq& T(n) \mathop{\max_{\ket{\psi}}}  \sqrt{ 1 - \left( 1 - 2 \mathop{\displaystyle\mathbb{E}}\left[\bra{\psi} \Pi_{*} \ket{\psi}\right] \right)^2} \\
&\leq& T(n) \sqrt{1-\left(1-\frac{2}{2^n}\right)^2} \, \leq \, \frac{2 T(n) }{\sqrt{2^n}}.\label{eq:upper_bound}
\end{eqnarray}
Using Markov's inequality, we can thus conclude that the probability of observing any significant trace distance (i.e. larger than negligible) is still negligible.
In fact, for any constant $\epsilon >0$, a magnitude of $2^{-\epsilon n}$ is negligible and
\begin{equation}
\label{distribution_bound}
\Pr_{s,r^*}\left[ \delta( \ket{\psi_f}, \ket{\psi_{f^*}}) \geq \frac{2 T(n) }{2^{n/2-\epsilon n}} \right]
\, \leq \, \frac{2^{n/2-\epsilon n}\mathop{\displaystyle\mathbb{E}}\left[ \delta( \ket{\psi_f}, \ket{\psi_{f^*}}) \right]}{2 T(n)}
\leq   \,  2^{-\epsilon n}.
\end{equation}
In the above, both the probability and the expectation are taken over all respective outcomes of sampling $r^* \rand  \{0,1\}^n$ and $s \rand \{0,1\}^m$.
Therefore, except with negligible probability, $\delta( \ket{\psi_f}, \ket{\psi_{f^*}}) = O(T(n)/\sqrt{2^{n}})$.
Finally, from \expref{Lemma}{lem:trace_bound}, it follows that any generalized \POVM measurement of the final output states of $\mathcal{D}$ reveals
statistically close outcome distributions, and
\begin{align}
\left|\Pr[ \mathcal{D}^{f}(1^n) = 1] -  \Pr_{s,r^*}[ \mathcal{D}^{f^*}(1^n) = 1] \right| \, \leq \, O(T(n)/\sqrt{2^{n}}) + \negl(n).
\end{align}
Therefore, \D cannot directly observe a mismatch in the challenge value $r^*$ and only at most succeed with the advantage above.

\qed
\end{proof}

\subsubsection{Adaptive Relabeling.}

In the following, we state a stronger variant of the relabeling game, a setting in which a quantum algorithm receives an arbitrary 
advice state (possibly even exponential-sized) for a function
and the goal is to detect whether it was relabeled at a small subset of its domain by querying its oracle.
Unlike in the previous non-adaptive variant, any distinguisher is able to adaptively make queries based on prior information on the target function from the pre-challenge phase.
Note that, in the following, $"||"$ refers to the concatenation of strings. We define the adaptive relabeling game as follows:

\begin{definition}[Adaptive Relabeling Game]\label{def:adaptive_relabeling_game} \ \\
Let $f:\{0,1\}^{n} \longrightarrow \{0,1\}^m$ be a function, $\ket{\psi^f}$ an arbitrary quantum advice state (possibly depending on $f$), and let
$\mu$ be an integer where $0 \leq \mu \leq n$. 
We define the adaptive experiment $\RelabelingGame^{(2)}(f,\mu,\D,\ket{\psi^f})$ between a \QPT algorithm $\algo D$ and a challenger as follows:
\begin{enumerate}
\item \emph{(setup)} \D receives the state $\ket{\psi^f}$; the challenger generates a bit $b \rand \bit$ and two strings $s\rand \bit^m$ and $r^*\rand \bit^\mu$;
\item \emph{(challenge phase)} 
depending on the random bit $b$, \D receives the following:
\begin{itemize}
\item $(b=0):$ \D receives quantum oracle access to $\mathcal O_f$;
\item $(b=1):$ \D receives quantum oracle access to $\mathcal O_{f^*}$,
where $f^*$ is the relabeled function,
\begin{equation*}
f^*(x) :=
    \begin{cases}
      f(x) \oplus s &  \text{if the last $\mu$ bits of $x$ are equal to $r^*$}, \\
      f(x) & \text{otherwise}.
    \end{cases}
\end{equation*}
\end{itemize}
\item \emph{(resolution)} $\mathcal{D}$ outputs a bit $b'$ and wins the game if $b'=b$.
\end{enumerate}
\end{definition}

We now show how to control the success probability of \D in terms of the number of queries it makes. 

\begin{theorem}\label{thm:adaptive_relabeling}
Let $f: \{0,1\}^{n} \longrightarrow \{0,1\}^m$ be any function, let $\ket{\psi^f}$ be an arbitrary advice state (possibly depending on $f$)
and let $\mu$ be an integer where $0 \leq \mu \leq n$. Then, any \QPT algorithm
$\mathcal{D}$ making $T(n) = \poly(n)$ many oracle queries succeeds at $\RelabelingGame^{(2)}(f,\mu,\D,\ket{\psi^f})$ with
advantage $O(T(n)/\sqrt{2^{\mu}})$, except with at most negligible probability.
\end{theorem}

\begin{proof}
We follow a similar proof as in \expref{Theorem}{thm:nonadaptive_relabeling} in order to show that any measurement of the output
states generated by $\mathcal{D}$ results in statistically close output distributions, irrespective of whether the oracle is relabeled.
The view of \D during the game is the following.
In the pre-challenge phase, \D receives an arbitrary advice state $\ket{\psi^f}$ that contains information on $f$.
Then, during the challenge phase, $\mathcal{D}$ receives an oracle $\mathcal{O}_{\varphi}$,
where $\varphi$ is a function $\varphi: \{0,1\}^{n} \longrightarrow \{0,1\}^m$
and the goal is to determine whether $\varphi = f$ or $\varphi = f^*$, for some $r^* \rand \{0,1\}^\mu$ and $s \rand  \{0,1\}^m$.
Upon receiving $\ket{\psi^f}$, we can write any quantum algorithm $\mathcal{D}$ as a sequence of $T(n)=\poly(n)$ oracle queries and
unitary operations $U_0,U_1,U_2, \dotsc, U_T$ that result in an output state
\begin{equation}\label{query_sequence_phi}
\ket{\psi_\varphi} = U_T \mathcal{O}_{\varphi} U_{T-1}  \dotsb
                   U_{1} \mathcal{O}_{\varphi} U_{0} \, \ket{\psi^f}.
\end{equation}
We adopt a hybrid approach and show that the output states $\ket{\psi_\varphi^f}$ remain close, irrespective of whether $\varphi = f$ or $\varphi = f^*$.
First, we argue that replacing the functionality of only a single oracle query results in statistically close output distributions.
To this end, we define the $k$-th hybrid state as:
\begin{equation}
\ket{\psi^{(k)}} = U_T \mathcal{O}_{f^*} U_{T-1} \dotsb \mathcal{O}_{f^*} U_k \mathcal{O}_f \dotsb \mathcal{O}_f U_0 \ket{\psi^f}.
\end{equation}
Similar to the proof of \expref{Theorem}{thm:nonadaptive_relabeling}, we can bound two successive hybrids by using the invariance of the trace distance with respect to simultaneous unitary transformations:
\begin{eqnarray*}
\delta( \ket{\psi^{(k)}}, \ket{\psi^{(k-1)}}) = \delta( \ket{\psi^{(k-1)}}, \mathcal{O}_{f}\mathcal{O}_{f^*}\ket{\psi^{(k-1)}}).
\end{eqnarray*}
We adopt the notation $\ket{\psi_f}$ and $\ket{\psi_{f^*}}$ to denote output states that are generated consistently by
the respective oracles throughout the entire challenge.
Using the triangle inequality, we can bound the total expected distance between the output states over
all hybrids as follows:
\begin{equation}
 \mathop{\displaystyle\mathbb{E}} \left[ \delta( \ket{\psi_f}, \ket{\psi_{f^*}}) \right] \leq
 \mathop{\displaystyle\mathbb{E}} \left[ \sum_{k=1}^T \delta( \ket{\psi^{(k)}}, \ket{\psi^{(k-1)}}) \right] =
 \sum_{k=1}^T  \mathop{\displaystyle\mathbb{E}} \left[ \delta( \ket{\psi^{(k)}}, \ket{\psi^{(k-1)}}) \right].
\end{equation}
Let us now consider an oracle $\mathcal{O}_{S^*}$ for the relabeled function $S^*: \{0,1\}^{n} \longrightarrow \{0,1\}^{m}$ such that,
for fixed $r^* \rand  \{0,1\}^\mu$ and $s \rand \{0,1\}^m$, $S^*$ is defined by
\begin{equation}
S^*(x) :=
    \begin{cases}
      s & \text{if the last $\mu$ bits of $x$ are equal to $r^*$}, \\
      0 & \text{otherwise}.
    \end{cases}
\end{equation}
We now observe the following crucial fact. Any state in the $k$-th hybrid computed prior to $U_k$, particularly the advice state $\ket{\psi^f}$ that contains information on $f$, 
is completely independent of $S^*$, hence we can bound any successive hybrids as follows:
\begin{eqnarray}
\delta( \ket{\psi^{(k)}}, \ket{\psi^{(k-1)}})
  &=& \delta( \ket{\psi^{(k-1)}}, \mathcal{O}_{f}\mathcal{O}_{f^*}\ket{\psi^{(k-1)}})\\
  &=& \delta( \ket{\psi^{(k-1)}}, \mathcal{O}_{S^*}\ket{\psi^{(k-1)}} ) \\
  &\leq& \mathop{\max_{\ket{\psi}}} \delta( \ket{\psi}, \mathcal{O}_{S^*} \ket{\psi}).
\end{eqnarray}
Therefore, we can find the following upper bound for the total expected trace distance:
\begin{eqnarray}
 \mathop{\displaystyle\mathbb{E}} \left[ \delta( \ket{\psi_f}, \ket{\psi_{f^*}}) \right] &\leq& T(n) \mathop{\max_{\ket{\psi}}} \mathop{\displaystyle\mathbb{E}} \left[  \delta( \ket{\psi}, \mathcal{O}_{S^*} \ket{\psi}) \right]  \\
 &= & T(n) \mathop{\max_{\ket{\psi}}} \mathop{\displaystyle\mathbb{E}} \left[  \sqrt{ 1 - |\bra{\psi} \mathcal{O}_{S^*} \ket{\psi}|^2} \right] \\
 & \leq & T(n) \mathop{\max_{\ket{\psi}}}  \sqrt{ 1 - \mathop{\displaystyle\mathbb{E}} \left[ |\bra{\psi} \mathcal{O}_{S^*} \ket{\psi}| \right]^2}. \label{max_state}
\end{eqnarray}
As in the proof of \expref{Theorem}{thm:nonadaptive_relabeling}, we project onto the relevant subspace with respect to $S^*$, as
given by $\supp \mathcal{O}_{S^*} = \Span\{ \, \ket{(z\, || \, r^*)}\otimes \ket{y} \, | \, z\in\{0,1\}^{n-\mu},\,y\in\{0,1\}^{m}\}$. We can
write the oracle $\mathcal{O}_{S^*}$ as the identity operator, except on the range of $\Pi_{*}$.
Consequently, irrespective of the output state produced by algorithm \D, we can now bound the expectation of the trace distance as follows:
\begin{eqnarray}
 \mathop{\displaystyle\mathbb{E}} \left[ \delta( \ket{\psi_f}, \ket{\psi_{f^*}}) \right] \leq T(n) \mathop{\max_{\ket{\psi}}}  \sqrt{ 1 - \left( 1 - 2 \mathop{\displaystyle\mathbb{E}}\left[\bra{\psi} \Pi_{*} \ket{\psi}\right] \right)^2} \leq \frac{2 T(n) }{\sqrt{2^\mu}}.\label{eq:upper_bound}
\end{eqnarray}
Using Markov's inequality, we can thus conclude that the probability of observing any significant trace distance (i.e. larger than negligible) is still negligible.
In fact, for any constant $\epsilon >0$, a magnitude of $2^{-\epsilon n}$ is negligible and
\begin{equation}
\label{distribution_bound}
\Pr_{s,r^*}\left[ \delta( \ket{\psi_f}, \ket{\psi_{f^*}}) \geq \frac{2 T(n) }{2^{\mu/2-\epsilon n}} \right]
\, \leq \, \frac{2^{\mu/2-\epsilon n}\mathop{\displaystyle\mathbb{E}}\left[ \delta( \ket{\psi_f}, \ket{\psi_{f^*}}) \right]}{2 T(n)}
\leq   \,  2^{-\epsilon n}.
\end{equation}
In the above, both the probability and the expectation are taken over all respective outcomes of sampling $r^* \rand  \{0,1\}^\mu$ and $s \rand \{0,1\}^m$.
Therefore, except with negligible probability, $\delta( \ket{\psi_f}, \ket{\psi_{f^*}}) = O(T(n)/\sqrt{2^{\mu}})$.
Finally, from \expref{Lemma}{lem:trace_bound}, it follows that any generalized \POVM measurement of the final output states of $\mathcal{D}$ reveals
statistically close outcome distributions, and
\begin{align}
\left|\Pr[ \mathcal{D}^{f}(1^n) = 1] -  \Pr_{s,r^*}[ \mathcal{D}^{f^*}(1^n) = 1] \right| \, \leq \, O(T(n)/\sqrt{2^{\mu}}) + \negl(n).
\end{align}
We observe that, in the special case when $\mu=n$, any \QPT $\algo D$ making $T=\poly(n)$ queries achieves an overall negligible advantage in the game.
\qed
\end{proof}
Note that the above bound is in fact tight in comparison with quantum searching using Grover's algorithm or amplitude amplification, both achieving the same success probability and proven to be optimal.

%% file: blinding_Gorjan.tex
\section{Post-Quantum Cryptography}\label{ch:pq_crypto}

Let us now extend the security notions behind chosen-ciphertext attacks from Chapter \ref{ch:security} to
a world of quantum computers.
In particular, we consider adversaries who receive quantum oracle access to both encryption and decryption at various times during the security game.
While the case of quantum \CCAA-security has already been introduced in \cite{BZ13}, we investigate a less powerful model
by considering \textit{non-adaptive quantum chosen-ciphertext attacks}.
In this security notion of quantum \CCA, an adversary is given quantum superposition access to both encryption and decryption
prior to the challenge phase, followed by a final phase of adaptive challenge access to the encryption oracle.

Our goal is to exploit the blindness of quantum query algorithms towards the class of functions that only differ at a single location
in order to provide secure constructions under a non-adaptive quantum chosen-ciphertext attack. 
To this end, we first define both the \INDQCCA, as well as the \SEMQCCA security game, and then propose
schemes based on quantum-secure pseudorandom functions and permutations that fulfill our definitions.\\

\subsection{Security Under Non-adaptive Quantum Chosen-Ciphertext Attacks}

In this section, we extend the definitions from Chapter \ref{ch:security} and introduce notions of security in the context of quantum adversaries. 
In providing a quantum encryption oracle $\Enc_k$, each query is answered by choosing a randomness and encrypting
each message in the superposition from the $r-$family of unitary operations such that:
\begin{align}\label{pq_enc}
\Enc: \sum_{m,c} \alpha_{m,c} \ket{m}\ket{c} \longrightarrow \sum_{m,c} \alpha_{m,c} \ket{m}\ket{c \oplus \Enc_k(m;r)}
\end{align}
Typically, we consider the case of sampling a randomness $r \rand \bit^n$ of equal length to a message space, where $m \in \bit^n$.
Moreover, we consider the quantum decryption oracle $\Dec_k$ to be deterministic, hence each oracle query is answered upon a superposition of ciphers as follows:
\begin{align}\label{pq_dec}
\Dec: \sum_{c,s} \beta_{c,s} \ket{c}\ket{s} \longrightarrow \sum_{c,s} \beta_{c,s} \ket{c}\ket{s \oplus \Dec_k(c)} \phantom{kk}
\end{align}
Note that, since $\Enc_k$ and $\Dec_k$ are required to be \PPT algorithms provided by the underlying
symmetric-key encryption scheme, both \eqref{pq_enc} and \eqref{pq_dec} correspond to efficient and reversible 
quantum operations. For the remainder of this chapter, we adopt the convenient notation $\Enc_k$ and $\Dec_k$ in order to refer to the above quantum oracles for
encryption and decryption.

\subsubsection{Indistinguishability.}

We begin by first introducing a notion of indistinguishability in the context of a quantum chosen-ciphertext attacks.
\begin{definition}[\INDQCCA] \label{ind-qcca}\ \\
Let $\Pi = (\KeyGen, \Enc, \Dec)$ be a symmetric-key encryption scheme and consider the
$\INDGame$ between a \QPT adversary \A and challenger \C, defined as follows:
\begin{enumerate}
\item \emph{(initial phase)} On input $1^n$, \C generates a key $k \from \KeyGen(1^n)$ and a bit $b \inrand \bit$;
\item \emph{(pre-challenge phase)} $\algo A$ receives oracles $\Enc_k$ and $\Dec_k$, then sends $(m_0, m_1)$ to \C;
\item \emph{(challenge phase)} \C replies with $\Enc_k(m_b)$ and \A receives an oracle for $\Enc_k$ only;
\item \emph{(resolution phase)} $\algo A$ outputs a bit $b'$, and wins if $b = b'$.
\end{enumerate}
We say $\Pi$ has indistinguishable encryptions under non-adaptive quantum chosen-ciphertext attack (or is \INDQCCA-secure) if, for every \QPT \A, 
there exists a negligible function $\negl(n)$ such that:
$\Pr[\A \text{ wins } \INDGame] \leq 1/2 + \negl(n)$.
\end{definition}

\subsubsection{Semantic Security.}

In this section, we first introduce semantic security under a \QCCA learning phase and then 
prove an important equivalence between our notion of indistinguishability
and semantic security.

\begin{definition}[\SEMQCCA]\label{def:sem_qcca} Let $\Pi = (\KeyGen, \Enc, \Dec)$ be an encryption scheme, and consider the experiment 
$\SEMGame$ with a \QPT $\algo A$, defined as follows.
\begin{enumerate}
\item \emph{(initial phase)} A key $k \from \KeyGen(1^n)$ and bit $b \rand \bit$ are generated;
\item \emph{(pre-challenge phase)} $\algo A$ receives access to oracles $\Enc_k$ and $\Dec_k$, then outputs a classical challenge template consisting of $(\Samp, h, f)$;
\item \emph{(challenge phase)} A plaintext $m \from \Samp$ is generated; $\algo A$ receives $h(m)$ and an oracle for $\Enc_k$ only; if $b = 1$, $\algo A$ also receives $\Enc_k(m)$.
\item \emph{(resolution)} $\algo A$ outputs a string $s$, and wins if $s = f(m)$.
\end{enumerate}
We say $\Pi$ is semantically secure under non-adaptive quantum chosen ciphertext attack (or is \SEMQCCA) if, for every \QPT $\algo A$, 
there exists a \QPT $\algo S$ such that the challenge templates output by $\algo A$ and $\algo S$ are identically distributed, and there exists a negligible
function $\negl(n)$ such that:
\begin{align*}
\left| \underset{k \rand \K}{\Pr}[\A (1^n,\Enc_k(m),h(m)) = f(m)] \, - \Pr[\Sim (1^n,|m|,h(m)) = f(m)] \right| \, \leq \, \negl(n),
\end{align*}
where, in both cases, the probability is taken over plaintexts $m \leftarrow \Samp$.
\end{definition}

\subsubsection{Equivalence of Indistinguishability and Semantic Security.}\label{ch:equiv}

Let us now show the equivalence of the new security notions we introduced in this chapter. In the following, we provide
a standard proof, following related approaches in \cite{GHS16}.

\begin{theorem}
Let $\Pi=(\KeyGen,\Enc, \Dec)$ be a symmetric-key encryption scheme. Then, $\Pi$ is $\INDQCCA$-secure if and only if $\Pi$ is $\SEMQCCA$-secure.
\end{theorem}
\begin{proof}
Suppose $\Pi$ is $\INDQCCA$-secure, i.e. has indistinguishable encryptions. Let \A be a \QPT algorithm against \SEM 
that receives a challenge $\Enc_k(m)$. Define a \QPT simulator \Sim that also challenges \SEM but simply runs \A as a subroutine as follows: 
Instead of receiving $\Enc_k(m)$ during the \SEM challenge, \Sim relies 
only on the side information
$h(m)$, in particular the plaintext length $|m|$, and simulates \A's encryption and decryption oracles by making use of its own $\QCCA$ learning phase. 
At the challenge phase, \Sim simply encrypts the string $1^{|m|}$ and forwards $\Enc_k(1^{|m|})$ to \A.
After another emulated learning phase, \Sim finally outputs the same target $f(m)$ that \A outputs. 
Since $\Pi$ has indistuinguishable encryptions by assumption,
\A's success probability must be negligibly close to the original \SEM game of \Sim.

Now, suppose $\Pi$ is not $\INDQCCA$-secure, hence there exists a \QPT distinguisher \A against $\INDQCCA$-security.
This allows us to build a distinguisher \D running \A as a subroutine against the \SEM security game as follows:
By using its oracles from the \QCCA learning phases, \D simulates the $\INDQCCA$ security game of \A by simply forwarding all queries to its own oracles.
At the \QIND challenge phase, \A prepares two messages $(m_0,m_1)$ and presents them to \D. Then, \D prepares a \SEM challenge template
$(U,h,f)$, where $U$ describes the uniform distribution over plaintexts $\{m_0,m_1\}$, the side information is given by the
length of the messages and where
the target function $f(m)$ concerns the function that distinguishes between $m_0$ and $m_1$, i.e. $f(m_0)=0$ and $f(m_1)=1$. Using this
\SEM template, \D receives a ciphertext $\Enc_k(m)$ and presents it to \A as an \INDQCCA challenge. Finally, \D simply
outputs whatever target bit \A outputs. By assumption, \A succeeds with nonnegligible probability and therefore \D breaks
the \SEMQCCA security game.
\end{proof}

\subsection{Quantum-secure Pseudorandom Functions}

In order to find constructions for quantum-secure symmetric-key cryptography,
we require the use of appropriate building blocks.
Let us now extend the concept of secure pseudorandom functions from Chapter \ref{ch:prf} to quantum-secure pseudorandom functions, a variant in which
an adversary in possession of a quantum computer can query the function on a superposition of inputs.
Surprisingly, one can find quantum-secure constructions for pseudorandom functions that are secure in this model. We then compare the definition
to post-quantum secure pseudorandom functions that do not require security against quantum superpositions.

\begin{definition}[Quantum-secure Pseudorandom Function]\label{QPRF}\ \\
Let $\mathcal F$ be an efficiently computable function $\mathcal F : \mathcal{K} \times \mathcal{X} \rightarrow \mathcal{Y}$ 
on a key-space $\mathcal{K}$, a domain $\mathcal{X}$ and a range $\mathcal{Y}$.
We say $\mathcal F= \{f_k\}_{k \in \K}$ is a family of quantum-secure pseudorandom functions $(\QPRF)$
if, for all $k$ and \QPT distinguishers \D,
there exists a negligible function $\negl(n)$ such that:
\begin{align}
\left| \underset{k \rand \mathcal{K}}{\Pr}[ \mathcal{D}^{f_k}(1^n) = 1] \, - \underset{f \rand \{F: \mathcal{X} \rightarrow \mathcal{Y}\}}{\Pr}[ \mathcal{D}^{f}(1^n) = 1] \right| \, \leq \, \negl(n)
\end{align}
\end{definition}
Similarly, we define post-quantum secure pseudorandom functions as a variant in which the adversary is quantum, but access to the function remains classical.

Finally, we provide schemes based on the above building blocks of pseudorandom functions
that are quantumly secure under a quantum-chosen ciphertext attack we 
introduced in this thesis.

\subsection{Secure Constructions}

We consider the \PRF scheme in \expref{Construction}{cons:prf} and show that it
is secure under a non-adaptive quantum chosen-ciphertext attack. The intuition is that,
once the pseudorandom function is taken to be quantum-secure, the pre-challenge phase
reveals at most a polynomial amount of evaluations of the pseudorandom function, despite
the presence of a decryption oracle. 
Note that, in this scheme, the encryption oracle only reveals a single functional evaluation of the \PRF at a random location at a time,
while the decryption oracle can additionally serve as a quantum oracle for the function of the underlying encryption scheme.
In fact, the adversary can generate superposition queries to a \PRF $f_k$ over the entire input space by simply 
initializing both the ciphertext and plaintext register to the all-zero state and
preparing a superposition over the random register:
\begin{equation}
\Dec_k:\,  \sum_{\tau \in \{0,1\}^n} \beta_\tau \ket{0}\ket{\tau}\ket{0} \longrightarrow \sum_{\tau \in \{0,1\}^n} \beta_\tau \ket{0}\ket{\tau}\ket{f_k(\tau)}.
\end{equation}
At first sight, however, it is not clear whether being able to generate superpositions gives the adversary additional power
during the challenge phase. Therefore, we have to bound the amount of information that 
quantum query algorithms can learn in a suitable way. Our contribution is to show that this advantage is still negligible
and that \expref{Construction}{cons:prf} has indistinguishable encryptions, even in the presence of decryption oracles 
under a non-adaptive quantum chosen-ciphertext attack. Finally, due to the equivalence results from the previous section,
our proposed scheme then also satisfies indistinguishability of encryptions and semantic security.

\begin{theorem}\label{thm:prf_scheme}
Let $\mathcal{F}$ be a family of quantum-secure pseudorandom functions. Then, the scheme
$\Pi[\mathcal{F}]=(\KeyGen,\Enc,\Dec)$ from \expref{Construction}{cons:prf} is $\INDQCCA$-secure.
\end{theorem}

\begin{proof}
We show that any \QPT adversary $\mathcal{A}$ wins \INDGame
under a \QCCA learning phase with at most negligible probability.
To this end, we introduce a sequence of indistinguishable hybrid games until we arrive at a security game in which 
the challenge is perfectly hidden and the adversary cannot win. Let us sketch the proof.

First, we replace the \QPRF $\mathcal{F}$ from the standard security game with a perfectly random function family. According to \QPRF security, 
we only negligibly affect the overall success probability of the adversary. In the \INDCCA model in which the adversary is classical, 
this hybrid already completes the proof, as the probability of observing the same randomness outside of the challenge is negligible. 
However, since the adversary in the quantum model is allowed to issue quantum queries, the previous observation does not easily account for 
an \QPT adversary who can evaluate the function in superposition over the entire input space.
We introduce a final hybrid game in which we choose the challenge randomness at the start of the game. Then, we relabel the challenge value with
a uniformly random string, while the challenge bit remains the same.
Due to \expref{Theorem}{thm:nonadaptive_relabeling}, the overall advantage of \A in the game is at most negligible. In fact, no \QPT algorithm can distinguish between these two hybrids with nonnegligible probability,
else it would imply a successful distinguisher against $\RelabelingGame^{(1)}$.

Fix a \QPT adversary \A against $\Pi[\mathcal F]$ and let $n$ denote the security parameter. It will be convenient to split \A into the pre-challenge algorithm $\algo A_1$ and the post-challenge algorithm $\algo A_2$.
We define several hybrid games, as follows:

\begin{description}
\item[\textsc{Game 0}:] This is the standard \INDQCCA security game, i.e., $\IndGame(\Pi[\mathcal F], \algo A, n)$. 
In the pre-challenge phase, $\algo A_1$ receives oracles $\Enc_k$ and $\Dec_k$. 
In the challenge phase, $\algo A_1$ outputs $(m_0, m_1)$ and its private data $\ket{\psi}$; a random bit $b \inrand \bit$ is sampled, and $\algo A_2$ is run on input $\ket{\psi}$ 
and a challenge ciphertext,
$$
c_b := \Enc_k(m_b) = (r^*, f_k(r^*) \oplus m_b),
$$
where $r^* \rand \bit^n$ is sampled uniformly at random. In the post-challenge phase, $\algo A_2$ only has access to $\Enc_k$ and must output a bit $b'$.

\item[\textsc{Game 1}:] This is the same game as \textsc{Game 0}, except we replace $f_k$ with a perfectly random function $f$ with the same domain and range.

\item[\textsc{Game 2}:] This is the same game as \textsc{Game 1}, except the challenge is now relabeled and answered with $(r^*,f(r^*)\oplus s \oplus m_b)$, while
the challenge bit $b$ remains identical to \textsc{Game 1}.
\end{description}

Let us first observe that $s$ is a uniformly random string that is independent of all other random variables in \textsc{Game 2}. 
By the information-theoretic security of the one-time pad, it follows that the success probability of \A in \textsc{Game 2} is at most $1/2$.

Next, we show that the output of \textsc{Game 0} is at most negligibly different from the output of \textsc{Game 1}. We do this by constructing a quantum oracle distinguisher $\algo D$ between the families $\mathcal F$ and $\mathcal R$ with distinguishing advantage
$$
\Delta^\D(\textsc{Game 0},\textsc{Game 1}) := \left| \underset{k \rand \bit^n}{\Pr}[ \mathcal{D}^{f_k}(1^n) = 1] \, - \underset{f \rand \mathcal R}{\Pr}[ \mathcal{D}^{f}(1^n) = 1] \right|
$$
which must then be negligible since $\mathcal F$ is a \QPRF. The distinguisher $\algo D$ receives quantum oracle access to a function $\varphi$ (sampled from either $\mathcal F$ or $\mathcal R$) and proceeds by simulating $\IndGame$ as follows:
\begin{enumerate}
\item Run $\algo A_1$, answering encryption queries using (random) classical calls to $\varphi$, and answering decryption queries using quantum oracle calls to $\varphi$, such that
$$
\ket{c_1}\ket{c_2}\ket{m}
\mapsto \ket{c_1}\ket{c_2}\ket{m \oplus c_2}
\mapsto \ket{c_1}\ket{c_2}\ket{m \oplus c_2 \oplus \varphi(c_1)}\,.
$$
\item simulate the challenge phase by sampling $b \inrand \bit$ and encrypting the challenge using a classical (random) call to $\varphi$;
\item run $\algo A_2$ and simulate the post-challenge phase by continuing to answer the encryption queries as before;
\item when $\algo A_2$ outputs $b'$, so does the distinguisher $\D$.
\end{enumerate}
Since the distinguisher $\D$ is a \QPT algorithm and $\mathcal F$ is a \QPRF, the distinguishing advantage $\Delta^\D(\textsc{Game 0},\textsc{Game 1})$ is at most $\negl(n)$.
Therefore, replacing \textsc{Game 0} with \textsc{Game 1} only negligibly affects the overall success probability of the adversary \A in the game.

It remains to show that the output of \textsc{Game 1} is at most negligibly different from the output of \textsc{Game 2}. 
For fixed strings $r^*$ and $s$, we introduce the relabeled function $f^*$, where
\begin{equation*}
f^*(x) :=
    \begin{cases}
      f(x) \oplus s, &  \text{if } x=r^*, \\
      f(x), & \text{if } x\neq r^*.
    \end{cases}
\end{equation*}
We then build a distinguisher $\algo D$ against $\RelabelingGame^{(1)}$ whose distinguishing advantage is
$$
\Delta^\D(\textsc{Game 1},\textsc{Game 2}) := \left| \underset{f \rand \mathcal R}{\Pr}[ \mathcal{D}^{f}(1^n) = 1] \, - \underset{r^*,s \rand \bit^n}{\Pr}[ \mathcal{D}^{f^*}(1^n) = 1] \right|,
$$
and since $\D$ is a \QPT algorithm, the advantage must then be at most $\negl(n)$ due to \expref{Theorem}{thm:nonadaptive_relabeling}.

$\RelabelingGame^{(1)}$ takes place as follows. First, let the function $f$ be sampled from $\algo R$, as in \textsc{Game 1}. 
The distinguisher \D is then defined as follows.
\begin{enumerate}
\item Run $\algo A_1$, answering its encryption and decryption querying using calls to $\mathcal O_f$.
\item When $\algo A_1$ outputs the challenge plaintexts $(m_0, m_1)$, output the entire (purified) private register of $\algo A$ (i.e., the state denoted $\ket{\psi}$ earlier) together with $m_0$ and $m_1$.
\item Receive the state $\ket{\psi^f}$ defined above. Sample $b \inrand \bit$, and invoke the single (possibly relabeled) oracle call on a random input $r$, receiving a string $g$. Then, set $c := (r, g \oplus m_b)$.
\item Run $\algo A_2$ on input $\ket{\psi}$ and $c$, answering its oracle queries making use of random input/output pairs of $f$ generated by the classical random example oracle.
\item When $\algo A_2$ outputs $b'$, the distinguisher $\D$ outputs $\delta_{b,b'}$.
\end{enumerate}
By construction, the relabeled case corresponds exactly to \textsc{Game 2}, while the unrelabeled case corresponds exactly to \textsc{Game 1}, as desired.
Due to \expref{Theorem}{thm:nonadaptive_relabeling}, the advantage $\Delta^\D(\textsc{Game 1},\textsc{Game 2})$ is at most negligible.
Putting everything together, we see that the output distribution of \textsc{Game 0} (i.e., the true security game) is negligibly far from that of \textsc{Game 2}, where the adversary succeeds with probability $1/2$.
\qed
\end{proof}
\noindent
Finally, as a direct consequence of the equivalance results of \expref{Section}{ch:equiv}, let us
conclude the previous result with the following additional observation:

\begin{corollary}
Let $\mathcal F$ be a function family of quantum-secure pseudorandom functions. Then,
$\Pi[\mathcal{F}]=(\KeyGen,\Enc,\Dec)$ from \expref{Construction}{cons:prf} is \SEMQCCA-secure.
\end{corollary}
\noindent
It is well known, as documented in \cite{KL15}, that the \PRF scheme is easily malleable by an adversary with adaptive decryption oracle access
in a \CCAA learning phase. Therefore, \expref{Construction}{cons:prf} is neither secure under classical nor quantum adaptive chosen-ciphertext attacks.

\begin{construction}[periodized \PRF scheme]\label{cons:p_prf}
For a security parameter $n$, let $\mathcal F$ be a family of keyed functions $f_k: \{0,1\}^n \longrightarrow \{0,1\}^{2n+3}$ 
over $\mathcal{K} = \bit^n$ and let
$\Primes{[}2^n/2,2^n)$ denote the set of primes on the given interval.
Let $\mathcal F'$ be the keyed and periodized function
family over primes $p \in \Primes{[}2^n/2,2^n)$ given by $f'_{k,p}: \bit^{2n+3} \longrightarrow \bit^{2n+3}$, $f'_{k,p}( \cdot) = f_k(\cdot \mod p)$,
and consider the following symmetric-key encryption scheme $\Pi'[\mathcal F']=(\KeyGen',\Enc',\Dec')$:
\begin{enumerate}
\item \emph{(key generation)} let $\KeyGen(1^n)$ generate $k \rand \bit^n$ and $p \rand \Primes{[}2^n/2,2^n)$;\footnote{Note that we can efficiently sample from $\Primes{[}2^n/2,2^n)$ by first generating integers uniformly at random and applying primality testing.
More sophisticated approaches are discussed in \cite{Mau95}.}

\item \emph{(encryption)} on message $m \in \bit^{2n+3}$, choose a randomness $r \rand\{0,1\}^{2n+3}$ and output $\Enc_k(m;r) = (r,f'_{k,p}(r) \oplus m)$;
\item \emph{(decryption)} on cipher $(r,c)$, output $\Dec_k(r,c) = c \oplus f'_{k,p}(r)$;
\item \emph{(correctness)} $(\Dec_k \circ \Enc_k)(m;r) = (f'_{k,p}(r) \oplus m) \oplus f'_{k,p}(r) = m$.
\end{enumerate}
\end{construction}

We now prove that the above construction achieves \INDQCCA security. Note that due to quantum period finding, as observed in \cite{BZ13},
it follows that if $\mathcal F$ is a family of \QPRF{s}, then $\mathcal F'$ from \expref{Construction}{cons:p_prf} is only post-quantum secure.

\begin{proposition}
Let $\Pi[\mathcal F]=(\KeyGen,\Enc,\Dec)$ be a symmetric-key encryption scheme under a \QPRF $\mathcal F$ and
define the periodized \PRF scheme $\Pi'[\mathcal F']=(\KeyGen',\Enc',\Dec')$ as in \expref{Construction}{cons:p_prf}.
Then, $\Pi'[\mathcal F']$ is \INDQCCA-secure even as $\mathcal F'$ is only post-quantum secure.
\end{proposition}

\begin{proof}
We prove the above claim by contradiction.
Suppose there exists a \QPT adversary $\mathcal{A}'$ against \INDQCCA security that wins $\IndGame(\Pi', \algo A', n)$ 
with nonnegligible probability.
We can then use a reduction argument and construct an adversary that runs $\A'$ in order to break the \INDQCCA security of the scheme $\Pi=(\KeyGen,\Enc,\Dec)$
from \expref{Construction}{cons:prf} under a family $\mathcal{F}$ of \QPRF{s} $f_k: \bit^{n} \longrightarrow \bit^{2n+3}$.

Let $\mathcal{A}$ be an adversary against $\IndGame(\Pi, \algo A, n)$ that receives pre-challenge oracle access to $\Enc_k$ and $\Dec_k$,
as well as post-challenge access to $\Enc_k$, whose goal is to output a bit $b'=b$ upon a challenge ciphertext $c^* = (r^*,f_k(r^*)\oplus m_b)$.
The adversary \A can now simulate $\IndGame(\Pi', \algo A', n)$ as follows.
\begin{enumerate}
\item $\A$ samples a prime $p \rand \text{Primes}(2^n/2,2^n)$.
\item whenever $\A'$ issues an encryption query,
\A first encrypts an auxiliary state $\ket{0^{2n+3}}\ket{0^{n}}\ket{0^{2n+3}}$ using its own encryption oracle $\Enc_k$.
As a result, a randomness $r \rand \bit^{n}$ is produced and \A receives a state $\ket{0^{2n+3}}\ket{r}\ket{f_k(r)}$.
Next, \A stretches the randomness in the second register by a random multiple of $p$, hence produces a
string $\tau$ in $\bit^{2n+3}$, such that
$$
\Enc'_k: \,  \ket{m}\ket{r}\ket{c} \mapsto \ket{m}\ket{r}\ket{c \oplus f_k(r) \oplus m} \mapsto \ket{m}\ket{\tau}\ket{c \oplus f_k(\tau \text{ mod }p) \oplus m}.
$$
\item whenever $\A'$ issues a decryption query during the pre-challenge phase, \A computes a
reduced randomness $(\cdot \mod p)$ from the second register into an auxiliary register by means of a reversible operation.
Then, \A evaluates the \QPRF $f_k$ upon the reduced randomness into the message register using its own decryption
oracle $\Dec_k$. A final bit-wise XOR operation results in the desired decryption operation
$$
\Dec'_k: \, \ket{c}\ket{\tau}\ket{h} \longrightarrow \ket{c}\ket{\tau}\ket{h \oplus f_k(\tau \text { mod }p) \oplus c}.
$$
\item at the challenge, \A uses the plaintext pair it receives from $\A'$ for its own challenge phase. 
Next, \A stretches the randomness $r^*$ of the resulting challenge ciphertext $c^* = (r^*,f_k(r^*)\oplus m_b)$
by a random multiple of $p$ and sends it back to $\A'$.
\item when $\algo A'$ outputs $b'$, \A outputs $\delta_{b,b'}$.
\end{enumerate}
By construction, the simulated game above corresponds exactly to $\IndGame(\Pi', \algo A', n)$, as desired.
Therefore, from the assumption that the \QPT adversary $\mathcal{A}'$ wins $\IndGame(\Pi', \algo A', n)$ with nonnegligible probability,
\A must now succeed at $\IndGame(\Pi, \algo A, n)$ with nonnegligible probability - in violation of \expref{Theorem}{thm:prf_scheme}.
\qed
\end{proof}

%% file: ph_realization.tex
\newpage
\section{The Physical Realization of Quantum Computation}\label{ch:physical}

The simulation of quantum systems turns out to scale surprisingly poorly on conventional classical computers. In order to simulate $N$ spin $1/2$ particles
and to solve the Schr\"odinger equation, one needs to store vectors of size $2^N$, as well as manipulate 
matrices of size $2^N \times 2^N$. Due to this exponential scaling, it is well known that classical computers are highly inefficient
in simulating the dynamics of quantum systems. This fact has already puzzled physicists in the 1980s, who hypothesized that an 
intrinsically quantum mechanical computer could potentially be more suitable for these tasks.
In fact, it is often attributed to Richard Feynman \cite{Fey82} to have been the first to speculate on the possibility of building quantum computers.
The prospect of building a computer that would outperform any classical computing architecture and efficiently simulate quantum physics
seemed captivating.
In 1989, David Deutsch gave the first example of a quantum algorithm for a 
black box problem which could be solved faster with quantum mechanical means than with classical ones \cite{Deu89}.
Perhaps most notably, it was Peter Shor's 1994 discovery of efficient quantum
algorithms for the factoring of integers and computing discrete logarithms \cite{Sho94} that truly drew the attention
towards the field of quantum computation.
Only a few years later, Ignacio Cirac and Peter Zoller proposed a physical system of trapped ions
on which quantum information processing could be realized \cite{CZ95}.
In this architecture, single trapped ions are engineered to carry quantum information
and are both manipulated and measured with focused laser beams. Already within
a year's time, David Wineland's group at National Institute of Standards and
Technology achieved a breakthrough in ion-trap quantum computers
\cite{MMKW95}, namely a controlled bit flip on a single ion. This experiment is often considered as the birth of experimental
quantum computation. In this chapter, we give a basic introduction to trapped-ion quantum computers.
To this end, we follow an excellent survey on trapped-ion computation by H{\"a}ffner and Blatt \cite{HRB08}.
In later sections, we also describe recent implementations of quantum algorithms and further advances in the field.

\subsection{DiVincenzo Criteria}

All quantum information processing is concerned with the storage and coherent manipulation of information in a quantum system.
In the previous chapters, we showed how quantum computers could solve certain mathematical problems faster than classical computers 
using the principles of quantum mechanics. 
In order to harness this quantum speed-up, however, one has to realize quantum computation in a physical system.
Fortunately, nature presents us with many possible ways of realizing a qubit in a physical system. As typical representations of a qubit
are found in the two states of a spin $1/2$ particle, the vertical or horizontal polarization of a photon or simply the ground and excited states of an atom,
each representation comes with its own drawbacks and advantages. While photons are easy to generate, they 
have proven to be difficult to interact on the basis of nonlinear materials alone \cite{NC10}. Similarly, both the observation and control of spin states
poses great difficulty, unless a carefully engineered environment is achieved. An example of such circumstances is realized in a trapped-ion quantum
computer, where ions are confined in a potential trap and subsequently cooled.

In 1996, David DiVincenzo, at the IBM Thomas J. Watson Research Center, proposed the following list of guidelines for a successful physical implementation of
a quantum computer: \cite{DiV00}
\begin{enumerate}
\item A scalable physical system with well characterized qubits.
\item The ability to initialize the state of the qubits to a simple initial state.
\item Long relevant decoherence times, much longer than the gate operation time.
\item A universal set of quantum gates.
\item A qubit-specific measurement capability.
\end{enumerate}
Having the possibility of a functioning interface between quantum computers and devices for quantum communication in mind, 
DiVincenzo also added two additional requirements:
\begin{enumerate}
\item[6.] The ability to interconvert stationary and flying qubits.
\item[7.] The ability to faithfully transmit flying qubits between specified locations.
\end{enumerate}
Currently, the trapped-ion computer is oftentimes regarded as the leading quantum computing architecture,
while the runner-up technology is believed to be that of the solid-state architecture of superconducting qubits \cite{LMRD17}. In this thesis,
we present the ion-trap quantum computer as a model for quantum computation and discuss its fundamental properties and capabilities.
When performing quantum information processing, such as the algorithms from the previous chapters,
the ability to coherently manipulate as well as store information with low rates of error is crucial.
Decoherence of quantum systems poses enormous difficulty to both of these tasks.
In the next section, we will outline the extent to which trapped ion computation satisfies these criteria.

\newpage

\subsection{Ion-Trap Implementation}

In this section, we discuss the physical realization of quantum computation
on the basis of the ion-trap, the most successful quantum computing architecture to date.
As the representation of a qubit is found in the hyperfine levels of an ion,
we begin with a section on the hyperfine structure of atoms and continue with
the experimental setup in the subsequent chapter. In the following chapters thereafter, we discuss
how the ion-trap quantum computer satisfies DiVincenzo's Criteria, as well as the extent
to which all the necessary ingredients for the implementation of quantum algorithms of the previous chapters are realized. 
In particular, we show how to perform elementary single-qubit and two-qubit gates using focused laser beams.
Finally, in the last section,
we present a recent performance comparison between a state-of-the-art solid-state device running the same algorithms.

\subsubsection{Hyperfine Structure.}

In order to realize a qubit as a physical carrier of information, one has to represent it in an appropriate two-level quantum system.
Following DiVincenzo's criteria, the task is to define a qubit that is not only well-characterized, but can also be controlled and manipulated
for the purpose of information processing. In the ion-trap quantum computer, a qubit is found in the internal atomic states of
the ion. Although a single trapped ion features a broad energy landscape, sophisticated use of lasers allows us to isolate
just two levels in the energy spectrum of the atom.
Alkali atoms present a popular choice for ion-trap experiments, as they feature a single valence electron in the outer shell,
thus offering a simple and well-studied electronic structure.
Typical ion candidates are the alkaline earth metals $^9$Be$^+$, $^{24}$Mg$^+$, $^{40}$Ca$^+$, as well as $^{171}$Yb$^+$, which, once ionized, behaves
quite similarly. Each ion comes with a different mass and electronic transition at a certain wavelength, both highly relevant factors
for trapping, as well as laser manipulation. For example, while lighter ions carry less inertia and are therefore easier to trap,
they tend to exhibit electronic transitions at wavelengths in the deep ultra-violet that are less suitable for fiber-optics.

At high resolution, the atomic spectrum is known to feature a splitting of energy levels into further substructures, the so-called 
\textit{fine structure} and \textit{hyperfine structure}. Neither of the two structres are explained in the original Bohr model or 
predicted by Schr\"odinger theory
and result from spin contributions of the electron spin and nuclear spin.
The atomic states relevant for the representation of a qubit result from
the sum of electron spin $S$ and nuclear spin $I$, giving a total of $F = S + I$, where $F$ is the total angular momentum.
Using the long-lived states of the hyperfine structure, especially long coherence times can be achieved which, in this regard, make ion traps an
ideal choice for a quantum computer.

Particularly the ytterbium isotope $^{171}$Yb$^+$ has become a favorable choice in recent experiments \cite{DLFL16}\cite{FHM16} 
due to its large hyperfine splitting and strong $^2$S$_{1/2} \leftrightarrow \hspace{0.1mm}^2$P$_{1/2}$ electronic transition around a wavelength of $369.53$nm.
An example of the hyperfine structure of $^{171}$Yb$^+$ is shown in \expref{Figure}{fig:hyperfine}.
This particular isotope exhibits long trapping lifetimes and, due to its strong electronic transition, it is well suited for broadband laser manipulation, 
as well as integration with optical fibers.

\noindent Typically, the qubit is taken to be the two (first-order magnetic field-insensitive) hyperfine levels
of the $^2$S$_{1/2}$ ground state \cite{OYM07}:
\begin{eqnarray*}
\ket{0} & \equiv & ^2\text{S}_{1/2} \ket{ F = 0, m_F = 0 }\\
\ket{1} & \equiv & ^2\text{S}_{1/2} \ket{ F = 1, m_F = 0 }.
\end{eqnarray*}
Here, $F$ and $m_F$ denote the quantum numbers assosciated with the total angular momentum and its projection along the quantization axis defined by an
applied magnetic field of $5.2 \,G$. Since the magnetic quantum number is $m_F=0$, the two hyperfine states carry only a quadratic \textit{Zeeman shift}.
Consequently, the $^{171}$Yb$^+$ ion features a particular insensitivity with respect to magnetic field fluctuations.
Often in the ion-trap literature, the notation $\ket{g} = \ket{0}$ and $\ket{e}= \ket{1}$, denoting the ground and excited states respectively, is adopted 
in order to avoid confusion around additional coupling with vibrational modes of the ion chain.
The qubit frequency splitting between the above $^2$S$_{1/2}$ states is in the order of $\nu_0 = 12.642821$ GHz. Most notably,
Monroe et al. \cite{OYM07} have measured average qubit coherence times of $2.5(3)$s that are significantly longer than the typical gate operation time
at microseconds.

\begin{figure}[tbp]
\centering 
\includegraphics[width=.65\textwidth,origin=c]{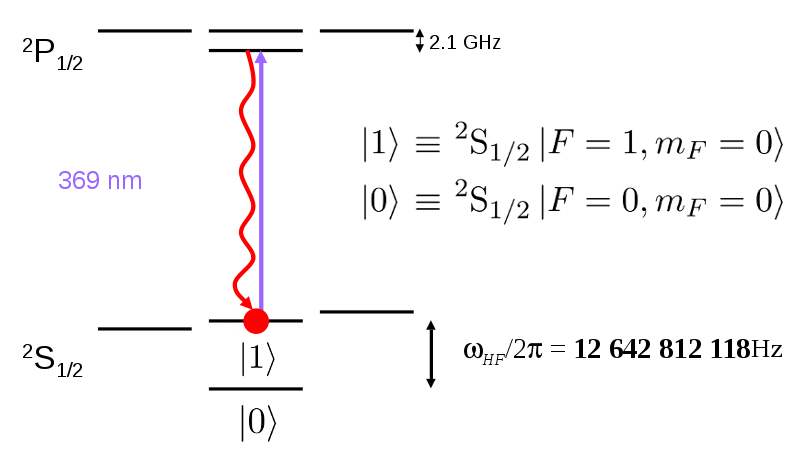}
\caption{\label{fig:i}(\cite{Mon13}) The spin contribution to the atomic energy levels of $^{171}$Yb$^+$. The hyperfine splitting results from the interaction
between the spin-$1/2$ electron and the spin-$1/2$ nucleus.}
\label{fig:hyperfine}
\end{figure}

\subsubsection{Experimental Setup.} \label{experimental_setup}

\begin{figure}[tbp]
\centering 
\includegraphics[width=.65\textwidth,origin=c]{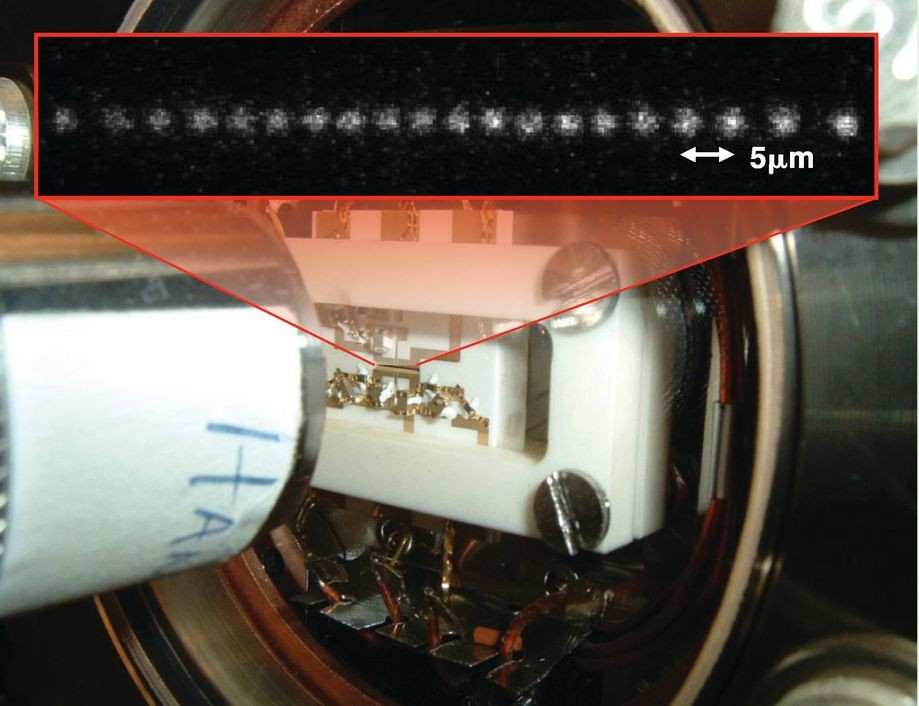}
\caption{\label{fig:i}(\cite{Mon13}) The University of Maryland ion-trap at the Chris Monroe lab, 2013. Each dot represents a single $^{171}$Yb$^+$ ion
exhibiting state dependent fluorescence when driven by individual focused laser beams. A CCD camera collects fluorescence 
from the scattering of photons and creates an image over thousands of measurements. }
\end{figure}

The main component of an ion-trap quantum computer is an electromagnetic trap, a vacuum chamber surrounded by four cylindrical electrodes. 
In order to produce the necessary trapping potential towards axial confinement of the ions, the end caps of the rods are biased
at different voltages.
However, as we will show now, trapping ions by means of static electric fields alone is not possible.
\textit{Earnshaw's theorem} states that a charged particle cannot be confined in
three dimensions by static electric fields, as the divergence of the field vanishes in empty space. This fact
can easily be verified in the following short argument. Let $\bold r_0=(x_0,y_0,z_0)$ denote the coordinates of
the charge. Then, for any trapping potential, we require that the particle returns to its equilibrium position once it is displaced.
Consequently, we demand from the potential energy $U(\bold r)$ that:
\begin{equation}
\label{curved_potential}
\nabla U(\bold r_0) =0, \hspace{4mm} \nabla^2 U(\bold r_0) >0.
\end{equation}
However, since the electric potential energy is given by $U(\bold r) = q\,\Phi(\bold r)$, where $\Phi$ is the electrostatic potential,
we must conclude from \textit{Gauss' law} that:
\begin{equation}
\nabla^2 U(\bold r) \, = \, q\,\nabla^2 \Phi(\bold r) \, = \, 0.
\end{equation}
Thus, the electric potential obeys the \textit{Laplace equation} and violates the previous condition Eq.\eqref{curved_potential} we required for the
desired potential trap.
In order to achieve confinement, one has to adopt time-varying electric fields that, on average, create an effective trap in three dimensions.
In practice, this is can be realized in a \textit{rf Paul trap} \cite{POF58}, where a combination of both static as well 
as oscillating electric fields switching at rates around radio frequency produce an effective harmonic trap. 
As a result, a potential of \textit{quadropole geometry} is generated that confines a charge in all three dimensions.
A linear rf Paul trap can also succesfully confine several ions simultaneously along its trap axis (typically taken to be the $\hat z$-axis)\cite{Pau90}.
Since the vibrations of trapped ions around their equilibrium positions are strongly coupled
due to the Coulomb interaction, motion of any one single ion induces a joint oscillation in all other ions.
The Hamiltonian describing the motion of $N$ confined ions together with the Coulomb repulsion is given by: 
\begin{equation}
\mathcal{H} = \sum_{j=1}^N {\frac{M}{2} \left(  \frac{\hat{p}_j^2}{M^2} +  \omega_x^2 x_j^2 + \omega_y^2 y_j^2 + \omega_z^2 z_j^2 \right)}
    + \sum_{j=1}^N \sum_{i>j} \frac{e^2}{4\pi \epsilon_0 |\hat r_j - \hat r_i|},
\end{equation}
where $M$ is the mass of a single ion and $\omega_x,\omega_y$ and $\omega_z$ describe the frequency of oscillation along the respective directions.
For the sake of simplicity,
one typically considers a linear Paul trap design in which only a single
motional direction along the trap axis is selected for in which all ions lie along the $\hat z$-axis.
If the displacement due to the oscillations is much smaller than the spatial separation between the ions, we can describe the
vibrations (i.e. phonons) in an harmonic oscillator approximation \cite{WMI98}. 
At low temperature, the linear chain of ions freezes into a crystal where, for a quantum of vibrational energy $\hbar \omega_z$,
the desired cooling requires both that $k_B T \ll \hbar \omega_z$, as well as that the thermal energy $T$ drops below the 
energy difference of the two atomic levels.
A chain of $N$ ions exhibits various normal modes of vibration, both radial and axial, each at frequencies independent of $N$. The axial
mode of lowest frequency is given by the \textit{center-of-mass mode} (COM), a collective motion of the entire ion chain along the trap axis.
In order to control and encourage such joint motion in the COM mode, 
it is necessary to surpress vibrational modes of higher frequency (such as relative or radial motion) by applying
Doppler-cooling and preparing the ions in their motional ground state \cite{WI79}.

The use of resonant laser light is a fundamental component of the ion-trap computer and appears throughout multiple stages
of quantum information processing, such as groundstate-cooling, qubit initialization, qubit gate-operations and state detection.
In order to achieve state initialization, sophisticated use of \textit{optical pumping} can drive hyperfine transitions 
into short-lived and energetically distant states that subsequently decay back to the ground state
according to known \textit{selection rules} (\expref{Figure}{pumping}). Measurement, or state detection, works using state dependent fluorescence as follows: 
If the qubit state is in the excited state $\ket{1}$, the 369.53 nm light applied for
detection is nearly on resonance, and the ion exhibits fluorescence by scattering many
photons. If, however, the state is in the ground state $\ket{0}$, very few photons are scattered and we observe a dark state. Finally, a photon count
results in accurate state detection.
Moreover, as we discuss in the subsequent chapters, manipulation of qubits can be realized as an optical Rabi oscillation under resonant laser light.

\begin{figure}[tbp]
\centering 
\includegraphics[width=.45\textwidth,origin=c]{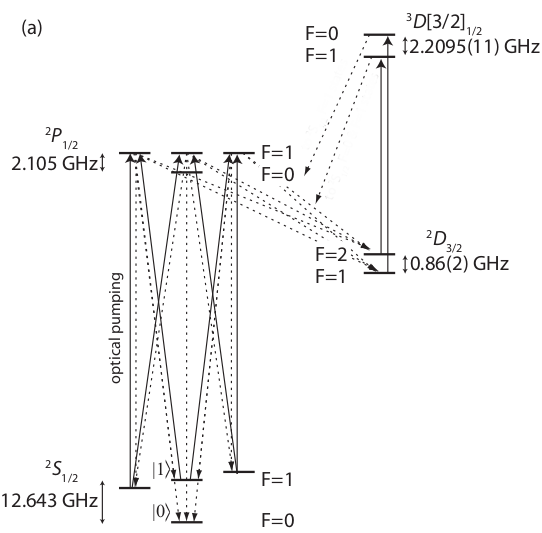}
\hfill
\includegraphics[width=.45\textwidth,origin=c]{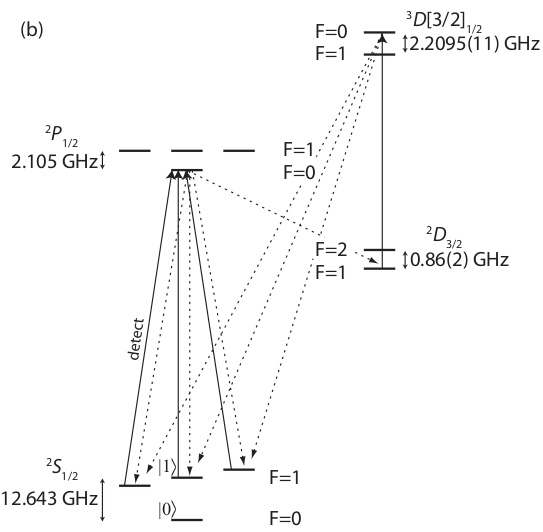}
\caption{\label{fig:I}(\cite{OYM07}) Optical pumping for state initialization (a) and state detection (b) of $^{171}$Yb$^+$.
The nuclear spin is given by $I = 1/2$. Appropriate polarization of the incoming laser beam can exploit atomic selection rules 
and initialize the desired groundstate.}
\label{pumping}
\end{figure}
\ \\
Let us conclude this section by briefly summarizing how ion-trap quantum computers fulfill the
DiVincenzo criteria:
\begin{enumerate}
\item \textbf{A scalable physical system with well characterized qubits:} The atomic hyperfine states are exceptionally long lived and 
serve as an ideal choice for qubits, see \cite{OYM07}.
Though scalable in principle, complications typically arise in both mass and mode structure of sufficiently long ion chains. 
Modern techniques circumvent these problems and address scalability by means of ion-transport among
multiple ion traps \cite{WMI98}\cite{FHM16}.

\item \textbf{The ability to initialize the state of the qubits to a simple initial state:}
State initialization is achieved using optical pumping, a technique that prepares hyperfine ground states with average fidelities $>0.99$, 
for example \cite{OYM07}.

\item \textbf{Long relevant decoherence times, much longer than the gate operation time:}
Typical coherence times of modern ion-trap architectures average around a few
milliseconds and are therefore several orders of magnitude longer than the time scale required
for quantum gate operations, see \cite{OYM07}.

\item \textbf{A universal set of quantum gates:}
All single-qubit gates can be performed using laser pulses that drive Rabi oscillations between the two atomic levels.
Two-qubit gates are implemented by exploiting the long range Coulomb interaction, such as in the original
proposal by Cirac and Zoller \cite{CZ95}.
Quantum information from a single ion can be transferred into the common
motional degree of freedom of the entire ion string using a sideband transition by focused laser pulses.
Such conditional quantum dynamics are sufficient to give rise to elementary two-qubit gates needed for universal computation.
Moreover, sources of error during larger scale quantum operations can be controlled for by more sophisticated types of multiparticle entanglement, such as
M{\o}lmer-S{\o}rensen interactions \cite{MS99}.

\item \textbf{A qubit-specific measurement capability:}
Measurements are performed using state dependent fluorescence in which photon scattering
allows for state read-out of individual qubits.

\end{enumerate}
\subsubsection{The Hamiltonian.}

In this section, we discuss the basic Hamiltonian of a single trapped ion interacting with near resonant laser light.
The two-level approximation is valid in this regime, as all other atomic levels are energetically far away and highly detuned.
Similar to a spin-$1/2$ system under a time-dependent magnetic field, the two-level atom undergoes an optical Rabi oscillation 
under the action of the electromagnetic field.

Let us consider a Hamiltonian of a two-level system interacting with a quantized harmonic oscillator of vibrational modes through a laser beam,
where:
\begin{equation}
\mathcal{H} = \mathcal{H}_{\text{atom}} + \mathcal{H}_{\text{free}} + \mathcal{H}_{\text{int}}.
\end{equation}
Recall from the previous section that the Hamiltonian describing the free motion of a single ion along the trap axis
in an effective harmonic potential can be written as:
\begin{equation}
\mathcal{H}_{\text{free}} = \frac{p^2}{2M} + \frac{1}{2}M\omega_z^2 z^2,
\end{equation}
where $\omega_z$ is the frequency of oscillation around the equilibrium position in the $\hat z$ direction. If the coupling to the external field
is small and the ion inside the vaccum chamber is well isolated from its surroundings, its motion becomes quantized and we can introduce raising and lowering operators
$z = \sqrt{\frac{\hbar}{2M\omega_z}}(a + a^{\dagger})$ and $p = i\sqrt{\frac{\hbar M \omega_z}{2}}(a - a^{\dagger})$.
Consequently, together with the Hamiltonian corresponding to the internal atomic levels, we can write:
\begin{eqnarray}
\mathcal{H}_{\text{atom}} & = & \hbar \omega_{eg} \frac{\sigma_z}{2}\\
\mathcal{H}_{\text{free}}\,  & = & \hbar \omega_z \left(a^{\dagger} a + \frac{1}{2}\right).
\end{eqnarray}
In the following, we will denote $\mathcal{H}_0 = \mathcal{H}_{\text{atom}} + \mathcal{H}_{\text{free}}$. The Hamiltonian $\mathcal{H}_{\text{int}}$ 
describes the atom-light interaction of the  ion with the laser.
Following Wineland et al.~\cite{WBB03}, the interaction between the ion and the electric field of the laser beam is given by:
\begin{equation}
\label{dipole_Hamiltonian}
\mathcal{H}_{\text{int}}(t) = - \vec{ d} \cdot \vec{ E} = - \vec{ d} \cdot E_0 \, \hat \epsilon_L \cos(  k z - \omega_L t + \phi),
\end{equation}
where $\vec{d}$ is the electric dipole operator, $\vec{E}$ is the (classical) electric field, $E_0$ the field strength, $z$ is the position operator of 
the ion for displacement
from its equilibrium position, $\hat \epsilon_L$ is the laser beam polarization, $\omega_L$ is the frequency of the laser,
$k$ is the laser beam's $k$-vector parallel to $\hat z$ (the axis of the trap) and where $\phi$ is the phase of the laser at the mean position of the ion.
In the dipole approximation, $\vec{d}$ can be further expanded in terms of the internal states of the atom, since it is proportional to $\sigma_+ + \sigma_-$, 
where $\sigma_+ = \ket{e}\bra{g}$ and $\sigma_- = \ket{g}\bra{e}$.
By introducing the Rabi flop frequency $\Omega = - \frac{E_0}{ \hbar}\displaystyle\braket{e| \vec{d} \cdot \hat \epsilon_L |g}$ and the Lamb-Dicke parameter
$\eta = k \, \sqrt{\frac{\hbar}{2M\omega_z}}$, we can express Eq.\eqref{dipole_Hamiltonian} as:
\begin{equation}
\label{H_int}
\mathcal{H}_{\text{int}}(t) = \hbar \frac{ \Omega}{2} \, (\sigma_+ + \sigma_{-}) \left( e^{i( \eta (a + a^\dagger) - \omega_Lt + \phi)} + e^{-i( \eta (a + a^\dagger) - \omega_Lt + \phi)} \right).
\end{equation}
Taking the width of the ion's oscillation along the trap axis at low temperatures to be small compared to the wavelength of the incoming laser beam, we can apply the
Lamb-Dicke limit ($ \eta \sqrt{ \braket{(a + a^{\dagger})^2}} \ll 1$) and further expand the relevant exponential from Eq.(\ref{H_int}):
\begin{equation}
e^{i \eta (a + a^\dagger)} \, = \,  1 + i \eta \, ( a  + a^{\dagger} )  + \, \mathcal{O}(\eta^2).
\end{equation}
It is now convenient to work in the interaction picture $\mathcal{H}'_{\text{int}} = e^{i\mathcal{H}_0 t / \hbar} \mathcal{H}_{\text{int}} e^{-i\mathcal{H}_0 t / \hbar}$. Using the
\textit{Baker-Campbell-Hausdorff lemma},
\begin{equation}
e^{\alpha A} B e^{-\alpha A} = B + \alpha \left[A,B\right] + \frac{\alpha^2}{2!} \left[A,\left[A,B\right]\right] +\frac{\alpha^3}{3!} \left[A,\left[A,\left[A,B\right]\right]\right] + \,...,
\end{equation}
we get the following identities:
\begin{align}
e^{i \omega_z a^\dagger a \,t} \left[1 + i \eta \, ( a  + a^{\dagger} )\right] e^{-i\omega_z a^\dagger a \,t} &= 1 + i \eta \, ( a e^{-i \omega_z t} + a^{\dagger} e^{i \omega_z t})  \\
e^{i \omega_{eg} \sigma_z t /2} \sigma_+ e^{-i\omega_{eg} \sigma_z t /2} &= \sigma_+ e^{i\omega_{eg} t}.
\end{align}
By applying a rotating wave approximation and assuming near resonance $\Delta = \omega_L - \omega_{eg} \approx 0$, we ignore all rapidly oscillating terms of the form
$\exp(\pm i(\omega_{L} + \omega_{eg})t)$ and find:
\begin{equation}
\mathcal{H}'_{\text{int}}(t) = \, \hbar \frac{ \Omega}{2} \, \sigma_+ \, e^{-i(\Delta t - \phi)} \left[ 1 + i \eta \, ( a e^{-i \omega_z t} + a^{\dagger} e^{i \omega_z t})  \right] + \text{ h.c.} + \, \mathcal{O}(\eta^2).
\end{equation}

\begin{figure}[tbp]
\centering 
\includegraphics[width=.55\textwidth,origin=c]{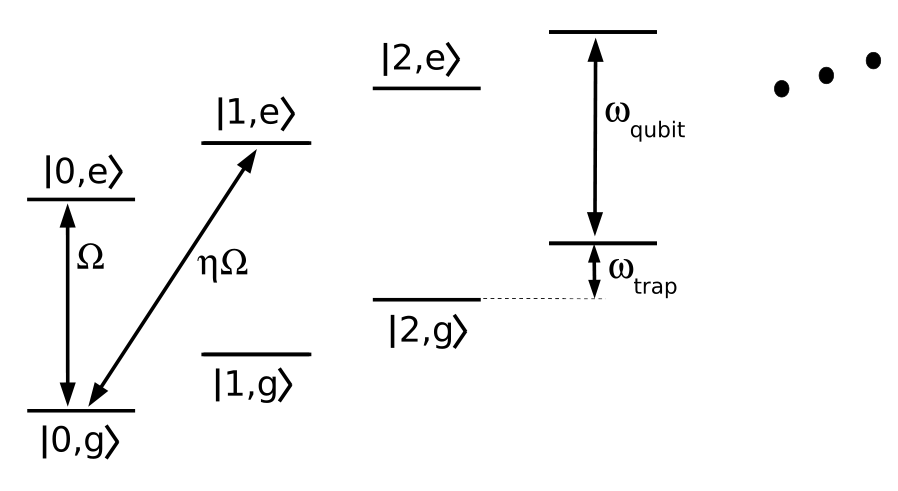}
\caption{\label{fig:i}(\cite{HRB08}) Transitions between atomic levels and phonon modes.}
\end{figure}

\noindent
A second rotating wave approximation assumes that only one transition at a time is considered, now gives the Hamiltonian:
\begin{equation}
\mathcal{H}'_{\text{int}}(t) = \frac{\hbar \Omega}{2} \left[ \sigma_+ e^{-i(\Delta t-\phi)} +  \sigma_- e^{i(\Delta t-\phi)} + i\eta ( \sigma_+ e^{-i(\Delta t-\phi)} -  \sigma_- e^{-i(\Delta t-\phi)} ) (a e^{- i \omega_z t} + a^{\dagger} e^{i \omega_z t}) \right]
\end{equation}
We can now identify three different cases of interest with respect to the detuning of the laser beam, the so-called carrier and sideband transitions:\cite{HRB08}
\begin{itemize}
\item[1.] The carrier transition ($\omega_L = \omega_{eg}, \, \Delta = 0$):
\begin{equation}
\mathcal{H}_{C} = \hbar \frac{\Omega}{2} \,(\sigma_+ e^{i \phi} + \sigma_- e^{-i \phi}).
\end{equation}
In this regime, transitions $\ket{g}\ket{n} \leftrightarrow \ket{e}\ket{n}$ between the atomic states of the ion can be performed.
\item[2.] The blue sideband transition ($\omega_L = \omega_{eg} + \omega_z, \, \Delta = \omega_z$):
\begin{equation}
\mathcal{H}_{+} = i \hbar \frac{\Omega}{2} \eta \,(\sigma_+ a^{\dagger} e^{i \phi} - \sigma_- a e^{-i \phi}).
\end{equation}
This allows for the creation of a phonon mode and simultaneous excitation of the atomic state: 
$\ket{g}\ket{n} \leftrightarrow \ket{e}\ket{n+1}$.
\item[3.] The red sideband transition ($\omega_L = \omega_{eg} - \omega_z, \, \Delta = -\omega_z$):
\begin{equation}
\mathcal{H}_{-}= i \hbar \frac{\Omega}{2} \eta \,( \sigma_+ a e^{i \phi} + \sigma_- a^{\dagger} e^{-i \phi}).
\end{equation}
Simultaneous to exciting the atomic state of the ion, a phonon mode is destroyed. Thus the following transitions can be performed:
$\ket{g}\ket{n} \leftrightarrow \ket{e}\ket{n-1}$.
\end{itemize}
Note that the red sideband Hamiltonian is formally equivalent to the well-known \textit{Jaynes-Cummings} Hamiltonian in quantum optics
that describes a two-level atom interacting with a quantized mode of an optical cavity.
Since the Coulomb interaction provides a strong coupling among the ions, the entire chain of ions exhibits various normal modes of
motion, each at different frequencies, such as center-of-mass mode, the stretch mode or the axial mode \cite{HRB08}.

In order to describe the full Hamiltonian of a linear crystal consisting of $N$ ions, we can introduce a sum over
all single ion contributions and respective vibrational modes of the entire ion chain, as follows:
\begin{eqnarray}
\mathcal{H}_{0} & = & \sum_{j=1}^N \hbar \omega_{eg} \frac{\sigma_{zj}}{2} \, + \, \sum_{l=1}^N \hbar \omega_{z_l} \left(a_l^{\dagger} a_l + \frac{1}{2}\right)\\
\mathcal{H}'_{\text{int}} &=& \sum_{j=1}^N \frac{\hbar \Omega_j}{2} \, \sigma_{+j} \, e^{-i(\Delta t - \phi)} \exp{\left(i\sum_{l=1}^N \eta_{jl} \, [ a_l e^{-i \omega_z t} + a_l^{\dagger} e^{i \omega_z t}]  \right)} + \text{ h.c.}.
\end{eqnarray}
Repeating the analysis of the single-ion Hamiltonian, we can write the interaction Hamiltonian in 
the rotating wave approximation and the Lamb-Dicke limit as:
\begin{equation}
\mathcal{H}'_{\text{int}} = \sum_{j,l=1}^N \frac{\hbar \Omega_j}{2} \left[ \sigma_{+}^{(j)} e^{-i(\Delta t-\phi)} +  \sigma_{-}^{(j)} e^{i(\Delta t-\phi)} + i\eta_{jl} ( \sigma_{+}^{(j)} e^{-i(\Delta t-\phi)} -  \sigma_{-}^{(j)} e^{i(\Delta t-\phi)} ) (a_l e^{- i \omega_z t} + a_l^{\dagger} e^{i \omega_z t}) \right]
\end{equation}
Here, we applied ground-state cooling and prepared only the lowest frequency COM mode in which the same phonon is shared among all ions in the crystal.
In this regime, we can implement a two-qubit gate using the common motional degree of freedom as a \textit{bus} to transfer conditional information among the ions.
In the next sections, we describe how to realize single-qubit gates, as well as two-qubit gates, in the ion-trap quantum computer.

\subsubsection{Single-Qubit Gates.}\label{single_qubit_chapter}

In Chapter \ref{ch:universal}, we discussed how all quantum operations can be 
broken down into a sequence of single qubit and two-qubit operations.
A major advantage of trapped-ion quantum computers lies in the fact that single-qubit operations are particularly easy to implement, as well as to
control through the use of resonant laser light.
In fact, we can show that any single-qubit operation corresponds to a rotation on the Bloch sphere and can thus be realized 
as a Rabi oscillation between the two qubit levels using a resonant laser pulse. In practice,
such tuning of appropriate pulse parameters takes place at the control interface given by an acousto-optical modulator (AOM) \cite{DHL05}.

Consider a two-level system that starts out in some internal state $\ket{\psi} = c_g \ket{g} + c_e \ket{e}$ at time $t=0$. 
Under a stationary Hamiltonian, the subsequent time-evolution after time $\tau$ is governed by the unitary dynamics:
\begin{equation}
\label{rabi_evolution}
\ket{\psi(\tau)} = \exp{\left(\frac{-i \mathcal{H} \tau}{\hbar}\right)} \ket{\psi(0)}.
\end{equation}
Considering the stationary Hamiltonians $\mathcal{H}_C, \mathcal{H}_+$ and $\mathcal{H}_{-}$ from the previous section 
under radiation of pulse length $\tau$ and respective detuning, 
we arrive at unitary dynamics which induce the following rotations:
\begin{align}
R_C\left(\theta,\phi\right) &= \exp{\left(-i \theta /2 \,  (\sigma_+ e^{i \phi} + \sigma_- e^{-i \phi})\right)} \\
R_+\left(\theta,\phi\right) &= \exp{\left(-i \theta /2 \,  (\sigma_+ a^{\dagger} e^{i \phi} - \sigma_- a e^{-i \phi})\right)} \\
R_-\left(\theta,\phi\right) &= \exp{\left(-i \theta /2 \,  (\sigma_+ a e^{i \phi} + \sigma_- a^{\dagger} e^{-i \phi})\right)},
\end{align}
where the control parameters $\theta = \Omega \tau$ (or $\theta = \Omega \eta \tau$ for the sideband evolution) and phase $\phi$ determine the nature
of the rotation. Note that the phase parameter $\phi$ of the laser at the start of the interaction experiment is completely arbitrary
but sets the reference for all subsequent operations.
We can identify the result of any of the above dynamics by a rotation operator $R(\theta,\phi)$ acting on $\ket{\psi}$
in terms of a rotation in the equatorial plane by $\phi$ and a rotation $\theta$ in the vertical plane. For example,
in the case of the carrier evolution, this allows us to
decompose the evolution as:
\begin{eqnarray}
\label{carrier}
R_C\left(\theta,\phi\right) &=& \exp{\left(-i \theta /2 \,  (\sigma_+ e^{i \phi} + \sigma_- e^{-i \phi})\right)} \\
               &=& \mathbbm{1} \cos \theta/2 - i (\sigma_x \cos \phi - \sigma_y \sin \phi) \sin \theta/2 \\
               &=&
\begin{pmatrix} 
  \cos\theta/2            & \, -i e^{i\phi} \sin \theta/2\\ 
  -i e^{-i\phi} \sin\theta/2  & \cos \theta/2
\end{pmatrix}
\end{eqnarray}
Thus, by fixing $\phi$ appropriately, we can now identify the following set of rotation operators in the $x$ and $y$ plane:
\begin{eqnarray}
R_x(\theta) &=&
\begin{pmatrix} 
  \cos\theta/2     & \, -i \sin\theta/2\\ 
  -i \sin\theta/2  & \cos\theta/2
\end{pmatrix} \\
R_y(\theta) &=&
\begin{pmatrix} 
  \cos\theta/2  & \, -\sin\theta/2\\ 
  \sin\theta/2  & \cos\theta/2
\end{pmatrix}
\end{eqnarray}
In order to obtain $R_z(\theta)$, we can use a natural decomposition into rotations around the 
$x$ and $y$ axis by writing $R_z(\theta) = R_y(\frac{\pi}{2}) R_x(\theta) R_y(-\frac{\pi}{2})$. Thus, rotations
around the $z$ axis are given by the rotation operator:
\begin{equation}
R_z(\theta) =
\begin{pmatrix} 
  e^{-i \theta /2}  & \, 0\\ 
  0  & \, e^{i \theta /2}
\end{pmatrix} \hspace{10mm}
\end{equation}
In fact, \textit{any} unitary single-qubit operation $U$ can be decomposed using the rotation operators above, as stated in \expref{Theorem}{th:single}.
Consider, for example, a resonant pulse of length $ \Omega \tau = \pi$ which realizes a $180^\circ$ rotation (up
to an overall phase):
\begin{equation}
R_x(\pi) \ket{\psi}
=  -i\sigma_x \ket{\psi}.
\end{equation}
Another important gate is the Hadamard gate, which we can now realize as a $\frac{\pi}{2}$-pulse in the $y$-plane and write
$H = R_y(\pi/2)$.
For example, given the initial state $\ket{g}$, we can easily create an equal superposition by performing a Hadamard $\frac{\pi}{2}$-pulse:
\begin{equation}
H \ket{g} = R_y(\pi/2) \ket{g}
= \frac{\ket{g} + \ket{e}}{\sqrt{2}}.
\end{equation}
Starting from an initial state $\ket{g}$, we can also prepare any pure state $\ket{\psi}$ on the Bloch sphere (\expref{Figure}{fig:bloch}) by an
appropriate choice of control parameters $\theta = \Omega \tau$ and $\phi$ using the unitary dynamics in Eq.\eqref{carrier}:
\begin{equation}
\ket{\psi} = \cos\left(\frac{\theta}{2}\right) \ket{g} + e^{i \phi} \sin\left(\frac{\theta}{2}\right) \ket{e}.
\end{equation}

\noindent If the laser is slightly detuned, we can repeat the analysis above for sideband rotations that at the same
time increase the vibrational modes. Notice that now the control parameters are $\theta = \Omega \eta \tau$ and $\phi$.

\subsubsection{Two-Qubit Gates.}

According to early work by David Deutsch \cite{Deu89}, 
a universal set of gates can be achieved using single-qubit and two-qubit gates only. In the previous section,
we introduced the means to generate single-qubit operations under laser radiation and subsequent Rabi oscillation.
In order to describe a system of a linear chain of ions, each mutually coupled with the Coloumb interaction,
one has to adopt a Hamiltonian that includes total contribution of all ions.

An early proposal for a two-qubit gate can already be found in the Cirac and Zoller \cite{CZ95} 
design of the ion-trap computer.
The idea of the Cirac-Zoller-gate is the following. A red sideband pulse onto the first ion transfers information from the atomic state 
into the motional degree of freedom, conditioned on its state.
Once the ion begins oscillating, it affects the entire string of ions due to the strong Coulomb repulsion.
Thus, the second ion can now be addressed with operations that are conditioned on the motional state of the first ion.
Finally, another red sideband transition reverses the motional state and causes the first ion to return to its original state. The procedure works as follows:
\begin{figure}[ht!]
\centering 
\includegraphics[width=.95\textwidth,origin=c]{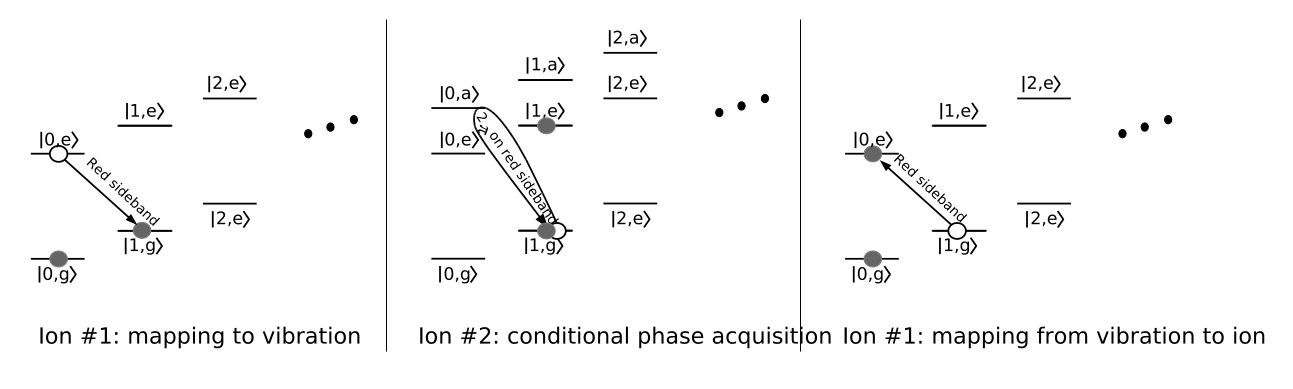}
\caption{\label{fig:i}(\cite{CZ01}) The two-qubit Cirac-Zoller gate.}
\end{figure}
\begin{enumerate}
\item A laser beam tuned to the red sideband frequency $\omega_{eg} - \omega_z$ and length $\theta = \pi$ is focused on the first ion.
Depending on the atomic state of the ion, a transfer into the motional degree of freedom may occur. Consequently, if the ion starts out in the ground state,
no state transfer occurs due to the detuning. If the ion is in an excited state, a phonon mode is created:
\begin{eqnarray*}
\ket{g}\ket{0} &\longrightarrow& \hphantom{-i}\ket{g} \ket{0}\\
\ket{e}\ket{0} &\longrightarrow& -i\ket{g} \ket{1}.
\end{eqnarray*}
\item A second laser tuned to the red sideband frequency with duration $\theta = 2\pi$ is now focused onto the second ion.
This induces a $2\pi$ rotation between $\ket{g}\ket{1}$ and an auxiliary hyperfine state $\ket{a}\ket{0}$. Note that the design of the transition is such
that all other states $\ket{g}\ket{0}$, $\ket{e}\ket{0}$, and $\ket{e}\ket{1}$ are left untouched, as there is insufficient energy
to drive any of these levels. As a result, the following qubit operation is performed at the second ion:
\begin{eqnarray*}
\ket{e}\ket{0} \hphantom{-}& \longrightarrow &  \hphantom{-}\,\ket{e}\ket{0}\\
\ket{e}\ket{1} \hphantom{-}& \longrightarrow &  \hphantom{-}\,\ket{e}\ket{1}\\
\ket{g}\ket{0} \hphantom{-}& \longrightarrow &  \hphantom{-}\,\ket{g}\ket{0}\\
\ket{g}\ket{1} \hphantom{-}& \longrightarrow &            \,- \ket{g}\ket{1}
\end{eqnarray*}
\item A final laser beam tuned to the red sideband frequency $\omega_{eg} - \omega_z$ and length $\theta = \pi$ is focused on the first ion to remove
the motional quantum and restore the first ion to its original state.
\begin{eqnarray*}
\ket{g}\ket{0} &\longrightarrow& \hphantom{-i}\ket{g} \ket{0}\\
\ket{g}\ket{1} &\longrightarrow& -i\ket{e} \ket{0}.
\end{eqnarray*}
\end{enumerate}
In summary, the Cirac-Zoller gate performs the following two-qubit operation:
\begin{align*}
                      &&   R_-^{(1)}(\pi,0)&&                         && \,\,  R_-^{(2)}(2\pi,0)&&                         && \,\,  R_-^{(1)}(\pi,0)&&                       \\
\ket{g}\ket{g}\ket{0} && \longrightarrow &&   \ket{g}\ket{g}\ket{0} && \longrightarrow &&   \ket{g}\ket{g}\ket{0} && \longrightarrow &&  \ket{g}\ket{g}\ket{0}\\
\ket{g}\ket{e}\ket{0} && \longrightarrow &&   \ket{g}\ket{e}\ket{0} && \longrightarrow &&   \ket{g}\ket{e}\ket{0} && \longrightarrow &&  \ket{g}\ket{e}\ket{0}\\
\ket{e}\ket{g}\ket{0} && \longrightarrow && -i\ket{g}\ket{g}\ket{1} && \longrightarrow &&  i\ket{g}\ket{g}\ket{1} && \longrightarrow &&  \ket{e}\ket{g}\ket{0}\\
\ket{e}\ket{e}\ket{0} && \longrightarrow && -i\ket{g}\ket{e}\ket{1} && \longrightarrow && -i\ket{g}\ket{e}\ket{1} && \longrightarrow &&  -\ket{e}\ket{e}\ket{0}
\end{align*}
The operation, as shown above, realizes the controlled-Z (CZ) gate. In the 2-qubit representation, it can be written as a unitary matrix:
\begin{equation}
U = 
\begin{pmatrix} 
  1   &   \hphantom{-}0   &   \hphantom{-}0   &  \hphantom{-}0\\
  0   &   \hphantom{-}1   &   \hphantom{-}0   &  \hphantom{-}0\\
  0   &   \hphantom{-}0   &   \hphantom{-}1   &  \hphantom{-}0\\
  0   &   \hphantom{-}0   &   \hphantom{-}0   &   -1
 
\end{pmatrix}
\end{equation}
By using a Ramsey-type experiment with two additional single-qubit $\pi/2$ pulses, the CZ-gate is easily turned into a CNOT gate, as follows:

\begin{figure}[ht!]
\centering
\includegraphics[width=70mm]{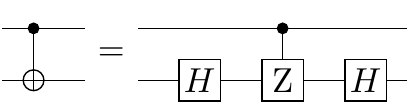}
\caption{\label{fig:i} Using a CZ gate and two $\pi/2$-pulses that perform Hadamard gates, one can construct a CNOT gate.}
\label{fig:cnot}
\end{figure}

In the next section, we discuss recent implementations of many of the algorithms we present in this thesis.

\subsubsection{Quantum Algorithms with Trapped Ions.}

Ever since the first quantum algorithms emerged after Deutsch's algorithm was first proposed,
many algorithms have in fact been implemented on a quantum computer. In this respect, the ion-trap quantum computer still largely dominates
all other architectures, particularly due to its long coherence times. Using $^{40}$Ca$^+$ ions, Deutsch's algorithm was first successfully implemented
as early as 2003 \cite{GHRL03}. Typical fidelities on identifying the function classes already
exceeded over $0.9$.\par
The most comprehensive report on the implementation of standard quantum algorithms up to date was recently published by Monroe at al.
at the University of Maryland \cite{DLFL16}. 
Using a linear chain of five $^{171}$Yb$^+$ hyperfine qubits, a programmable interface
allows the implementation of the Deutsch-Josza and Bernstein-Vazirani algorithm, Simon's algorithm \cite{Sim97} and the
quantum Fourier transform.
Compared to other architectures, such as the solid-state implementations, 
the ion-trap quantum computer is more flexible since it is easy to 
program by external fields and can thus be reconfigured to run any of the standard algorithms discussed in this thesis.

The setup of the programmable ion-trap computer is as follows.
At the top of the hierarchy, we find a flexible interface that allows a user to program the specifications of the desired algorithm. 
Here, a standard set of universal gates
such as the Hadamard, the CNOT or the CP gate are available for programming . In analogy to a classical compiler, a quantum compiler translates these
gates to a set of native gate instructions consisting of single-qubit rotation pulses or two-qubit Ising-like gates due
to M{\o}lmer-S{\o}rensen interactions \cite{MS99}. All native gates are finally performed
as external light pulses originating from the acousto-optical modulator (AOM).
\begin{figure}[tbp]
\centering 
\includegraphics[width=.50\textwidth,origin=c]{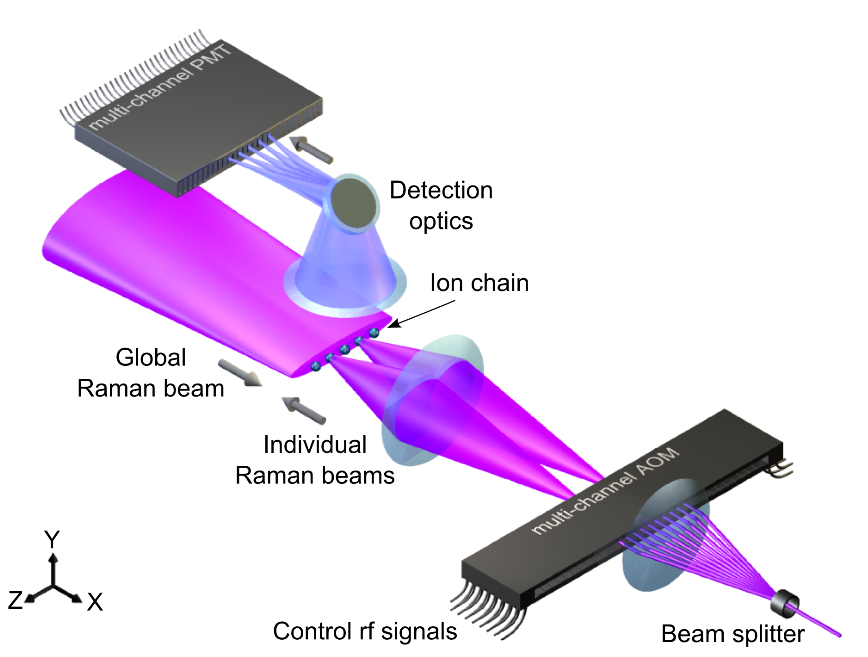}
\hspace{2mm}
\includegraphics[width=.47\textwidth,origin=c]{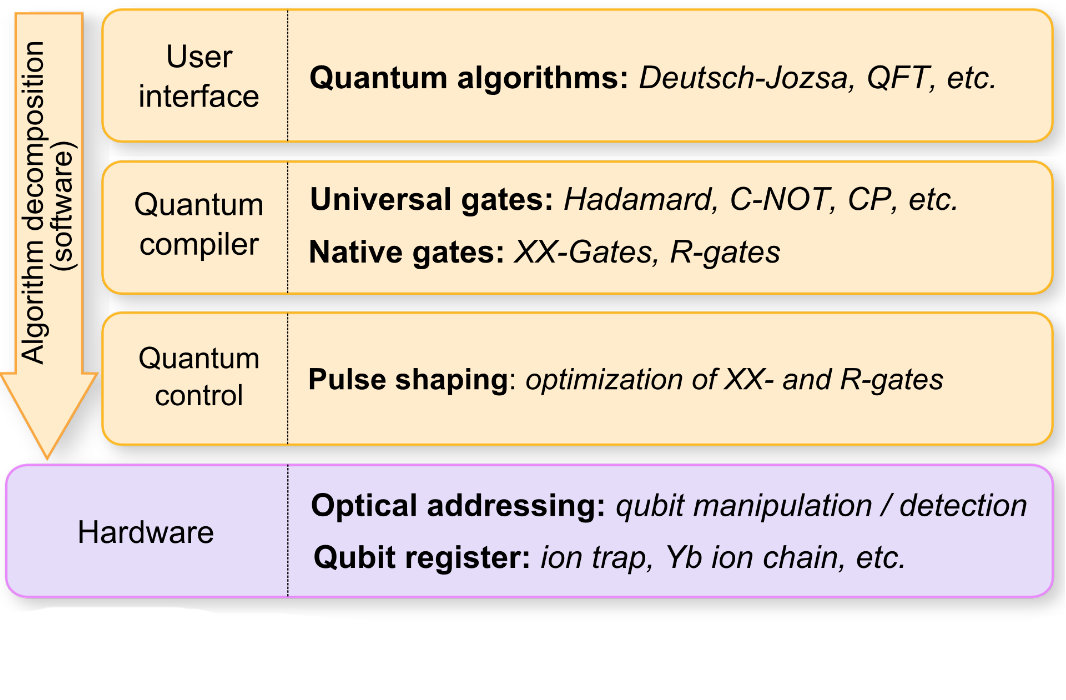}
\caption{\label{fig:i}(\cite{DLFL16}) The Maryland ion-trap setup. A user-interface is provided that allows versatile programming of a five-qubit ion-trap computer
to run standard algorithms. }
\end{figure}
\begin{figure}[tpb]
\centering
\includegraphics[width=90mm]{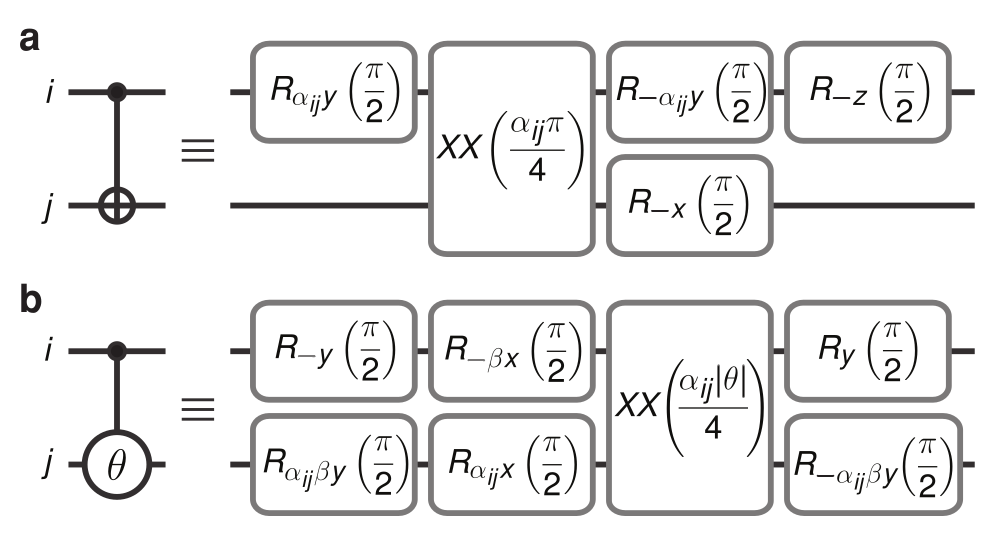}
\caption{\label{fig:i}(\cite{DLFL16}) Decomposition of two-qubit gates. Here (a) refers to the CNOT gate and (b) refers to the CP gate.
In order to achieve two-qubit operations that are less prone to errors,
Monroe et al. adopt decompositions into multiple single-qubit gates prior and after the two-qubit (XX)-gates based on 
M{\o}lmer-S{\o}rensen interactions \cite{MS99}.}
\end{figure}
In the usual setup, all ions are confined inside a linear rf-Paul trap and cooled near their motional ground state using Doppler cooling.
As a result, the chain of ions freezes into a linear crystal, with equal spacing of around $5 \mu$m. 
Next, the use of lasers and optical pumping achieves efficient state initialization, as described in \cite{OYM07}. 
All quantum gate implementations follow coherent rotations using 
Raman transitions to drive both, atomic transitions, as well as vibrational transitions,
in which lasers are focused on all of the ions in the chain simulatenously. Thus, in order to address ions individually, each Raman
beam is split into a static array of beams processed at the AOM, which then focuses the beams onto the individual ions. 
Measurement is the result of driving transitions near $369$nm of wavelength and simulatenously collecting state dependent fluorescence from each ion.
In practice, this is done by using a multi-channel photo-multiplier tube (PMT). 
Thus, if the qubit is in the state $\ket{1}$, the laser is on resonance and state-dependent fluorescence can be collected.
Else, if the qubit is in the state $\ket{0}$, the laser is sufficiently detuned and a dark state is observed.
\begin{figure}[tbp]
\centering
\includegraphics[width=120mm]{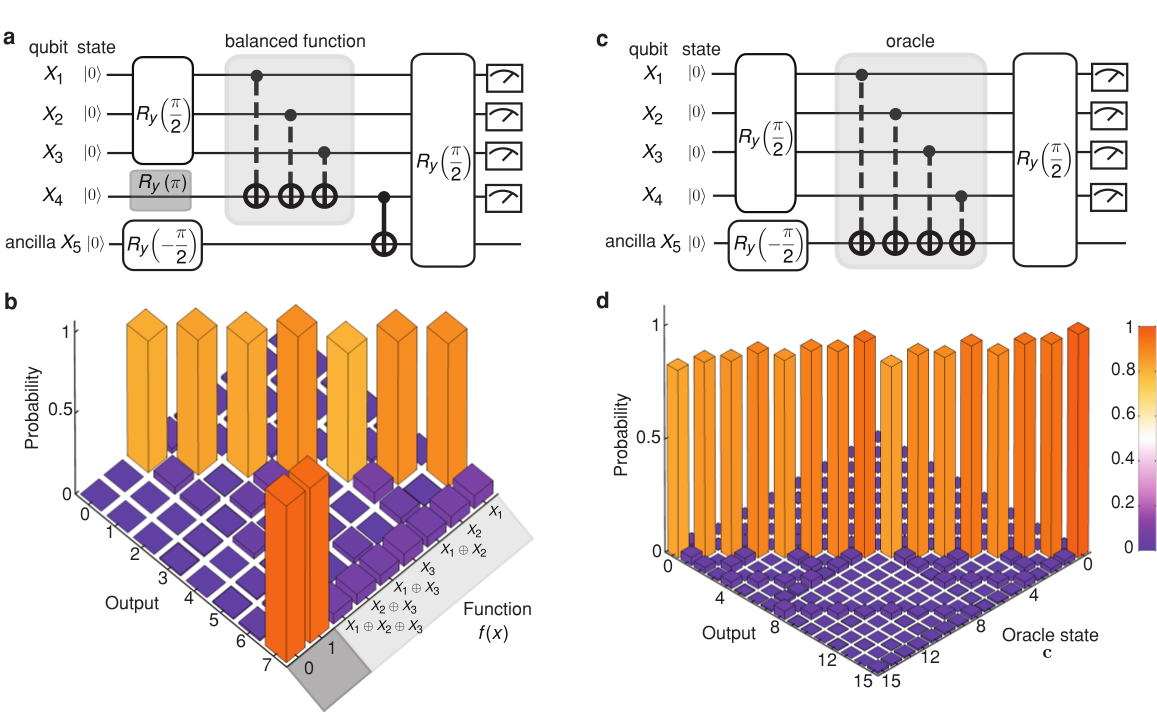}
\caption{\label{fig:i}(\cite{DLFL16}) A five-qubit implementation of the Deutsch-Josza (a) and Bernstein-Vazirani algorithm (b).}
\end{figure}
At the bottom of the hierarchy, qubit operations are performed via pairs of Raman beams from a single $355$nm YAG mode-locked laser. 
Here, the single qubit rotations R$_\varphi(\theta)$ are performed by a Raman beat-note of defined amplitude, phase and duration at 
the qubit resonance frequency $\nu_0 = 12.642821$GHz. 
As introduced in \expref{Section}{single_qubit_chapter}, $\theta$ describes the rotation angle and $\phi$ is determined by the duration and phase-offset of the 
beat-note and is programmed at the appropriate AOM channels. 
The two-qubit gates are performed using nearest-neighbour M{\o}lmer-S{\o}rensen interactions \cite{MS99} 
(XX-gates), a more sophisticated interaction for multi-particle entanglement which produces effective conditional dynamics
in order to realize two-qubit gates such as CNOT. Moreover, the Raman beat-notes are tuned close to 
resonance $\nu_0$, yet slightly detuned down to $\nu_0 \pm \nu_x$ by an offset $\nu_x$, in order to induce the necessary coupling.
In addition to our previous discussion on single-qubit and two-qubit gates, Debnath et al. used a modern variant of MS-interactions by decomposing a CNOT gate
into geometric phase gates, an approach that preserves the action of the CNOT, yet allows for an efficient and less error-prone realization by using
the collective motion of the chain \cite{MS99}\cite{HRB08}.

\begin{figure}
\centering
\includegraphics[width=90mm]{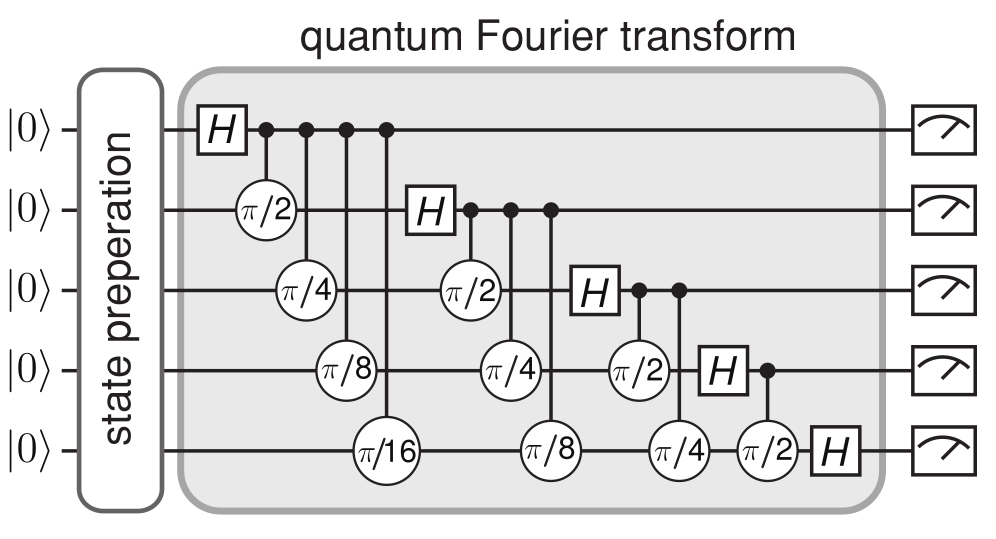}
\caption{\label{fig:i}(\cite{DLFL16}) A five-qubit implementation of the quantum Fourier transform.}
\end{figure}
In an attempt to explore novel architectures towards more scalable ion traps, Fellek et al. \cite{FHM16} at the Georgia Institute of Technology have independently implemented the 
Bernstein-Vazirani algorithm with a chain of $^{171}$Yb$^+$ qubits
using ion-transport in a microfabricated planar surface trap. 
Similarly, the algortihm succeeded at determining the unknown string with a success probability of $97.6\%$, 
for two ions, and $80.9\%$ in the case of three ions, using only a single oracle query. The gate implementations and optics are similar to the Maryland setup.
Single qubit gates are performed using a laser at $355$nm wavelength driving Raman transitions. Two-qubit gates are also provided
by nearest-neighbour M{\o}lmer-S{\o}rensen interactions \cite{MS99}.\par

\subsubsection{Decoherence and Sources of Error.}

In this section, we discuss decoherence mechanisms and error sources that drive imperfections in trapped ion quantum computation.
Just like in classical computation, a bit flip error $\ket{g} \leftrightarrow \ket{e}$ in a quantum state is devastating. 
As in most quantum devices, 
these errors  typically occur during population inversion of the energy eigenstates of the system due to photon absorption or spontaneous emission
and propagate through all subsequent computations. 
Thus, in any qubit manipulation requiring the use of lasers, there is some probability of driving an unwanted transition to other electronic levels.

In practice, alkali-earth-like metals, such as $^{171}$Yb$^+$, exhibit lifetimes of 
metastable states at about $2-3$s~\cite{OYM07}, hence
the coherence time of hyperfine states is several orders of magnitude longer than the gate times ranging at microseconds.
This fact makes ion-trap quantum computers surprisingly resistant to memory errors. On the contrary, external charge fluctuations in 
superconducting devices suffer considerably from bit-flip errors.

In terms of operational errors, both the original Cirac Zoller proposal and the XX-gates using nearest-neighbour M{\o}lmer-S{\o}rensen 
interactions are highly affected by population changes due to \textit{motional heating}. Electromagnetic background radiation at the trap frequencies 
can create motional quanta that subsequently corrupt two-qubit operations.
For example, due to the collective motional degree of freedom of the ion chain and the need for strict ground-state cooling,
any two-qubit operation in the Cirac Zoller proposal is significantly prone to highly correlated errors as a result of spontaneous emission or 
electromagnetic radiation. This suggests that independent noise models, such as in \expref{Section}{sec:quantum_noise}, are non-physical in the context
of an ion-trap architecture. Due to the longevity of hyperfine states and the nature of practical noise effects in ion-traps, 
it appears more reasonable to work with channels that capture dephasing.

Phase flip errors are more subtle and have important fatal consequences
in most quantum computations, as demonstrated in the well known \textit{Ramsey experiment}.
The phase evolution of a hyperfine state depends strongly on the magnetic field. A superposition of two state evolves due to
individual magnetic moments and thus experiences dephasing due to energy fluctuations resulting from a fluctuating magnetic field.
Typically, phase flips occur if the rf-Paul trap exhibits voltage fluctuations at the trap electrodes.
In fact, the coherence time of ion-trap quantum computers is currently mostly limited by magnetic field fluctuations in the order of just a few milliseconds
\cite{HRB08}.
Due to the fact that $^{171}$Yb$^+$ produces qubits with the same magnetic moment $m_F=0$, the ytterbium ion is a popular choice to 
significantly reduce dephasing effects.

In order to address errors during large scale computations in the future, one strives to adopt error correction, 
a procedure we discussed in \expref{Section}{sec:error_corr}.
However, as one increases the number of qubits in a linear ion-trap architecture, i.e. the size of the ion chain, the addressing of
individual ions with focused lasers onto the chain becomes increasingly difficult and complicates two-qubit operations with additional practical sources
of error. Moreover, growing mass of the ion chain also results in reduced coupling on the sideband transitions through the Lamb-Dicke parameter.

In the next section, we shed light on the overall performance of the ion-trap architecture, as compared to a solid-state device.

\begin{figure}[ht!]
\centering
\includegraphics[width=130mm]{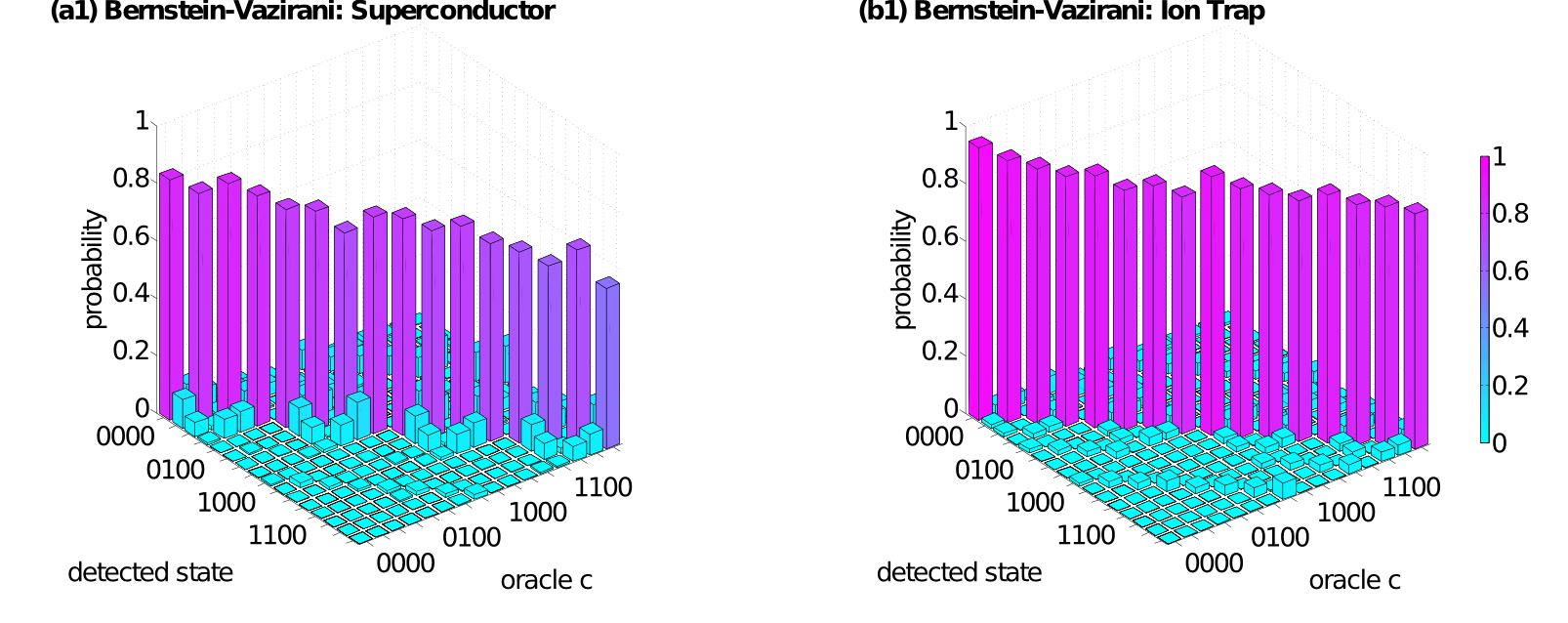}
\caption{\label{fig:i}(\cite{LMRD17}) Performance comparison between an ion-trap quantum computer
compared to a solid-state architecture when running the Bernstein-Vazirani algorithm.}
\label{fig:comparison}
\end{figure}

\begin{table}[ht!]
\center
\begin{tabular}{@{} *5l @{}}    \toprule
Result \quad\quad & \quad\quad Ion-Trap \quad\quad & \quad Superconductor \quad\quad  & \quad\quad Source \\\midrule
\textbf{Fidelity (\%)} \quad\quad & \quad\quad   \quad\quad & \quad\quad   \quad\quad  & \quad\quad \\ 
\,-Single-qubit gates \quad\quad & \quad\quad 99.1\% \quad\quad & \quad\quad 99.7\% \quad\quad  & \quad\quad \cite{LMRD17}\cite{FHM16}\\ 
\,-Two-qubit gates \quad\quad & \quad\quad 97.0\% \quad\quad & \quad\quad 96.5\% \quad\quad  & \quad\quad \cite{LMRD17}\cite{FHM16}\\ 
\,-Readout \quad\quad & \quad\quad 95.7\% \quad\quad & \quad\quad $\sim$80\% \quad\quad  & \quad\quad \cite{LMRD17}\\\midrule
\textbf{Gate times ($\tau_{op}$)} \quad\quad & \quad\quad   \quad\quad & \quad\quad   \quad\quad  & \quad\quad\\ 
\,-Single-qubit \quad\quad & \quad\quad 20$\mu$s \quad\quad & \quad\quad 130ns \quad\quad  & \quad\quad \cite{LMRD17}\\ 
\,-Two-qubit \quad\quad & \quad\quad 250$\mu$s \quad\quad & \quad\quad 250-450ns \quad\quad  & \quad\quad \cite{LMRD17}\\\midrule
\textbf{Decoherence times ($T_1,T_2$)} \quad\quad & \quad\quad   \quad\quad & \quad\quad   \quad\quad  & \quad\quad \\ 
\,-Depolarization \quad\quad & \quad\quad $\infty$ \quad\quad & \quad\quad $\sim$ 60$\mu$s \quad\quad  & \quad\quad \cite{LMRD17}\\ 
\,-Dephasing \quad\quad & \quad\quad $\sim$ 0.5s \quad\quad & \quad\quad $\sim$ 60$\mu$s \quad\quad  & \quad\quad \cite{LMRD17}\\\midrule
\textbf{Algorithms (success \%)} \quad\quad & \quad\quad   \quad\quad & \quad\quad   \quad\quad  & \quad\quad \\ 
\,-Bernstein-Vazirani \quad\quad & \quad\quad 85.1\%\quad\quad & \quad\quad $\sim$ 72.8\% \quad\quad  & \quad\quad \cite{LMRD17}\cite{DLFL16}\\ 
\,-Quantum Fourier Transform \quad\quad & \quad\quad 61.9\% \quad\quad & \quad\quad 61-75\% \quad\quad  & \quad\quad \cite{LMRD17}\cite{GMS17}\\
\bottomrule
 \hline
\end{tabular}
\caption{\label{tab:i} Summary of experimental comparisons between five-qubit ion-trap quantum computers and the five-qubit IBM superconducter device.
The results regarding \textit{readout}, the \textit{Bernstein-Vazirani algorithm} and the \textit{quantum Fourier transform} all refer to
five-qubit experiments. 
For further details, we refer to
the respective sources.}
\label{tab:results}
\end{table}

\newpage

\subsection{Performance Comparison of Quantum Computing Architectures.}

In the previous sections, we described how new advances in quantum computing technology made it possible to program algorithms
from a high level user interface. To this date, particularly the ion-trap and solid-state architectures have reasonably grown in maturity,
allowing for a variety of standard algorithms to be tested and evaluated for performance.
Recently, IBM has launched a public-access demonstration of a five qubit superconducting device
that can be run via their \textit{Quantum Experience} cloud service.\footnote{The IBM quantum experience platform can be found at \url{https://www.research.ibm.com/ibm-q/}}
Building up on previous work at the University of Maryland \cite{DLFL16}, Linke et al.~\cite{LMRD17} from the Monroe group 
put forward a state-of-the-art comparison between
the two leading quantum technologies, its own local ion-trap implementation vs. IBM's superconducting transmons, using the \textit{Quantum Experience} platform.
It was shown that, overall, the ion-trap quantum computer currently achieves higher success probabilities over the solid-state implementation 
from the IBM platform. Average success probabilities for running the Bernstein-Vazirani algorithm ranged around $85.1\%$ for ion-traps, and
around $72.8\%$ for the superconducting device~\cite{LMRD17}, see \expref{Table}{tab:results} and \expref{Figure}{fig:comparison}. 
Particularly concerning gate times, noticable differences were observed.
Typical ion-trap gate times for single-qubit operations for an ion-trap computer averaged at around $20 \mu$s and $250\mu$s for two-qubit gates, while
superconducting circuits reached times of only $130$ns and around $250-450$ns, respectively.
Overall, while current solid-state devices feature vastly higher clock-speeds, the ion-trap currently shows substantially higher absolute fidelities and 
longer coherence times. Nevertheless, the runner-up technology of the solid-state architecture offers a substantial promise for scalability.
It remains to be seen which of the two architectures, if any, is going to establish itself as the leading scalable quantum computing technology of the future.

%% file: conclusion.tex
\newpage
\section{Conclusion}

In this thesis, we shed light on how quantum algorithms achieve promising speed-ups over classical algorithms
in the context of computational learning theory, even in the presence of noise.
As quantum computational supremacy has yet to be demonstrated for a well-defined problem using only a few noisy qubits,
the study of quantum learning remains a particularly relevant area of research. For further reading
on the current status of quantum computational supremacy, we refer to an article by Harrow and Montanaro \cite{HM17}.
For an overview of recent progress in quantum learning theory, we refer to the survey by Arunachalam and de Wolf \cite{AdW17}.

By investigating the limitations of quantum algorithms through the use of relabeling, we proposed suitable constructions for new notions of security under
non-adaptive quantum chosen-ciphertext attacks. 
The pursuit of useful quantum-secure classical encryption schemes remains one of the key challenges in post-quantum cryptography.
Therefore, further research is now needed to investigate whether classical communication is ultimately feasibile in a quantum world.
Finally, for further reading on the current status of post-quantum cryptography, we refer to a recent article by Bernstein and Lange \cite{BL17}.

%% file: open_problems.tex
\section{Open Problems}

Due to the fact that \LWE is easy once algorithms receive quantum superposition access to noisy samples
on the secret string, it is tempting to explore circumstances in cryptography under which such access is granted.
While it is reasonable to assume that providing such quantum access is ill-advised for public-key cryptography,
it remains an open problem whether there are other realistic scenarios.
One possible direction to investigate is \textit{program obfuscation}, a recent breakthrough in cryptography
which concerns the process of obscuring software or code in order to preserve functionality, yet hide sensitive information
on the program itself with at most polynomial slowdown. An attacker in possession of a quantum computer could, in principle, 
implement the obfuscated circuit and then query it in superposition. Since \textit{indistinguishability obfuscation} allows
us to turn symmetric-key encryption schemes into public-key encryption schemes \cite{SW14}, this could potentially open up a door
to providing quantum access to \LWE samples in the context of the \LWE-\SKES.

Many of the classical notions of security are still widely unexplored in a quantum world, making both
quantum cryptography and classical quantum-safe cryptography a highly
relevant field of research. 
In 2016, Gagliardoni, H{\"u}lsing and Schaffner \cite{GHS16} provided security standards of quantum
indistuingishability under quantum chosen-plaintext-attacks and proposed secure quantum encryption schemes
for which such security notions can be achieved.
Further research is now needed to investigate secure constructions under a 
quantum chosen-ciphertext attack, in particular regarding a quantum indistinguishability notion of \QINDQCCA.
This indistinguishability setting naturally generalizes the \QCCA learning phase we considered in this thesis to
a now fully quantum challenge phase in which the challenge ciphertext is also a quantum state.
Thus, in the \QINDQCCA indistinguishability game, both the learning phase and the challenge phase concern fully quantum communication.
While many classical constructions for \CCAA security already exist, it also remains
an open problem to define a fully quantum notion of indistinguishability for \QCCAA that allows for secure quantum encryption schemes.

%% file: appendix.tex
\section{Supplementary Material}

\begin{lemma}[\cite{Hoe63}, Hoeffding's inequality]\label{lem:hoeffding}
Let $\{X_m\}$ be a set of $M$ independent random variables, such that $a_m \leq X_m \leq b_m$,
and let $c_m = b_m - a_m$, $\bar X= \frac{1}{M}\sum_{m=1}^M X_m$ and $\mu = \mathbb{E}[\bar X]$. Then:
\begin{equation}
\Pr\left[ |\bar X - \mu| \geq \delta \right] \, \leq \, 2 e^{-2 M^2 \delta^2/\sum_m c_m^2}.
\end{equation}
\end{lemma}

\begin{lemma}[\cite{Wil13}]\label{lem:trace_bound}
Let $\rho$,$\, \sigma \in \mathcal{D}(\mathcal{H})$ be quantum states and consider any measurement due to POVM $\, \mathcal{E} = \{E_i \}_{i \in I}$. We define $p_i = tr[\rho E_i]$ and $q_i = tr[\sigma E_i]$ 
as the corresponding probability distributions over measurement outcomes labeled by $i \in I$.
Then, the statistical distance between the distributions $p_i$ and $q_i$ is upper bounded by the trace distance between $\rho$ and $\sigma$:
\begin{equation}
\label{trace_bound}
\delta (p_i, q_i) \leq \delta(\rho,\sigma).
\end{equation}
\end{lemma}

\begin{lemma}[\cite{NC10}]\label{lem:trace_fidelity}
Let $\rho = \ket{\psi}\bra{\psi}$ and $\sigma = \ket{\phi}\bra{\phi}$ be pure states. 
Then the trace distance between $\rho$ and $\sigma$ can be expressed in terms of the fidelity:
\begin{equation}
\label{pure_fidelity}
\delta(\rho,\sigma) = \sqrt{ 1 - F(\rho,\sigma)^2   } = \sqrt{ 1 - |\braket{\psi | \phi}|^2}.
\end{equation}
\end{lemma}

\begin{lemma}[\cite{KL15}, Markov's inequality]\label{lem:markov}
Let $X$ be a nonnegative random variable. Then, for any $a>0$:
\begin{equation}
\Pr\left[ X \geq a \right] \, \leq \, \displaystyle\frac{\mathbb{E}[X]}{a} \,
\end{equation}
\end{lemma}